\documentclass[reqno,11pt]{amsart}
\usepackage{scalerel,stackengine}
\usepackage{graphicx}
\usepackage{slashed}
\usepackage{amscd}
\usepackage{amssymb}
\usepackage{esint}
\usepackage{enumitem}
\usepackage{amsmath}
\usepackage{dsfont}
\usepackage[mathscr]{eucal}
\ifx\makepics\undefined
\usepackage[linktocpage]{hyperref}
\usepackage[off]{auto-pst-pdf} 
\else
\usepackage[pspdf=-dALLOWPSTRANSPARENCY]{auto-pst-pdf} 
\fi
\usepackage{epsfig}
\usepackage{pstricks}
\usepackage{pst-all}

\numberwithin{equation}{section}
\allowdisplaybreaks[1]
\definecolor{labelkey}{gray}{.65}

\title[Derivation of the Einstein Equations]{A Geometric Derivation of the Einstein Equations from the Causal Action Principle}

\author[F.\ Finster]{Felix Finster}
\author[C.\ Krpoun]{Christoph Krpoun \\ \\ July / November 2026} 
\address{Fakult\"at f\"ur Mathematik \\ Universit\"at Regensburg \\ D-93040 Regensburg \\ Germany}
\email{finster@ur.de, christoph.krpoun@ur.de}

\newtheorem{Def}{Definition}[section]
\newtheorem{Thm}[Def]{Theorem}
\newtheorem{Prp}[Def]{Proposition}
\newtheorem{Lemma}[Def]{Lemma}
\newtheorem{Remark}[Def]{Remark}
\newtheorem{Corollary}[Def]{Corollary}

\newcommand{\Thanks}{\vspace*{.5em} \noindent \thanks}
\newcommand{\beq}{\begin{equation}}
\newcommand{\eeq}{\end{equation}}
\newcommand{\Proof}{\begin{proof}}
	\newcommand{\QED}{\end{proof} \noindent}
\newcommand{\QEDrem}{\ \hfill $\Diamond$}

\newcommand{\la}{\langle}
\newcommand{\ra}{\rangle}

\newcommand{\C}{\mathbb{C}}
\newcommand{\R}{\mathbb{R}}
\newcommand{\1}{\mbox{\rm 1 \hspace{-1.05 em} 1}}

\newcommand{\N}{\mathbb{N}}

\DeclareMathOperator{\tr}{tr}
\renewcommand{\O}{{\mathscr{O}}}
\renewcommand{\L}{{\mathcal{L}}}
\newcommand{\Sact}{{\mathcal{S}}}

\newcommand{\U}{\text{\rm{U}}}

\renewcommand{\H}{\mathscr{H}}

\newcommand{\Lin}{\text{\rm{L}}}
\newcommand{\F}{{\mathscr{F}}}

\DeclareMathOperator{\supp}{supp}

\newcommand{\scrU}{{\mathscr{U}}}

\newcommand{\s}{\mathfrak{s}}

\renewcommand{\v}{\mathfrak{v}}

\newcommand{\bitem}{\begin{itemize}[leftmargin=2.5em]}
\newcommand{\eitem}{\end{itemize}}

\newcommand{\G}{{\mathscr{G}}}

\newcommand{\nind}{{\,\centerdot\,}}

\newcommand{\0}{\mathbf{0}}
\newcommand{\Comm}{{\mathscr{C}}}

\newcommand{\dif}[1]{d#1}

\DeclareFontFamily{OT1}{rsfso}{}
\DeclareFontShape{OT1}{rsfso}{m}{n}{ <-7> rsfso5 <7-10> rsfso7 <10-> rsfso10}{}
\DeclareMathAlphabet{\mycal}{OT1}{rsfso}{m}{n}

\newcommand\Christoph[1]{}
\newcommand\Felix[1]{}

\begin{document}
\maketitle

\begin{abstract}
The causal action principle for causal fermion systems is analyzed for a minimizing measure whose support is assumed to have the structure of a smooth manifold $\tilde{M}$. The concept of osculating vacua is introduced. It is shown that the Lagrangian induces on $\tilde{M}$ a Lorentzian metric. Moreover, the Euler-Lagrange equations of the causal action imply that the Ricci tensor must satisfy the Einstein equations of General Relativity for an energy-momentum tensor given in terms of a power expansion in the regularization length. The gravitational coupling constant is found to be the square of the regularization length. Our methods provide a systematic procedure for deriving corrections to the Einstein equations.

The paper includes a self-contained introduction to causal variational principles and the causal action principle. Most geometric structures (connection, Riemannian metric and curvature) are introduced and analyzed in the general setting of causal variational principles for an arbitrary dimension of $\tilde{M}$. The Lorentzian setting works only for causal fermion systems and is worked out only in four spacetime dimensions.
\end{abstract}

\tableofcontents

\section{Introduction} \label{secintro}
The theory of {\em{causal fermion systems}} is a recent approach to fundamental physics
(for an introduction to the physical background and applications as well as to
the mathematical context, we refer the interested reader to the
review~\cite{review}, the textbooks~\cite{cfs, intro} or the 
website~\cite{cfsweblink}).
In this approach, spacetime and all structures therein are encoded in a
measure~$\rho$ on a set~$\F$ of linear operators on a Hilbert space~$\H$.
Spacetime~$M := \supp \rho$ is then defined as the support of this measure.
The fact that the spacetime points are linear operators on~$\H$ 
yields relations between the spacetime points and gives rise to
all structures in spacetime.
In this approach, the physical equations are formulated via a novel variational principle for the measure~$\rho$, the {\em{causal action principle}}.
{\em{Causal variational principles}} evolved as a mathematical generalization
of the causal action principle~\cite{continuum, jet, noncompact}
(an introduction to the causal action principle and causal variational principles
can be found for example in~\cite[Chapters~5 and~6]{intro}).
The setting of causal variational principles has the advantage that
it is easier and more general, making it clearer what the basic structures are.
In this setting, the set~$\F$ is a smooth manifold.
In the causal variational principle, given a {\em{Lagrangian}}~$\L : \F \times \F \rightarrow \R^+_0$, one minimizes the {\em{action}}
\[ \Sact(\rho) = \int_\F d\rho(x) \int_\F d\rho(y)\: \L(x,y) \]
under variations of the measure~$\rho$, keeping the total volume fixed
(for details see Section~\ref{seccvp}).
The drawback of working in the setting of
causal variational principles is that one has fewer structures
at one's disposal. In particular, we shall see that, in the setting of causal
variational principles, one can only formulate the Einstein equations in the
Riemannian setting, whereas the causal action principle for causal fermion
system also gives rise to a Lorentzian metric and the corresponding
Einstein equations in spacetime.

The Einstein equations of General Relativity describe how energy and momentum
determine the curvature of spacetime
(for the general background we refer to
the physics textbooks~\cite{wald, hawking+ellis, misner, straumann}
and the mathematical literature~\cite{oneillsemi, beem}).
The connection between the causal action principle and the Einstein equations
has already been made in various ways:
The first approach was made in~\cite{lqg} by introducing
the structures of a general ``quantum geometry''. However, at that time
it was not possible to relate these geometric structures to the Euler-Lagrange
equations of the causal action principle.
The first derivation of the Einstein equations from the causal action principle
was given in~\cite[Chapter~4]{cfs} using the continuum limit analysis
for linearized gravity.
Combined with the fact that the setup of causal fermion systems is diffeomorphism
invariant, this linearized analysis gives the Einstein equations, up to higher
order corrections in the curvature tensor.
More recently, the continuum limit analysis was extended to the non-perturbative
setting in a globally hyperbolic spacetime in~\cite{cfs-curved}.
Apart from giving a direct derivation of the Einstein equations, it was also
tried to work more indirectly by deriving effects of the Einstein equations
from the causal action principle. This approach led to a derivation of the
Einstein equations following Jacobson's idea via the connection between
matter flux and area change~\cite{jacobson}.
Moreover, the effects of gravity like positive mass and quasi-local mass
were studied in the setting of causal variational principles in~\cite{pmt, matter}.

In the present paper we proceed differently,
motivated by the aim to formulate the causal action principle in geometric terms
as a geometric variational principle.
Thus our main concern is how to get a direct link between the geometry
and the causal action principle. This entails introducing the geometric objects
in such a way that they fit together with the analytic structures of causal fermion systems.
We now explain a few considerations which turn out to be helpful in order to
achieve this goal. We denote the measure describing the interacting spacetime
by~$\tilde{\rho}$; the corresponding spacetime~$\tilde{M}$ is then given by
\[ \tilde{M} := \supp \tilde{\rho} \subset \F \:. \]
We assume throughout this paper that~$\tilde{M}$ has a {\em{smooth manifold
structure}} (for non-smooth spacetimes structures see the outlook in Section~\ref{secoutlook} and the more general constructions in~\cite{lcalc}).
It is useful to regard~$\tilde{M} \subset \F$ as an {\em{embedded}} manifold.
Next, the concept of {\em{osculating vacua}} is very helpful.
This concept is motivated by the notion of the tangent space in
differential geometry. In simple terms, the idea is to ``attach'' to each
spacetime point~$p \in \tilde{M}$ a vacuum spacetime~$M_p \subset \F$
which, similar to the tangent space of an embedded manifold, ``approximates
the geometry of~$\tilde{M}$ near~$p$.'' This intuitive notion will be made precise
in Section~\ref{secosculate} by setting up a variational principle (see~\eqref{Spdef}).
This concept is illustrated in Figure~\ref{figosculate} on page~\pageref{figosculate}
(or, in its generalization to discrete spacetimes, in Figure~\ref{figosculate-discrete}
on page~\pageref{figosculate-discrete}).
Working with osculating vacua, one can use integrals over the Lagrangian
to introduce distinguished charts, a connection~$\nabla^\L$ and a Riemannian metric~$g$ (see Sections~\ref{secchart}--\ref{secg}). The resulting so-called
$\L$-geometry is introduced and studied in Section~\ref{secL}.

The core of the paper is to work out how these geometric structures
can be connected to the Euler-Lagrange equations of the causal action principle.
To this end, we work with alignment vector fields and expansions
of divergences of this vector field (Section~\ref{secEL}).
These methods are related to similar constructions developed in~\cite{matter}
in the context of defining a positive quasi-local mass in static spacetimes.
Combining our methods and results, we obtain the Einstein equations
in the Riemannian setting (Theorem~\ref{thmrein}).

In order to extend our results to Lorentzian signature,
we need to move to the more specialized setting of causal fermion systems.
This gives rise to a vector field~$u$ dubbed {\em{regularization vector field}},
making it possible to introduce a Lorentzian metric~$\eta$
(see Definition~\ref{defflip}). This vector field is shown to be
almost parallel (see Lemma~\ref{lemmau}), making it possible to reformulate
the Einstein equations for the Lorentzian metric (Theorem~\ref{thmlein}).

We mention one technical point which is important for the
general understanding of our method. We assume that the Lagrangian
is of {\em{short range}} in the sense that~$\L(x,y)$ vanishes if the
distance of~$x$ and~$y$ is larger than a given length scale~$\delta$ (for details see~\eqref{shortrange}). The length scale~$\delta$ should be thought of as being
much smaller than the typical length scales of microscopic physics
(more specifically, one can identify~$\delta$ with the Planck scale).
This justifies that we perform expansions in powers of~$\delta$
(see for example Lemmas~\ref{lemmalocexpand}, \ref{lemmaf},
\ref{lemmaxdiv} and~\ref{lemmapdiv}). It is a main result of our analysis
that, using the EL equations of the causal action, the Einstein tensor
equals a tensor which is of the order~$\O(\delta^2)$ and which we define
to be the energy-momentum tensor.
In this way, our methods and results give an explanation to
why the gravitational coupling constant is so small.
Moreover, our methods have the benefit that they
provide a systematic procedure for deriving correction terms
(as is discussed in Section~\ref{secoutlook}).

The paper is organized as follows. Section~\ref{secprelim} provides the
necessary background on causal variational principles and the causal
action principle. In Section~\ref{secosculate} the concept of osculations
is introduced. In Section~\ref{secL} the $\L$-geometry is developed.
In Section~\ref{secEL} it is explained how the EL equations can be used
to compute and expand the Ricci tensor.
In Section~\ref{seceinstein} the previous methods and results are
combined in order to derive the Einstein equations, first in the
Riemannian and then in the Lorentzian setting.
In Section~\ref{seccorrect} the different corrections are discussed.
Section~\ref{secoutlook} gives an outlook and non-smooth and
quantum spacetimes.
The appendices provide supplementary material: 
In Appendix~\ref{apph} a Weingarten-type map is introduced, and we
argue why it does not give rise to notions of extrinsic curvature and
corresponding Gau{\ss}-Codazzi equations.
In Appendix~\ref{appregvac} it is explained how the Riemannian and Lorentzian metrics introduced in
Sections~\ref{secg} and~\ref{secregvec} can be computed in the Minkowski
vacuum. Appendix~\ref{appoptimal} gives a construction of osculations which
are almost optimal in a quantified sense.
In Appendix~\ref{appnonoptimal} it is explained how the definition
of the connection~$\nabla^\L$ generalizes to non-optimal osculations,
and it is shown that this connection in general has torsion.

\section{Preliminaries} \label{secprelim}
Since most of our constructions work in the general setting of causal
variational principles, we introduce this setting first (Section~\ref{seccvp}).
Working in this setting will give us the Einstein equations in Riemannian signature
(Section~\ref{secriemann}). In order to get the Einstein equations in
the physical Lorentzian signature (Section~\ref{seclorentz}),
one needs additional structures specific to causal fermion systems.
For this reason, we also introduce causal fermion systems and the
causal action principle (Section~\ref{seccfs}).
Our presentation is intended to be brief, but self-contained.
For more details and the general background we refer to the
textbooks~\cite{intro, cfs}.

\subsection{Causal Variational Principles} \label{seccvp}
We begin with the setup of smooth causal variational principles in the
non-compact setting (see for example~\cite[Chapter~6]{intro}). Thus we let~$\F$ be a smooth
(possibly non-compact) manifold. The {\em{Lagrangian}}~$\L$ is a given smooth function
\[ 
\L \in C^\infty(\F \times \F, \R^+_0) \:. \]
Moreover, we assume that~$\L$ has the following properties:
\bitem
\item[{\rm{(i)}}] $\L$ is {\em{symmetric}}: $\L(x,y) = \L(y,x)$ for all~$x,y \in \F$.\label{Cond1}
\item[{\rm{(ii)}}] $\L$ is {\em{strictly positive on the diagonal}}: $\L(x,x)>0$ for all~$x \in \F$. \label{Cond2}
\eitem
Finally, we need to assume that the Lagrangian decays sufficiently fast
if its arguments~$x$ and~$y$ are far apart. One way of doing so is to
use the following notion first introduced in~\cite[Definition~3.3]{noncompact}
(see also~\cite[Definition~8.1.1]{intro}):
\bitem
\item[{\rm{(iii)}}] $\L$ has {\em{compact range}}: For every compact set~$K \subset \F$ there is a compact set~$K' \subset \F$ such that
\[ \L(x,y) = 0 \qquad \text{for all~$x \in K$ and~$y \not \in K'$}\:. \]
\eitem
This condition could be relaxed by demanding that~$\L$ and all its derivatives
have rapid decay. We will implicitly use this weaker assumption in some of
the examples.

The {\em{causal variational principle}} is to minimize the causal action defined by
\beq \label{Sactcompact}
\Sact(\rho) = \int_\F d\rho(x) \int_\F d\rho(y)\: \L(x,y)
\eeq
under variations of the measure~$\rho$ in the class of all regular Borel measures on~$\F$, under the constraint that the total volume is kept fixed.

The existence of minimizers has been established in~\cite{noncompact}
(see also~\cite[Chapter~12]{intro}). It is also shown that minimizing measures
are {\em{locally finite}} in the sense that~$\rho(K)<\infty$ for any compact~$K \subset \F$.
A minimizing measure satisfies the {\em{Euler-Lagrange (EL) equations}}, which state that the function~$\ell$ defined by
\beq \label{elldef}
\ell(x) = \int_\F \L(x,y)\: d\rho(y) - \s \::\: \F \rightarrow \R
\eeq
satisfies for a suitable parameter~$\s>0$ the equation
\beq \label{EL1}
\ell|_M \equiv \inf_\F \ell = 0 \:.
\eeq
The derivation can be found in~\cite[Section~4]{noncompact} or~\cite[Chapter~7]{intro}.

We remark that causal variational principles can be regarded as a
generalization of the causal action principle, being at the heart of the physical
theory of causal fermion systems. For the purposes of the present paper, we do not need
to enter the details of the connection to physics or to spacetime geometry.
We refer the reader interested in the physical background to the recent text book~\cite{intro};
in particular, the connection between the causal action principle and causal variational
principles is explained in detail in~\cite[Chapter~6]{intro}).
For what follows, it suffices to note that the support of the measure~$\rho$ denoted by
\[ M := \supp \rho \]
is considered as the underlying {\em{space}} or {\em{spacetime}} (we will use the terms "space" or "spacetime" interchangeably). It is by definition
a closed subset of the manifold~$\F$.

The above assumption of compact range is suitable for proving existence and
studying general properties of minimizers, but it is not strong enough for a
the more quantitative analysis to be performed here. In particular, we need
that the range of the Lagrangian is much smaller than the macroscopic length scales
of the system. To this end, we need that the
Lagrangian is of {\em{short range}} in the following sense.
We let~$d \in C^0(M \times M, \R^+_0)$ be a distance function on~$M$. The assumption
of short range means that~$\L$ vanishes on distances larger than~$\delta$, i.e.
\beq \label{shortrange}
d(x,y) > \delta \quad \Longrightarrow \quad \L(x,y) = 0 \:.
\eeq
Here~$\delta>0$ is a parameter which determines the range of the potential.
In our analysis, we will assume that~$\delta$ is very small, and we shall perform
expansions in powers of~$\delta$ (in a sense to be specified in the paragraph
before Lemma~\ref{lemmametric}).
As we shall see, the parameter~$\delta^2$ will play the role of the gravitational coupling constant
(see the explanations after the proof of Theorem~\ref{thmlein} at the end of 
Section~\ref{seclorentz}).

In the above general setup, we did not need to specify the manifold~$\F$.
Typically, this manifold is formed of operators acting on a Hilbert space
of functions or sections on a given base space. This connection and simple
examples are given in~\cite{topology}.
Here we move on to the setting of causal fermion systems.

\subsection{Causal Fermion Systems and the Causal Action Principle} \label{seccfs}
We begin with the general definitions.
\begin{Def} \label{defcfs} (causal fermion systems) {\em{ 
Given a separable complex Hilbert space~$\H$ with scalar product~$\la .|. \ra_\H$
and a parameter~$n \in \N$ (the {\em{``spin dimension''}}), we let~$\F \subset \Lin(\H)$ be the set of all
symmetric operators on~$\H$ of finite rank, which (counting multiplicities) have
at most~$n$ positive and at most~$n$ negative eigenvalues. On~$\F$ we are given
a positive measure~$\rho$ (defined on a $\sigma$-algebra of subsets of~$\F$).
We refer to~$(\H, \F, \rho)$ as a {\em{causal fermion system}}.
}}
\end{Def} \noindent

The dynamical equations are formulated via an action principle,
which we now introduce. For brevity of the presentation, we only consider the
{\em{reduced causal action principle}} where the so-called boundedness constraint has been built
incorporated by a Lagrange multiplier term. This simplification is no loss of generality, because
the resulting EL equations are the same as for the non-reduced action principle
as introduced for example in~\cite[Section~\S1.1.1]{cfs}.

For any~$x, y \in \F$, the product~$x y$ is an operator of rank at most~$2n$. 
However, in general it is no longer a symmetric operator because~$(xy)^* = yx$,
and this is different from~$xy$ unless~$x$ and~$y$ commute.
As a consequence, the eigenvalues of the operator~$xy$ are in general complex.
We denote the rank of~$xy$ by~$k \leq 2n$. Counting algebraic multiplicities, we choose~$\lambda^{xy}_1, \ldots, \lambda^{xy}_{k} \in \C$ as all the nonzero eigenvalues and set~$\lambda^{xy}_{k+1}, \ldots, \lambda^{xy}_{2n}=0$.
We refer to the resulting collection of complex numbers~$\lambda^{xy}_1, \ldots, \lambda^{xy}_{2n}$
as the {\em{non-trivial eigenvalues}} of~$xy$.
Given a parameter~$\kappa>0$ (which will be kept fixed throughout this paper),
we introduce the $\kappa$-Lagrangian and the causal action by
\begin{align}
\text{\em{$\kappa$-Lagrangian:}} && \L(x,y) &= 
\frac{1}{4n} \sum_{i,j=1}^{2n} \Big( \big|\lambda^{xy}_i \big|
- \big|\lambda^{xy}_j \big| \Big)^2 + \kappa\: \bigg( \sum_{j=1}^{2n} \big|\lambda^{xy}_j \big| \bigg)^2 \label{Lagrange} \\
\text{\em{causal action:}} && \Sact(\rho) &= \iint_{\F \times \F} \L(x,y)\: d\rho(x)\, d\rho(y) \:. \label{Sdef}
\end{align}
The {\em{reduced causal action principle}} is to minimize~$\Sact$ by varying the measure~$\rho$
under the following constraints,
\begin{align*}
\text{\em{volume constraint:}} && \rho(\F) = 1 \quad\;\; \\ 
\text{\em{trace constraint:}} && \int_\F \tr(x)\: d\rho(x) = 1 \:. 
\end{align*}

At first sight, the setting of the causal action principle seems considerably more
complicated than that of causal variational principles. But the settings
fit together as follows (for more details see~\cite[Section~6.5]{intro}).
First, let us assume that~$\H$ is {\em{finite-dimensional}} (this can be justified by
an ultraviolet cutoff and an exhaustion of Hilbert spaces; details can be found for example in~\cite[Section~1.2]{cfs}). Moreover, we assume that both~$\rho$
and~$\tilde{\rho}$ are {\em{regular}} in the sense that all spacetime point operators
have maximal rank (i.e., that they have exactly~$n$ positive and~$n$ negative eigenvalues).
The EL equations of the causal action principle yield that all the spacetime point
operators have the same trace.
This makes it possible to replace~$\F$ by the set of all symmetric
operators of fixed trace which have exactly~$n$ positive and~$n$ negative eigenvalues.
This set of operators has a manifold structure.
Finally, the measures~$\tilde{\rho}$ and~$\rho$ are minimizers of the
causal action principle on this restricted set of operators even if the trace
constraint is dropped. In this way, one gets back to the general setting of causal
variational principles, but with~$\F$ chosen as a specific set of linear operators
and the specific Lagrangian~\eqref{Lagrange}.

\section{Osculating Vacua} \label{secosculate}
\subsection{Formulation for Causal Variational Principles}
In what follows, we assume that we are given two minimizing measures: The measure~$\rho$
describing the vacuum, and a measure~$\tilde{\rho}$ describing the interacting measure.
We assume that the {\em{vacuum spacetime}}~$M:= \supp \rho$
has the structure of a $k$-dimensional real vector space (with~$k \in \N$ arbitrary;
in Section~\ref{secosccfs} we will specialize the setting to dimension~$k=4$).
This means in particular that
one point of~$M$ is distinguished as the origin; we denote it by~$\0 \in M$.
In what follows, we always identify~$M$ with its tangent space~$T_\0M$.
Next, we assume that the {\em{Lagrangian}} is {\em{translation invariant}}
on~$M$, i.e.\
\beq \label{translation}
\L(x,y) = L[y-x] \qquad \text{for all~$x,y \in M$}
\eeq
with a function~$L : M \rightarrow \R^+_0$. The fact that the Lagrangian is symmetric
in its two arguments implies that the function~$L$ is {\em{reflection
symmetric}}, i.e.\
\beq \label{reflection}
L[\xi] = L[-\xi] \qquad \text{for all~$\xi \in M$} \:.
\eeq
Moreover, in order to exclude trivial situations, we assume that~$L$ is not
everywhere zero.
Finally, we we assume that, using the above vector space structure,
the spacetime measure~$\rho|_M$ is a multiple of the Lebesgue measure.
For the applications, one can think of~$M$ as Minkowski space. However, in the constructions of this paper, we do not want to make use of the Minkowski metric. Instead, we only want to make use of the form of the measure~$\rho$ and the 
Lagrangian as specified above.

The structure of the {\em{interacting spacetime}}~$\tilde{M} \subset \F$, however, can be very general. The only assumption we use throughout this paper is that~$\tilde{M}$
has the structure of a $k$-dimensional {\em{smooth differential manifold}}.
This assumption makes it possible to use the usual notions of differential
geometry like the tangent space, local expansions in charts, and all that.
For non-smooth spaces without manifold structure we refer to~\cite{lcalc}
and Section~\ref{secoutlook}.

Following~\cite[Definition~6.1]{matter} we introduce symmetry transformations of the Lagrangian.
\begin{Def} \label{defiso}
A diffeomorphism~$\Phi \in C^\infty(\F, \F)$ describes a {\bf{symmetry of the Lagrangian}} if
\[ \L\big(\Phi(x), \Phi(y) \big) = \L(x,y) \qquad \text{for all~$x,y \in \F$}\:. \]
\end{Def} \noindent
Such diffeomorphisms form a group, denoted by~${\mathcal{G}}$, the {\em{group of symmetries of the Lagrangian}}.

For introducing our concepts, it is best to begin with the simplest setting
where we assume that the group of symmetries~$\G$ of the Lagrangian
{\em{acts transitively on~$\tilde{M}$}} in the sense that for every~$p \in \tilde{M}$
there is a symmetry transformation~$\Phi \in \G$ with~$\Phi(\0)=p$
(the case that~$\G$ does not act
transitively on~$\tilde{M}$ will be considered in Section~\ref{secosccfs}).
In this case, given a spacetime point~$p \in \tilde{M}$, we consider the set of all symmetry transformations which map~$\0 \in M$ to~$p$,
\[ 
\mathcal{G}_p := \{ \Phi_p \in \G \:|\: \Phi_p(\0) = p \} \:. \]
Given~$\Phi_p \in \mathcal{G}_p$, we set
\[ \rho_p := (\Phi_p)_* \rho \qquad \text{and} \qquad M_p := \supp \rho_p \:. \]
We refer to~$\Phi_p$ as the {\em{osculation}} and~$\rho_p$ as the {\em{osculating vacuum}} at~$p 
\in \tilde{M}$. These notions are illustrated in Figure~\ref{figosculate}.
\begin{figure}[tb]
\psset{xunit=.5pt,yunit=.5pt,runit=.5pt}
\begin{pspicture}(619.43726547,107.31803594)
{
\newrgbcolor{curcolor}{0 0 0}
\pscustom[linewidth=1.88976378,linecolor=curcolor]
{
\newpath
\moveto(0.69662362,27.79057389)
\lineto(358.14092598,71.63985515)
}
}
{
\newrgbcolor{curcolor}{0 0 0}
\pscustom[linewidth=3.02362209,linecolor=curcolor]
{
\newpath
\moveto(239.81681386,1.47974428)
\lineto(617.54029606,80.5463383)
}
}
{
\newrgbcolor{curcolor}{0 0 0}
\pscustom[linewidth=3.77952756,linecolor=curcolor]
{
\newpath
\moveto(9.87085228,6.6370457)
\curveto(9.87085228,6.6370457)(109.63859906,46.26609515)(172.71522142,47.9744605)
\curveto(241.80664441,49.84572728)(325.70925354,23.08691783)(429.50823307,40.45856003)
\curveto(501.21970394,52.46008696)(618.65824252,105.59630995)(618.65824252,105.59630995)
}
}
{
\newrgbcolor{curcolor}{0 0 0}
\pscustom[linestyle=none,fillstyle=solid,fillcolor=curcolor]
{
\newpath
\moveto(151.53502865,47.58587578)
\curveto(152.26522455,45.63038348)(151.25289474,43.50691068)(149.27392674,42.84296653)
\curveto(147.29495873,42.17902237)(145.09874745,43.22602901)(144.36855155,45.18152131)
\curveto(143.63835564,47.13701361)(144.65068545,49.26048641)(146.62965345,49.92443056)
\curveto(148.60862146,50.58837472)(150.80483274,49.54136808)(151.53502865,47.58587578)
\closepath
}
}
{
\newrgbcolor{curcolor}{0 0 0}
\pscustom[linewidth=3.40157403,linecolor=curcolor]
{
\newpath
\moveto(151.53502865,47.58587578)
\curveto(152.26522455,45.63038348)(151.25289474,43.50691068)(149.27392674,42.84296653)
\curveto(147.29495873,42.17902237)(145.09874745,43.22602901)(144.36855155,45.18152131)
\curveto(143.63835564,47.13701361)(144.65068545,49.26048641)(146.62965345,49.92443056)
\curveto(148.60862146,50.58837472)(150.80483274,49.54136808)(151.53502865,47.58587578)
\closepath
}
}
{
\newrgbcolor{curcolor}{0 0 0}
\pscustom[linewidth=3.02362209,linecolor=curcolor]
{
\newpath
\moveto(0.3078274,15.34911751)
\lineto(356.97456378,89.52445105)
}
}
{
\newrgbcolor{curcolor}{0 0 0}
\pscustom[linestyle=none,fillstyle=solid,fillcolor=curcolor]
{
\newpath
\moveto(402.30814229,36.69959924)
\curveto(403.0383382,34.74410694)(402.02600839,32.62063414)(400.04704038,31.95668998)
\curveto(398.06807238,31.29274583)(395.87186109,32.33975247)(395.14166519,34.29524477)
\curveto(394.41146929,36.25073706)(395.42379909,38.37420987)(397.4027671,39.03815402)
\curveto(399.3817351,39.70209817)(401.57794639,38.65509153)(402.30814229,36.69959924)
\closepath
}
}
{
\newrgbcolor{curcolor}{0 0 0}
\pscustom[linewidth=3.40157403,linecolor=curcolor]
{
\newpath
\moveto(402.30814229,36.69959924)
\curveto(403.0383382,34.74410694)(402.02600839,32.62063414)(400.04704038,31.95668998)
\curveto(398.06807238,31.29274583)(395.87186109,32.33975247)(395.14166519,34.29524477)
\curveto(394.41146929,36.25073706)(395.42379909,38.37420987)(397.4027671,39.03815402)
\curveto(399.3817351,39.70209817)(401.57794639,38.65509153)(402.30814229,36.69959924)
\closepath
}
}
{
\newrgbcolor{curcolor}{0 0 0}
\pscustom[linewidth=1.88976378,linecolor=curcolor]
{
\newpath
\moveto(240.8694085,17.7314709)
\lineto(618.14464252,59.28063751)
}
\rput[bl](0,80){$\F$}
\rput[bl](142,57){$p$}
\rput[bl](585,110){$\tilde{M}$}
\rput[bl](365,57){$T_p\tilde{M}$}
\rput[bl](365,80){$M_p$}
\rput[bl](395,10){$q$}
\rput[bl](190,10){$T_q\tilde{M}$}
\rput[bl](207,-15){$M_q$}
}
\end{pspicture}
\caption{Osculating vacua.}
\label{figosculate}
\end{figure}%
Intuitively speaking, the idea is to choose~$\Phi_p$ in such a way that the spacetimes~$M$ 
and~$\tilde{M}_p$ ``agree as much as possible near~$p$''. In order to make this statement 
mathematically precise, one can set up the following
variational principle at~$p$: We choose a fixed basis~$(e_i)_{i=, \ldots, k}$ of~$M$
(for example an orthonormal basis with respect to a given scalar product on~$M$,
as will be explained for causal fermion systems in Section~\ref{secosccfs} below).
Denoting the total derivative by~$D$, i.e.\
\[ D \Phi_p \::\: T_x\F \rightarrow T_{\Phi(x)}\F 
\qquad \text{and} \qquad
D \Phi_p|_\0 \::\: M=T_\0\F \rightarrow T_p\F \:, \]
we obtain a basis of~$M_p \simeq T_p M_p \subset T_p\F$ by
\[ e_i(p) := D \Phi_p|_{\mathbf{0}} \,e_i \in T_p \F \]
We set
\beq \label{Spdef}
\Sact_p(\Phi_p) = \sum_{i=1}^k D^2 \tilde{\ell}|_p \big( e_i(p), e_i(p) \big) \:,
\eeq
where~$\tilde{\ell}$ is defined by~\eqref{elldef} with~$\rho$ replaced by~$\tilde{\rho}$. Note that, in view of the EL equations~\eqref{EL1} for~$\tilde{\rho}$, the first derivative of~$\tilde{\ell}$
vanishes at~$p$. This is why second derivative in~\eqref{Spdef} is well-defined without the need to specify
a connection on~$\F$. Moreover, it follows from~\eqref{EL1} that the second derivatives in~\eqref{Spdef} are non-negative.
Therefore, the action~$\Sact_p$ is non-negative.
Now one chooses~$\Phi_p$ as a minimizer of the~$\Sact_p$ under variations~$\Phi_p \in {\mathcal{G}}_p$. More details on this variational principles will be given in the setting of causal fermion systems in Sections~\ref{secosccfs}--\ref{secoscEL}.

We finally comment on the origin of the name ``osculating vacua''.
The notion ``osculating'' (literally ``kissing'') can be found in the older literature
(see for example~\cite{lichnerowiczintro, laugwitz}) for a
Euclidean metric which approximates a Riemannian metric in a neighborhood of a point. In this setting, the Euclidean and Riemannian metrics are osculating
if they coincide in a first order Taylor expansion about a base point.
This is as good as possible, because the second derivatives of the Riemannian
metric involve curvature, which clearly cannot be compensated by a coordinate
transformation. Our notion of ``osculating vacua'' is similar in the sense
that~$M_p$ should approximate~$\tilde{M}$ near~$p$ as good as possible
(as will be made precise in Section~\ref{secosccfs} by the variational principle~\eqref{Spdef}). But, in contrast to the setting of Riemannian geometry,
the symmetry transformation~$\Phi_p$ not only tries to adjust the geometry,
but also tries to transform away the matter and fields
described by~$\tilde{\rho}$.

\subsection{Formulation for the Causal Action Principle} \label{secosccfs}
We now specify to the setting of causal fermion systems.
In this setting, the symmetry transformations and the construction of
osculations can be worked out more concretely.

We are given two causal fermion systems~$(\H, \F, \rho)$
(the {\em{vacuum}}) and~$(\tilde{\H}, \tilde{\F}, \tilde{\rho})$ (the {\em{interacting system}}).
We assume that both~$\rho$ and~$\tilde{\rho}$ are minimizers of the causal
action principle for the same value of the Lagrange parameter~$\s$
(that~$\s$ has the same value for both measures can be arranged
by a rescaling~$\rho \rightarrow \lambda \rho$ with~$\rho>0$, as
one verifies from~\eqref{elldef}).
After identifying the Hilbert spaces via a unitary
transformation~$V \::\: \H \rightarrow \tilde{\H}$, we can always
work in the Hilbert space~$\H$ (and correspondingly with measures on~$\F$).
However, the non-uniqueness of the identification of the Hilbert spaces
shows up in the freedom to perform unitary transformation of vacuum measure
\[ 
\rho \mapsto \scrU \rho \:, \]
where~$\scrU \rho$ is defined by
\[ 
(\scrU \rho)(\Omega) := \rho \big( \scrU^{-1} \,\Omega\, \scrU \big) 
\qquad \text{for} \qquad \Omega \subset \F \:. \]

The Lagrangian is unitary invariant in the sense that
\[ \L(\scrU x \scrU^{-1}, \scrU y \scrU^{-1}) = \L(x,y) \qquad \text{for all~$x,y \in \F$}\:. \]
Therefore, the transformation~$\Phi(x) = \scrU x \scrU^{-1}$ is a symmetry
of the Lagrangian as introduced in Definition~\ref{defiso}.
For our purposes, it suffices to restrict attention to symmetry transformations of this form (in fact, in~\cite{paganini+yadav} it is even proven that
every symmetry transformation which preserves the constraints
in the causal action principle can be realized by a unitary transformation).
Then the group~$\G$ of symmetries of the Lagrangian can be identified with
the unitary group~$\U(\H)$ of the Hilbert space.

We choose the vacuum measure as the regularized Dirac sea
vacuum in four-di\-men\-sio\-nal Minkowski space
(as constructed in detail in~\cite{oppio}; see also the
textbooks~\cite[Section~5.5]{intro} and~\cite[Section~1.2]{cfs}).
The detailed form of this measure will not used here;
we only need that~$M:= \supp \rho$ can be identified with 
four-dimensional Minkowski space.
We choose such an identification, which will be kept fixed throughout our
constructions.
In particular, on~$M$ we have a Lorentzian metric~$\eta$ and a distinguished point~$\mathbf{0}$ 
(the origin). Moreover, the regularization distinguishes a time direction. For convenience, we only consider bases~$(e_i)$ where~$e_0$ points in the time direction distinguished by the regularization. Moreover, we choose the bases to be pseudo-orthonormal
with respect to the Minkowski metric~$\eta$.
Clearly, the Lagrangian is translation symmetric on~$M$~\eqref{translation}
and reflection symmetric~\eqref{reflection}, and the spacetime measure~$\rho|_M$
is a multiple of the Lebesgue measure.
Finally, all spacetime point operators have the same eigenvalues (including
multiplicities), implying that
the symmetry group~$\G$ acts transitively on~$M$.

We again begin with the simplest setting that~$\G$ also acts transitively
on~$\tilde{M}$. This condition means that all spacetime point operators of~$\tilde{M}$
must have the same eigenvalues as the operator~$\0 \in M$ (again with
multiplicities). Let~$p \in \tilde{M}$.
The there are unitary transformations~$\scrU$ with the property that
\beq \label{transprop}
p=\scrU \mathbf{0} \scrU^{-1} \:.
\eeq
Given such an operator, we define the measure
\beq \label{rhop}
\rho_p = \scrU \rho \qquad \text{and set} \qquad
M_p := \supp \rho_p = \scrU M \scrU^{-1} \:.
\eeq
Then~$p$ is the origin of~$M_p$ and, moreover, it is a point of the interacting spacetime~$\tilde{M}$.
In view of the EL equations~\eqref{EL1} for~$\tilde{\rho}$, we know that
the function~$\tilde{\ell}$ defined by
\[ \tilde{\ell} : \F \rightarrow \R \:,\qquad \tilde{\ell}(x) = \int_\F \L(x,y)\:
d\tilde{\rho}(y) - \s \]
is minimal at~$p$. Therefore, it has a vanishing derivative at~$p$, and thus
\[ D_{e_i} \tilde{\ell}(p) = 0 \qquad \text{for all~$i=0,\ldots 3$} \]
(where for notational simplicity we set~$e_i=e_i(p)$).
Moreover, second derivatives are non-negative
(for details see~\cite{positive}). With this in mind, we set up the following
variational principle. Restricting~$\tilde{\ell}$ to~$M_p$, we obtain a mapping
\[ \tilde{\ell}|_{M_p} : M_p \rightarrow \R \:. \]
Being a mapping from a vector space to the reals, we can take ordinary
higher derivatives. The Hessian is positive because of the EL equations.
Therefore, the functional
\beq \label{Spdefnew}
\Sact_p(\scrU) := \sum_{i=0}^3 D^2_{e_i, e_i} \tilde{\ell} \big|_{M_p}(p)
\eeq
is non-negative. Moreover, it is independent of the choice of the basis~$(e_i)_{i=1,\ldots, k}$.
Indeed, again restricting attention to pseudo-orthonormal
bases where~$e_0$ points into the distinguished time direction,
the sum in~\eqref{Spdefnew} is obviously invariant under the remaining freedom of spatial
rotations of the basis vectors.

We now minimize~$\Sact$ under variations of~$\scrU$.
\begin{Def} \label{def:osculating} The unitary operator~$\scrU$ is said to be {\bf{osculating}} 
at~$p \in \tilde{M}$ if it is a minimizer of the functional~$\Sact_p(\scrU)$. We also denote an 
osculating unitary operator by~$\scrU_p$.
\end{Def}

We next move on to the general case that~$\G$ does not act transitively on~$\tilde{M}$.
In this case, for~$x \in \F$ we introduce the corresponding Euclidean sign
operator~$s_x$ as the unique operator which has the same eigenspaces as~$x$,
with eigenvalue one for the positive eigenvalues, and eigenvalue minus one for
the negative eigenvalues. Then the operators~$s_\0$ and~$s_p$ are
isospectral (having the non-zero eigenvalues plus and minus one, both of
multiplicity~$n$). We now replace~\eqref{transprop} by the weaker condition
\beq \label{notransprop}
s_p = \scrU s_\0 \scrU^{-1} \:.
\eeq
We now minimize the functional~\eqref{Spdefnew} on the resulting class of unitary operators.

We finally remark that, in most physical situations, the eigenvalues of
all spacetime point operators agree, up to errors of higher order in the
regularization length. With this in mind, it is a very good approximation to
assume that~$\G$ acts transitively on~$\tilde{M}$.

\subsection{The Euler-Lagrange Equations for the Osculation} \label{secoscEL}
In the case~$\dim \H < \infty$, the existence of minimizers of the variational principle
for the osculation follows immediately from a compactness argument using that
the integrated Lagrangian is continuous (as worked out in detail
in~\cite{lagrange-hoelder}). We do not expect in general that the minimizers
will be unique.

We now work out the corresponding EL equations, referred to as the {\em{osculation equations}}.
We begin with a particularly simple special situation when there
is a unitary operator~$\scrU$ for which the functional~$\Sact_p(\scrU)$ vanishes.
\begin{Def} \label{defoptimal} $M_p$ is an {\bf{optimal osculation}} at~$p$ if~$\Sact_p(\scrU)=0$.
\end{Def} \noindent
In this case, the osculation conditions are simple to state.
\begin{Prp} \label{prpoptimal}
If~$M_p$ is an optimal osculation at~$p$, then the osculation equations
imply that the following Hessian vanishes,
\beq \label{osceq}
D^2 \tilde{\ell}|_{M_p}(p) = 0
\eeq
\end{Prp} \noindent
The notation for the derivatives in~\eqref{osceq} is understood as follows.
We first restrict~$\tilde{\ell}$ to~$M_p$. Since~$M_p$ is a vector space, we
can take second derivatives (which are directional derivatives tangential to~$M_p$).
\Proof[Proof of Proposition~\ref{prpoptimal}.] Since the Hessian of of~$\tilde{\ell}_{M_p}$ is non-
negative, the vanishing of its trace implies that all its matrix entries are zero.
\QED
There is a simple special case when optimal osculations exist.
\begin{Lemma} \label{lemmaosc}
Assume that the tangent spaces of~$\tilde{M}$ and~$M_p$ coincide,
\beq \label{TpMp}
T_p\tilde{M} = T_\0M_p \:.
\eeq
Then~$M_p$ is an optimal osculation at~$p$.
\end{Lemma}
\Proof According to the EL equations~\eqref{EL1}, the function~$\tilde{\ell}$
is minimal and vanishes identically on~$\tilde{M}$. Using that~$M_p$ is tangential~$\tilde{M}$ 
in~$p$, second derivatives along~$M_p$ can be rewritten as second derivatives along~$\tilde{M}$ and 
first derivatives transversal to~$\tilde{M}$.
All these derivatives vanish by the EL equations.
\QED

A simple example where optimal osculations exist is the trivial case~$\tilde{\rho}=\rho$
of the vacuum. In general interacting situations, we cannot expect that optimal
osculations exist. Even in this case, the notion of optimal osculations is important
because, in many situations, it is a very good approximation to assume that the osculation is optimal (for details see Appendix~\ref{appoptimal}).
In the general case of a non-optimal osculation, we have
the following weaker result. Suppose that~$\scrU_\tau$ with~$\tau \in (-\varepsilon, \varepsilon)$ (and~$\varepsilon>0$
is an {\em{admissible variation}} in the sense
that it is a smooth curve of unitary operators with~$\scrU_0=\scrU$
which (depending on whether we are in case~\eqref{transprop}
or~\eqref{notransprop}) satisfies the constraints
\[ \scrU_\tau \0 \scrU_\tau^{-1} = p \qquad \text{or} \qquad
\scrU_\tau s_\0 \scrU_\tau^{-1} = s_p \:. \]
Then the infinitesimal generator of this variation is a vector
field~$\Comm$ on~$\F$ acting on a real-valued function~$h$ on~$\F$ by
\[ D_\Comm h(x) := \frac{d}{d\tau} h\big(\scrU_\tau \scrU^{-1} \,x\, \scrU \scrU_\tau^{-1}
\big) \big|_{\tau=0} \]
($\Comm$ is also referred to as a {\em{commutator jet}}; for
the general context see~\cite{dirac} or~\cite[Section~8.2]{intro}).
Since minimality implies criticality, we immediately obtain the
following result.

\begin{Prp} \label{prposc1}
If~$\scrU$ is osculating, then
\beq \label{osceqgen} 
\sum_{i=0}^3 D_{e_i, e_i} D_\Comm \tilde{\ell}|_{M_p}(p) = 0
\eeq
for all admissible variations.
\end{Prp}

The EL equation~\eqref{osceqgen} does not immediately
give information about the resulting correction to the osculation equation in
the optimal case~\eqref{osceq}. To this end, the following result seems
helpful.
\begin{Lemma} For a general osculation~$M_p$
and any basis~$\tilde{e}_i$ of $T_p\tilde{M}$,
\beq \label{osceqgen2}
D^2_{e_i, e_j} \tilde{\ell}|_{M_p}(p) =
D^2\tilde{\ell}|_p\big(e_i- \tilde{e}_i,\: e_j - \tilde{e}_j \big)  \:.
\eeq
\end{Lemma}
\Proof According to the El equations~\eqref{EL1} for~$\tilde{\ell}$,
we know that~$D \tilde{\ell}(x)=0$ for all~$x \in \tilde{M}$.
Differentiating at~$x=p$, we conclude that
\[ D^2 \tilde{\ell}|_p\big(\tilde{e}_i, u) = 0 \qquad \text{for all~$u \in T_p \F$} \:. \]
Using this equation after multiplying out the right side 
of~\eqref{osceqgen2} gives the result.
\QED
We again point out that one should carefully distinguish between
the tangent space~$T_p \tilde{M}$ and the osculating vacuum~$M_p$
(as is highlighted in Figure~\ref{figosculate}).
On the other hand, one should also keep in mind that, in most physical
situations, $T_p \tilde{M}$ and~$M_p$ will agree up to 
small errors (for details see the analysis in Appendix~\ref{appoptimal}).
With this in mind, in this paper we will restrict attention to the case
that the osculating vacua are tangential at every point.
\begin{Def} \label{defoscoptimal}
The family~$(M_p)_{p \in \tilde{M}}$ forms {\bf{optimal osculations}}
if for every~$p \in \tilde{M}$ the condition~\eqref{TpMp} holds.
\end{Def} \noindent
This assumption is a major simplification, because for optimal 
osculations we can work  at every point~$p \in \tilde{M}$ with the
simple osculation equation~\eqref{osceq}.
Nevertheless, the osculation equations
in Proposition~\ref{prposc1} will be important when working out corrections
to the Einstein equations, as will be discussed in Section~\ref{seccorrect}
and Appendix~\ref{appnonoptimal}.
\section{The $\L$-Geometry} \label{secL}
We now specialize the setting by assuming that~$\tilde{M}$ has a smooth manifold structure.

\subsection{$\L$-Induced Charts} \label{secchart}
We consider the setting of causal variational
principles introduced in Section~\ref{seccvp}.
Then~$\tilde{M}$ has the structure
of a smooth $k$-dimensional manifold embedded in~$\F$.
Moreover, we assume that we have chosen
an optimal osculation~$M_p$ at every point~$p \in \tilde{M}$
(see Definition~\ref{defoscoptimal}). Given a base point~$p \in \tilde{M}$,
we introduce the mapping
\beq \label{phiint}
\phi_p : \tilde{M} \rightarrow M_p \:,\qquad \phi_p(\tilde{x}) := 
\fint_{M_p} \L(\tilde{x},y)\: y\: d\rho_p(y) \:,
\eeq
where the integral sign with bar means that the integral is rescaled with
a factor~$1/\s$, i.e.\
\[ 
\fint \cdots := \frac{1}{\s} \int \cdots \:. \]
The differential of the this map at~$p$
gives a linear mapping from the geometric tangent space to~$M_p$,
\beq \label{frakSdef}
D_{\tilde{x}} \phi_p(\tilde{x})|_{\tilde{x}=p} \::\: T_p \tilde{M} \rightarrow T_\0 M_p \simeq M_p\:.
\eeq
As explained in Section~\ref{secoscEL}, we here restrict attention to the
case that the {\em{osculations are optimal}}. In this case, 
the mappings~$(\phi_p)_*|_p$ are all the identity.
Therefore, in the following constructions
we will simply identify every tangent space~$T_p \tilde{M}$ with
the corresponding osculating vacuum~$M_p$.
This means that the origin of~$M_p$, which we denote for clarity by~$\0_p$,
can be identified with the base point~$p$ of the tangent space, i.e.\ $\0_p = p$.
Using these identifications, we have the following result.
\begin{Lemma} \label{lemmadphi} The mapping~$\phi_p$ has the properties
\begin{align}
\phi_p(p) &= \0_p \label{phi0} \\
D_{\tilde{x}} \phi_p({\tilde{x}}) \big|_{{\tilde{x}}=p} &= \1_{M_p} \:, \label{delta}
\end{align}
where~$\1_{M_p} : M_p \rightarrow M_p$ is the identity map.
\end{Lemma} \noindent
\Proof The relation~\eqref{phi0} is an immediate consequence of the
translation symmetry of the Lagrangian on~$M_p$.
For the proof of~\eqref{delta}, we choose a tangent vector~$v \in M_p$.
Then
\[ D_v \big( \phi_p (\tilde{x}) \big) \Big|_{{\tilde{x}}=p} = 
\fint_{M_p} D_{1,v} \L(x,y)\: y\: d\rho_p(y) \Big|_{x=p} = 
\fint_{M_p} v^j\frac{\partial}{\partial x^j} \L(x,y)\: y\: d\rho_p(y) \Big|_{x=p} \]
(where we represent the vectors of~$M_p$ in a basis by~$x=x^j \partial_j|_p$).
Now we can use the translation symmetry of the Lagrangian on~$M_p$ to obtain
\begin{align*}
D_v \big( \phi_p (\tilde{x}) \big) \Big|_{{\tilde{x}}=p} &= 
-\fint_{M_p} v^j \:\bigg( \frac{\partial}{\partial y^j} \L(x,y) \bigg)\: y\: d\rho_p(y) \Big|_{x=p} \\
&=\fint_{M_p} v^j \:\L(x,y) \:\frac{\partial y}{\partial y^j} \: d\rho_p(y) \Big|_{x=p} 
=v \:\fint_{M_p}\L(p,y) \: d\rho_p(y) = v
\end{align*}
(where in the last line we integrated by parts). This concludes the proof.
\QED
In view of~\eqref{delta}, the mapping~$\phi_p$ is a local diffeomorphism.
Therefore, there is a neighborhood~$V_p \subset \tilde{M}$ of~$p$ such that the
restriction
\[ \phi_p|_{V_p} \::\: V_p \rightarrow \phi_p(V) \subset M_p \]
is a diffeomorphism. Since~$M_p$ is a vector space (which we could
identify with~$\R^k$), we can regard~$(V_p, \phi_p|_{V_p})$ as a chart of~$\tilde{M}$.

\begin{Def} We refer to~$(V_p, \phi_p|_{V_p})$ as the
{\bf{$\L$-induced chart}} centered at~$p$.
\end{Def}

\subsection{The Connection $\nabla^\L$}
The $\L$-induced charts immediately give rise to a connection.
To this end, given a vector field~$u \in
\Gamma(\tilde{M}, T\tilde{M})$, we first take the directional derivative of~$\phi_p$,
\[ (u \phi_p) : \tilde{M} \rightarrow M_p\:,\qquad
(u \phi_p)({\tilde{x}}) := D_u \phi_p({\tilde{x}}) \in M_p \:. \]
Acting with a tangent vector~$v \in T_p \tilde{M}$, we can introduce the
covariant derivative by
\beq \label{nablaLdef}
\nabla^\L_v u \big|_p := v \big( u \phi_p \big) \big|_{{\tilde{x}}=p} \in M_p\:.
\eeq
This construction is illustrated in Figure~\ref{figLconn}.
\begin{figure}[tb]
\psset{xunit=.8pt,yunit=.8pt,runit=.8pt}
\begin{pspicture}(273.72626838,75.42655644)
{
\newrgbcolor{curcolor}{0 0 0}
\pscustom[linewidth=1.51181105,linecolor=curcolor,linestyle=dashed,dash=0.40000001 0.80000001]
{
\newpath
\moveto(0.00023433,69.30907739)
\lineto(273.72628535,69.39162227)
}
}
{
\newrgbcolor{curcolor}{0 0 0}
\pscustom[linestyle=none,fillstyle=solid,fillcolor=curcolor]
{
\newpath
\moveto(108.60202056,69.0095837)
\curveto(108.60202056,67.44405213)(107.33290646,66.17493803)(105.76737489,66.17493803)
\curveto(104.20184332,66.17493803)(102.93272922,67.44405213)(102.93272922,69.0095837)
\curveto(102.93272922,70.57511528)(104.20184332,71.84422937)(105.76737489,71.84422937)
\curveto(107.33290646,71.84422937)(108.60202056,70.57511528)(108.60202056,69.0095837)
\closepath
}
}
{
\newrgbcolor{curcolor}{0 0 0}
\pscustom[linewidth=0,linecolor=curcolor]
{
\newpath
\moveto(108.60202056,69.0095837)
\curveto(108.60202056,67.44405213)(107.33290646,66.17493803)(105.76737489,66.17493803)
\curveto(104.20184332,66.17493803)(102.93272922,67.44405213)(102.93272922,69.0095837)
\curveto(102.93272922,70.57511528)(104.20184332,71.84422937)(105.76737489,71.84422937)
\curveto(107.33290646,71.84422937)(108.60202056,70.57511528)(108.60202056,69.0095837)
\closepath
}
}
{
\newrgbcolor{curcolor}{0 0 0}
\pscustom[linewidth=1.75748033,linecolor=curcolor]
{
\newpath
\moveto(2.28240378,12.09753156)
\curveto(37.60322268,40.25385156)(72.92334992,68.40962353)(108.4463622,68.65832022)
\curveto(143.96937071,68.90701314)(179.69170772,41.24922038)(206.9542715,26.25219487)
\curveto(234.21683528,11.25516936)(253.01650772,8.92076306)(271.81659591,6.58630385)
}
}
{
\newrgbcolor{curcolor}{0 0 0}
\pscustom[linestyle=none,fillstyle=solid,fillcolor=curcolor]
{
\newpath
\moveto(240.0259275,69.59102286)
\curveto(240.0259275,68.02549129)(238.75681341,66.7563772)(237.19128183,66.7563772)
\curveto(235.62575026,66.7563772)(234.35663616,68.02549129)(234.35663616,69.59102286)
\curveto(234.35663616,71.15655444)(235.62575026,72.42566853)(237.19128183,72.42566853)
\curveto(238.75681341,72.42566853)(240.0259275,71.15655444)(240.0259275,69.59102286)
\closepath
}
}
{
\newrgbcolor{curcolor}{0 0 0}
\pscustom[linewidth=0,linecolor=curcolor]
{
\newpath
\moveto(240.0259275,69.59102286)
\curveto(240.0259275,68.02549129)(238.75681341,66.7563772)(237.19128183,66.7563772)
\curveto(235.62575026,66.7563772)(234.35663616,68.02549129)(234.35663616,69.59102286)
\curveto(234.35663616,71.15655444)(235.62575026,72.42566853)(237.19128183,72.42566853)
\curveto(238.75681341,72.42566853)(240.0259275,71.15655444)(240.0259275,69.59102286)
\closepath
}
}
{
\newrgbcolor{curcolor}{0 0 0}
\pscustom[linestyle=none,fillstyle=solid,fillcolor=curcolor]
{
\newpath
\moveto(171.23977739,49.36410425)
\curveto(171.23977739,47.79857267)(169.97066329,46.52945858)(168.40513172,46.52945858)
\curveto(166.83960015,46.52945858)(165.57048605,47.79857267)(165.57048605,49.36410425)
\curveto(165.57048605,50.92963582)(166.83960015,52.19874992)(168.40513172,52.19874992)
\curveto(169.97066329,52.19874992)(171.23977739,50.92963582)(171.23977739,49.36410425)
\closepath
}
}
{
\newrgbcolor{curcolor}{0 0 0}
\pscustom[linewidth=0,linecolor=curcolor]
{
\newpath
\moveto(171.23977739,49.36410425)
\curveto(171.23977739,47.79857267)(169.97066329,46.52945858)(168.40513172,46.52945858)
\curveto(166.83960015,46.52945858)(165.57048605,47.79857267)(165.57048605,49.36410425)
\curveto(165.57048605,50.92963582)(166.83960015,52.19874992)(168.40513172,52.19874992)
\curveto(169.97066329,52.19874992)(171.23977739,50.92963582)(171.23977739,49.36410425)
\closepath
}
}
{
\newrgbcolor{curcolor}{0 0 0}
\pscustom[linewidth=1.00157475,linecolor=curcolor]
{
\newpath
\moveto(169.71995717,49.04216967)
\lineto(212.25194835,28.34074857)
}
}
{
\newrgbcolor{curcolor}{0 0 0}
\pscustom[linewidth=1.00157475,linecolor=curcolor]
{
\newpath
\moveto(208.46040411,26.84445418)
\lineto(212.47708993,28.23116633)
\lineto(211.09037777,32.24785216)
}
}
{
\newrgbcolor{curcolor}{0 0 0}
\pscustom[linewidth=1.00157475,linecolor=curcolor]
{
\newpath
\moveto(223.36683591,18.38922983)
\lineto(263.07587906,2.20448022)
}
}
{
\newrgbcolor{curcolor}{0 0 0}
\pscustom[linewidth=1.00157475,linecolor=curcolor]
{
\newpath
\moveto(259.39117826,0.46158474)
\lineto(263.30775247,2.10997245)
\lineto(261.65936475,6.02654666)
}
}
{
\newrgbcolor{curcolor}{0 0 0}
\pscustom[linewidth=1.75748033,linecolor=curcolor]
{
\newpath
\moveto(105.92195906,69.17900117)
\lineto(176.27749417,69.55536652)
}
}
{
\newrgbcolor{curcolor}{0 0 0}
\pscustom[linewidth=1.75748033,linecolor=curcolor]
{
\newpath
\moveto(171.47269682,64.25714694)
\lineto(176.71685797,69.55771689)
\lineto(171.41628802,74.80187804)
}
}
{
\newrgbcolor{curcolor}{0 0 0}
\pscustom[linewidth=1.00157475,linecolor=curcolor]
{
\newpath
\moveto(30.99008126,34.65488778)
\lineto(66.13538268,61.96694448)
}
}
{
\newrgbcolor{curcolor}{0 0 0}
\pscustom[linewidth=1.00157475,linecolor=curcolor]
{
\newpath
\moveto(65.80429805,57.90429908)
\lineto(66.33309466,62.12059008)
\lineto(62.11680367,62.64938668)
}
}
{
\newrgbcolor{curcolor}{0 0 0}
\pscustom[linestyle=none,fillstyle=solid,fillcolor=curcolor]
{
\newpath
\moveto(32.3859984,34.47998694)
\curveto(32.3859984,33.43629922)(31.53992234,32.59022316)(30.49623462,32.59022316)
\curveto(29.4525469,32.59022316)(28.60647084,33.43629922)(28.60647084,34.47998694)
\curveto(28.60647084,35.52367466)(29.4525469,36.36975072)(30.49623462,36.36975072)
\curveto(31.53992234,36.36975072)(32.3859984,35.52367466)(32.3859984,34.47998694)
\closepath
}
}
{
\newrgbcolor{curcolor}{0 0 0}
\pscustom[linewidth=0,linecolor=curcolor]
{
\newpath
\moveto(32.3859984,34.47998694)
\curveto(32.3859984,33.43629922)(31.53992234,32.59022316)(30.49623462,32.59022316)
\curveto(29.4525469,32.59022316)(28.60647084,33.43629922)(28.60647084,34.47998694)
\curveto(28.60647084,35.52367466)(29.4525469,36.36975072)(30.49623462,36.36975072)
\curveto(31.53992234,36.36975072)(32.3859984,35.52367466)(32.3859984,34.47998694)
\closepath
}
}
{
\newrgbcolor{curcolor}{0 0 0}
\pscustom[linestyle=none,fillstyle=solid,fillcolor=curcolor]
{
\newpath
\moveto(226.91309609,17.62498436)
\curveto(226.91309609,16.58129665)(226.06702002,15.73522058)(225.02333231,15.73522058)
\curveto(223.97964459,15.73522058)(223.13356853,16.58129665)(223.13356853,17.62498436)
\curveto(223.13356853,18.66867208)(223.97964459,19.51474814)(225.02333231,19.51474814)
\curveto(226.06702002,19.51474814)(226.91309609,18.66867208)(226.91309609,17.62498436)
\closepath
}
}
{
\newrgbcolor{curcolor}{0 0 0}
\pscustom[linewidth=0,linecolor=curcolor]
{
\newpath
\moveto(226.91309609,17.62498436)
\curveto(226.91309609,16.58129665)(226.06702002,15.73522058)(225.02333231,15.73522058)
\curveto(223.97964459,15.73522058)(223.13356853,16.58129665)(223.13356853,17.62498436)
\curveto(223.13356853,18.66867208)(223.97964459,19.51474814)(225.02333231,19.51474814)
\curveto(226.06702002,19.51474814)(226.91309609,18.66867208)(226.91309609,17.62498436)
\closepath
}
\rput[bl](275,2){$\tilde{M}$}
\rput[bl](102,52){$p$}
\rput[bl](164,33){$\tilde{x}$}
\rput[bl](233,54){$y$}
\rput[bl](280,62){$M_p = T_p \tilde{M}$}
\rput[bl](135,74){$v$}
\rput[bl](190,40){$u(\tilde{x})$}
\rput[bl](244,-2){$u$}
\rput[bl](45,55){$u$}
}
\end{pspicture}
\caption{The connection $\nabla^\L$.}
\label{figLconn}
\end{figure}%

\subsection{The Riemannian Metric~$g$ on~$\tilde{M}$} \label{secg}
The Lagrangian also induces a Riemannian metric on~$\tilde{M}$, as we now
explain. We denote the dual space of~$M_p$ by~$M_p^*$.
For two forms~$\phi, \phi' \in M_p^*$ we set
\beq \label{gpstar}
g^*_p \big(\phi, \phi' \big) := \frac{1}{\delta^2} \fint_{M_p} \L(p,y)\: \phi(y)\: \phi'(y)\: d\rho_p(y)
\eeq
(the factor~$1/\delta^2$ is introduced in order to ensure that~$g^*$ stays finite
in the limit~$\delta \searrow 0$ when the range of the Lagrangian tends to zero).
This defines a positive semi-definite bilinear form on~$M_p^*$,
\[ g^*_p : M_p^* \times M_p^* \rightarrow \R \:. \]
This bilinear form is even positive definite, as the following consideration shows:
If~$g^*_p$ had a kernel, then the function~$\L(p,.)$ would be supported on a hyperplane of~$M_p$. But, in this case, due to our smoothness assumption,
$\L(p,.)$ would vanish everywhere. In view of the translational invariance~\eqref{translation},
the Lagrangian would vanish on~$M \times M$, a contradiction.

We conclude that~$g^*_p$ defines a scalar product on~$M_p^*$.
We thus obtain an identification of~$M_p^*$ and~$M_p$ by
\[ g^*( \,.\, , \phi) \in (M^*_p)^* \simeq M_p \:. \]
This makes it possible to define a scalar product~$g_p$ on~$M_p$ via the relation
\[ g_p\big(u, g^*_p( \,.\,, \phi) \big) = \phi(u) \qquad \text{for all~$u \in M_p$ and~$\phi \in M_p^*$}\:. \]
Choosing a basis of~$M_p$ and a corresponding dual basis of~$M_p^*$,
the metric~$g_p$ is represented as usual by the inverse of~$g_p^*$
(for general background on Riemannian geometry see for example~\cite{lee-manifold}).

We remark that, using similar constructions, one can also introduce a
{\em{Weingarten map}} on~$\tilde{M}$. Since these constructions will
not be used for
the formulation of the Einstein equations, we do not give them here
but refer the interested reader to Appendix~\ref{apph}.

At this stage, it is unclear how the metric~$g_p$ and the connection~$\nabla^\L$
are related to each other. Before entering a detailed study of this question,
we explain how these objects can be described in local charts.

\subsection{Description with Osculation Maps} \label{secosc}
It is most convenient to work with $\L$-induced charts.
We denote the base point by~$q \in \tilde{M}$. It will be fixed throughout our constructions.
Therefore, we omit the subscripts~$q$ by writing~$M=M_q$, $V=V_q$,
$\phi = \phi_q$, $\rho=\rho_q$, and so on. Then the $\L$-induced chart~$(\phi, V)$
gives rise to a local parametrization
\begin{align} \label{eq:Fparam}
\tilde{F} := (\phi|_V)^{-1} \::\: U := \phi^{-1}(V) \subset M \rightarrow \tilde{M} \:.
\end{align}

This parametrization makes it possible to write the
interacting measure locally in a neighborhood of~$q$ in the usual form as
\beq \label{rhotilderho}
\tilde{\rho} = \tilde{F}_* (\tilde{f} \rho) \qquad \text{with} \qquad
\tilde{f} \in C^\infty(U, \R^+), \;\tilde{F} \in C^\infty(U, \F)
\eeq
with a weight function~$\tilde{f}$ (some properties of the weight function
will be derived in Section~\ref{secf} using the EL equations).
Finally, we choose a basis~$e_i$ of~$M$ and denote the components by~$x^i$, i.e.\
\[ x = x^i e_i \in M \:. \]
We note that, working in this chart and parametrization, for all~$x \in U \subset M$,
\beq \label{xalign}
x = \fint_{M} \L\big( \tilde{F}(x), y \big)\:y\: d\rho(y)
\eeq
or, in basis components of~$M$,
\[  x^k = \fint_{M} \L\big( \tilde{F}(x), y \big)\:y^k\: d\rho(y)\:. \]
Next consider~$p \in U \subset M$. In our chart it describes a point~$\tilde{F}(p) \in \tilde{M}$.
For ease in notation, we denote the osculating vacuum
at~$\tilde{F}(p)$ by~$M_p$ (thus, for the considered optimal osculation,
$M_p = T_{\tilde{F}(p)} \tilde{M}$). Our chart gives rise to a distinguished basis
of~$M_p$
\[ \frac{\partial}{\partial x^i} \Big|_p \in M_p \:. \]
Its dual basis of~$M_p^*$ is denoted by
\[ dx^i \big|_p \in M_p^*\:,\qquad \text{so that}\quad dx^i \big|_p \bigg( \frac{\partial}{\partial x^j} \Big|_p \bigg) = \delta^i_j \:. \]
We will always work with these bases in components.

In order to illustrate this formalism, we write out the covariant derivative~$\nabla^\L$ in components of~$M_p$,
\begin{align*}
\nabla^\L_i u^j \big|_p &= \frac{\partial}{\partial x^i} \big( u \phi^j_p(x) \big) \Big|_{x=p}
= \frac{\partial}{\partial x^i} \bigg( u^k(x) \frac{\partial}{\partial x^k} \phi^j_p(x) \bigg) \bigg|_{x=p} \\
&= \partial_i u^j(p) + u^k(p)\: \frac{\partial^2}{\partial x^i \partial x^k} \phi^j_p(x) \bigg|_{x=p} \:.
\end{align*}
Hence we can write the $\L$-connection with usual Christoffel symbols as
\beq \label{Lconn}
\nabla^\L_i u^j \big|_p = \partial_i u^j \big|_p + \Gamma^j_{ik}(p)\: u^k \big|_p 
\qquad \text{with} \qquad
\Gamma^j_{ik}(p) = \frac{\partial^2}{\partial x^i \partial x^k} \phi^j_p(x) \Big|_{x=p} \:.
\eeq
This formula shows explicitly that the $\nabla^\L$-connection is torsion-free.

Our next task is to write the integrals over~$M_p$ (as in~\eqref{phiint}
and~\eqref{gpstar}) in components. Here the difficulty arises that this
makes it necessary to also describe the change of the osculation.
Working with a unitary transformations (as in~\eqref{rhop}) has the disadvantage
that the resulting isometry from~$M$ to~$M_p$ may not have a simple form
in our chosen coordinate system. For this reason, it is more convenient
to describe the osculation in a way similar to~\eqref{rhotilderho}.
\begin{Lemma} \label{lemmaFp} There is a unique
\[ \text{invertible affine linear map} \qquad F_p : M \rightarrow M_p \]
with the following properties:
\bitem
\item[{\rm{(i)}}] The mapping is compatible with the bases of the osculating vacua
in the sense that
\[ 
F_p \: \frac{\partial}{\partial x^i} \Big|_q = \frac{\partial}{\partial x^i} \Big|_p \:. \]
\item[{\rm{(ii)}}] The base point is mapped to the origin,
\beq \label{Fp2}
F_p(p) = \0_{\tilde{F}(p)} \:.
\eeq
\eitem
Moreover, there is a unique number~$f_p>0$ with the property that
\beq \label{rhoprep}
d\rho_p = f_p \: dx^1|_p \wedge \cdots \wedge dx^k|_p \:.
\eeq
Finally, this number is related to the weight function~$\tilde{f}$ in~\eqref{rhotilderho} by
\beq \label{fpf}
f_p = \frac{\displaystyle \fint_M \L\big( \tilde{F}(p), \tilde{F}(y) \big)\: \tilde{f}(y)\: d\rho(y)}
{\displaystyle \fint_M \L\big( F_p(p), F_p(y)\big)\: d\rho(y)} \:.
\eeq
\end{Lemma}
\Proof The existence and uniqueness of~$F_p$ follows immediately from the fact that an affine transformation is
uniquely determined by its action on the origin and on the basis vectors.
For the proof of~\eqref{rhoprep} we
note that the measure~$\rho_p$ on~$M_p$ coincides with the Lebesgue measure on~$M_p$ in our basis up to a positive constant. 
Finally, in order to derive~\eqref{fpf}, we make use of the
EL equations in the vacuum and in the interacting spacetime,
 \begin{align*}
\int_M \L\big( F_p(x), F_p(y) \big)\: f_p\: d\rho(y) = \s
= \int_M \L\big( \tilde{F}(p)), \tilde{F}(y) \big)\: \tilde{f}(y)\: d\rho(y) \:.
\end{align*}
Since~$f_p$ is constant, we obtain~\eqref{fpf}.
\QED
We remark that~\eqref{fpf} implies that~$f_p \approx \tilde{f}(p)$, with a very small error
which will be specified in Lemma~\ref{lemmafpf}.

We refer to~$F_p$ as the {\em{osculation map}} at~$p$. Our notions are 
illustrated in Figure~\ref{figFp}.
\begin{figure}[tb]
\psset{xunit=.6pt,yunit=.6pt,runit=.6pt}
\begin{pspicture}(355.74823334,158.41741607)
{
\newrgbcolor{curcolor}{0 0 0}
\pscustom[linewidth=1.88976378,linecolor=curcolor]
{
\newpath
\moveto(2.02378205,135.44578394)
\lineto(355.74722646,135.8279282)
}
}
{
\newrgbcolor{curcolor}{0 0 0}
\pscustom[linestyle=none,fillstyle=solid,fillcolor=curcolor]
{
\newpath
\moveto(142.11193858,135.13889519)
\curveto(142.11193858,133.05151976)(140.41978645,131.35936763)(138.33241102,131.35936763)
\curveto(136.24503559,131.35936763)(134.55288346,133.05151976)(134.55288346,135.13889519)
\curveto(134.55288346,137.22627062)(136.24503559,138.91842275)(138.33241102,138.91842275)
\curveto(140.41978645,138.91842275)(142.11193858,137.22627062)(142.11193858,135.13889519)
\closepath
}
}
{
\newrgbcolor{curcolor}{0 0 0}
\pscustom[linewidth=0,linecolor=curcolor]
{
\newpath
\moveto(142.11193858,135.13889519)
\curveto(142.11193858,133.05151976)(140.41978645,131.35936763)(138.33241102,131.35936763)
\curveto(136.24503559,131.35936763)(134.55288346,133.05151976)(134.55288346,135.13889519)
\curveto(134.55288346,137.22627062)(136.24503559,138.91842275)(138.33241102,138.91842275)
\curveto(140.41978645,138.91842275)(142.11193858,137.22627062)(142.11193858,135.13889519)
\closepath
}
}
{
\newrgbcolor{curcolor}{0 0 0}
\pscustom[linewidth=2.89133853,linecolor=curcolor]
{
\newpath
\moveto(2.46744189,104.37093733)
\curveto(55.71952252,121.67212536)(108.9705411,138.97296946)(155.80024441,134.90825386)
\curveto(202.62994016,130.84354205)(243.03377008,105.41410583)(273.76454173,79.5181852)
\curveto(304.49527559,53.62226457)(325.54954961,27.262543)(346.60431496,0.90222804)
}
}
{
\newrgbcolor{curcolor}{0 0 0}
\pscustom[linewidth=0.99999871,linecolor=curcolor]
{
\newpath
\moveto(321.17701795,129.16426697)
\curveto(324.15251906,117.37767421)(324.89862047,105.03232744)(323.36425701,92.97317405)
\curveto(321.79583244,80.64631295)(317.84443843,68.62707516)(311.79240945,57.77422602)
}
}
{
\newrgbcolor{curcolor}{0 0 0}
\pscustom[linestyle=none,fillstyle=solid,fillcolor=curcolor]
{
\newpath
\moveto(310.38731153,55.254523)
\lineto(310.31028932,60.24948327)
\lineto(314.67719231,57.8143049)
\closepath
}
}
{
\newrgbcolor{curcolor}{0 0 0}
\pscustom[linewidth=0.66666583,linecolor=curcolor]
{
\newpath
\moveto(310.38731153,55.254523)
\lineto(310.31028932,60.24948327)
\lineto(314.67719231,57.8143049)
\closepath
}
}
{
\newrgbcolor{curcolor}{0 0 0}
\pscustom[linewidth=0.99999871,linecolor=curcolor]
{
\newpath
\moveto(3.07768894,132.45869141)
\lineto(3.07768894,111.26877771)
}
}
{
\newrgbcolor{curcolor}{0 0 0}
\pscustom[linestyle=none,fillstyle=solid,fillcolor=curcolor]
{
\newpath
\moveto(3.07768894,108.38378143)
\lineto(0.57769217,112.70877585)
\lineto(5.57768572,112.70877585)
\closepath
}
}
{
\newrgbcolor{curcolor}{0 0 0}
\pscustom[linewidth=0.66666583,linecolor=curcolor]
{
\newpath
\moveto(3.07768894,108.38378143)
\lineto(0.57769217,112.70877585)
\lineto(5.57768572,112.70877585)
\closepath
}
}
{
\newrgbcolor{curcolor}{0 0 0}
\pscustom[linewidth=0.99999871,linecolor=curcolor]
{
\newpath
\moveto(27.94864777,132.46335912)
\lineto(27.94864777,118.95255723)
}
}
{
\newrgbcolor{curcolor}{0 0 0}
\pscustom[linestyle=none,fillstyle=solid,fillcolor=curcolor]
{
\newpath
\moveto(27.94864777,116.06756096)
\lineto(25.448651,120.39255537)
\lineto(30.44864454,120.39255537)
\closepath
}
}
{
\newrgbcolor{curcolor}{0 0 0}
\pscustom[linewidth=0.66666583,linecolor=curcolor]
{
\newpath
\moveto(27.94864777,116.06756096)
\lineto(25.448651,120.39255537)
\lineto(30.44864454,120.39255537)
\closepath
}
}
{
\newrgbcolor{curcolor}{0 0 0}
\pscustom[linewidth=1.88976378,linecolor=curcolor]
{
\newpath
\moveto(187.08135307,157.71806363)
\lineto(353.58805795,6.4448882)
}
}
{
\newrgbcolor{curcolor}{0 0 0}
\pscustom[linewidth=2.89133853,linecolor=curcolor]
{
\newpath
\moveto(2.46744189,104.37093733)
\curveto(55.71952252,121.67212536)(108.9705411,138.97296946)(155.80024441,134.90825386)
\curveto(202.62994016,130.84354205)(243.03377008,105.41410583)(273.76454173,79.5181852)
\curveto(304.49527559,53.62226457)(325.54954961,27.262543)(346.60431496,0.90222804)
}
}
{
\newrgbcolor{curcolor}{0 0 0}
\pscustom[linestyle=none,fillstyle=solid,fillcolor=curcolor]
{
\newpath
\moveto(284.10529872,73.63402776)
\curveto(284.10529872,71.54665233)(282.41314659,69.8545002)(280.32577116,69.8545002)
\curveto(278.23839573,69.8545002)(276.5462436,71.54665233)(276.5462436,73.63402776)
\curveto(276.5462436,75.72140319)(278.23839573,77.41355532)(280.32577116,77.41355532)
\curveto(282.41314659,77.41355532)(284.10529872,75.72140319)(284.10529872,73.63402776)
\closepath
}
}
{
\newrgbcolor{curcolor}{0 0 0}
\pscustom[linewidth=0,linecolor=curcolor]
{
\newpath
\moveto(284.10529872,73.63402776)
\curveto(284.10529872,71.54665233)(282.41314659,69.8545002)(280.32577116,69.8545002)
\curveto(278.23839573,69.8545002)(276.5462436,71.54665233)(276.5462436,73.63402776)
\curveto(276.5462436,75.72140319)(278.23839573,77.41355532)(280.32577116,77.41355532)
\curveto(282.41314659,77.41355532)(284.10529872,75.72140319)(284.10529872,73.63402776)
\closepath
}
}
{
\newrgbcolor{curcolor}{0 0 0}
\pscustom[linewidth=0.99999871,linecolor=curcolor]
{
\newpath
\moveto(281.30904945,129.4368121)
\lineto(280.87176945,87.91337651)
}
}
{
\newrgbcolor{curcolor}{0 0 0}
\pscustom[linestyle=none,fillstyle=solid,fillcolor=curcolor]
{
\newpath
\moveto(280.84138947,85.0285402)
\lineto(278.38707496,89.37962062)
\lineto(283.38679127,89.326969)
\closepath
}
}
{
\newrgbcolor{curcolor}{0 0 0}
\pscustom[linewidth=0.66666583,linecolor=curcolor]
{
\newpath
\moveto(280.84138947,85.0285402)
\lineto(278.38707496,89.37962062)
\lineto(283.38679127,89.326969)
\closepath
}
}
{
\newrgbcolor{curcolor}{0 0 0}
\pscustom[linewidth=0.99999871,linecolor=curcolor]
{
\newpath
\moveto(284.83443402,131.71976694)
\curveto(288.60967181,124.69883717)(290.79358866,116.82811465)(291.17583118,108.86571969)
\curveto(291.59311748,100.17332788)(289.8561411,91.38554079)(286.1631874,83.50557355)
}
}
{
\newrgbcolor{curcolor}{0 0 0}
\pscustom[linestyle=none,fillstyle=solid,fillcolor=curcolor]
{
\newpath
\moveto(284.93890951,80.89322841)
\lineto(284.51053482,85.87038178)
\lineto(289.03799607,83.74858301)
\closepath
}
}
{
\newrgbcolor{curcolor}{0 0 0}
\pscustom[linewidth=0.66666583,linecolor=curcolor]
{
\newpath
\moveto(284.93890951,80.89322841)
\lineto(284.51053482,85.87038178)
\lineto(289.03799607,83.74858301)
\closepath
}
}
{
\newrgbcolor{curcolor}{0 0 0}
\pscustom[linewidth=0.99999871,linecolor=curcolor]
{
\newpath
\moveto(253.78754646,131.71976694)
\curveto(255.58307528,128.06254111)(256.70104819,124.07401701)(257.06798362,120.01636158)
\curveto(257.55057638,114.67974804)(256.7233663,109.22946142)(254.67901984,104.27637544)
}
}
{
\newrgbcolor{curcolor}{0 0 0}
\pscustom[linestyle=none,fillstyle=solid,fillcolor=curcolor]
{
\newpath
\moveto(253.57833013,101.60960105)
\lineto(252.91751454,106.56125565)
\lineto(257.53930725,104.65364783)
\closepath
}
}
{
\newrgbcolor{curcolor}{0 0 0}
\pscustom[linewidth=0.66666583,linecolor=curcolor]
{
\newpath
\moveto(253.57833013,101.60960105)
\lineto(252.91751454,106.56125565)
\lineto(257.53930725,104.65364783)
\closepath
}
\rput[bl](310,0){$\tilde{M}$}
\rput[bl](360,130){$M$}
\rput[bl](110,105){$q=\tilde{F}(q)$}
\rput[bl](120,142){$\0_q=q$}
\rput[bl](235,55){$\tilde{F}(p)$}
\rput[bl](278,145){$p$}
\rput[bl](350,10){$M_p$}
\rput[bl](23,145){$x$}
\rput[bl](15,83){$\tilde{F}(x)$}
\rput[bl](-16,112){$\tilde{F}$}
\rput[bl](327,80){$F_p$}
}
\end{pspicture}
\caption{The osculation map~$F_p$.}
\label{figFp}
\end{figure}%
Working with the osculation maps, the previous formulas for the $\L$-induced charts
and the Riemannian metric can be rewritten more explicitly in components as
\begin{align}
\phi^k_p(x) &= \fint_M \L\big( \tilde{F}(x), F_p(y) \big)\: (y-p)^k\: f_p\: d\rho(y) \label{phichart} \\
g^{jk}(p) &= \frac{1}{\delta^2} \fint_M \L\big( \tilde{F}(p), F_p(y) \big)\: (y-p)^j\, (y-p)^k\: f_p\: d\rho(y) \:.
\label{gchart}
\end{align}

Moreover, working in charts, it becomes possible to differentiate
in a straightforward way with respect to~$p$, as is illustrated in the next lemma.
Here and in what follows, we shall make use of the following {\em{smoothness
and decay assumptions}}: We assume again that the Lagrangian is smooth
and of short range (as specified in Section~\ref{seccvp}) (our results hold similarly
if the Lagrangian has suitable decay properties if~$x$ and~$y$ are far apart,
but for technical simplicity we do not specify these decay assumptions here).
Moreover, we assume that all the mappings~$\tilde{F}(x)$, $\tilde{f}(x)$ as well as~$F_p(x)$
and~$f_p(x)$ are smooth in both~$x$ and~$p$.
\begin{Lemma} The mapping~$\phi_p$ satisfies the relation
\[ \frac{\partial}{\partial p^j} \phi_p(x)^k \Big|_{x=p} = -\delta^k_j \:. \]
\end{Lemma}
\Proof By the above smoothness assumptions and using that the Lagrangian is of short
range, the mapping $(x,p)  \mapsto \phi_p(x)$ with~$\phi_p(x)$ as in~\eqref{phichart}
is smooth. Moreover, we may interchange the derivative with the integral.
Hence we may differentiate with respect to~$p$ under the integral sign.
Using the osculation map, the formulas of Lemma~\ref{lemmadphi} can be written as
\beq \label{twoeq}
\phi^k_p(p) = 0 \qquad \text{and} \qquad
\frac{\partial}{\partial x^j} \phi^k_p(x) \Big|_{x=p} = \delta^j_k \:,
\eeq
where, for simplicity we denote the zero vector in~$M_p$ by~$0$ (rather than~$\0_p$).
This is unproblematic, because all the spaces~$M_p$ are identified with~$M$
via the osculation maps.
Differentiating the first relation in~\eqref{twoeq}, we obtain
with the product and chain rules
\[ 0 = \frac{\partial}{\partial p^j} \phi^k_p(p) = \Big( \frac{\partial}{\partial x^j} + \frac{\partial}{\partial p^j} \Big) \phi^k_p(x) \Big|_{x=p} \:. \]
Combining this with the second equation in~\eqref{twoeq} gives the result.
\QED

\begin{Lemma} \label{lemmagauss}
In an $\L$-induced chart centered at~$q$, the Christoffel symbols
vanish at~$q$.
\end{Lemma}
\Proof We specialize~\eqref{phichart} to~$p=q$ and use~\eqref{xalign},
\[ \phi^k_q(x) = \fint_M \L\big( \tilde{F}(x), y \big)\: (y-q)^k\: d\rho(y) = (x-q)^k \:. \]
Taking second derivatives in~\eqref{Lconn} gives zero.
\QED
In view of this result, the $\L$-induced charts are the analog of Gaussian charts
for the $\nabla^\L$-connection.

In view of our assumption of optimal osculations (see Definition~\ref{defoscoptimal}),
the osculating vacua are tangential to the manifold~$\tilde{M}$.
This gives rise to the following local expansions of the osculation maps.
\begin{Lemma} \label{lemmalocexpand}
Under the assumption of optimal osculations, in the $\L$-induced chart centered at~$q$,
\begin{gather}
\tilde{F}(x) = \1 + \O \big( (x-q)^2 \big) \:,\qquad \tilde{f}(x) = 1 + \O \big( (x-q)^2 \big) 
\label{ex1} \\
\tilde{F}(x) - F_p(x) = \O \big( (x-p)^2 \big) \label{oscFFp} \\
\Big( \frac{\partial}{\partial x^i} + \frac{\partial}{\partial p^i} \Big)
\big( \tilde{F}(x) - F_p(x) \big) = \O \big( (x-p)^2 \big) \label{ex3} \\
\Big( \frac{\partial}{\partial x^i} + \frac{\partial}{\partial p^i} \Big)
\big( \tilde{F}(p) - F_p(x) \big) = \O(x-p) \label{ex4} \:.
\end{gather}
\end{Lemma}

\Proof The fact that~$M$ is tangential to~$\tilde{M}$ at~$q$ immediately implies~\eqref{ex1}.
Similarly, the fact that~$M_p$ is tangential to~$\tilde{M}$ at~$p$
gives~\eqref{oscFFp}.
Finally, we differentiate through~\eqref{oscFFp} to obtain~\eqref{ex3}. Note that the operator 
$\partial_x + \partial_p$ leaves $x-p$ invariant, so that the order in $(x-p)$ is preserved.
In order to prove~\eqref{ex4}, we first use that~$\tilde{F}_p = F_p(p)$. Next, using 
translation invariance on~$M_p$, we know that
\[ \Big( \frac{\partial}{\partial x^i} + \frac{\partial}{\partial y^i} \Big)
\big( \tilde{F}_p(y) - F_p(x) \big) = 0 \:. \]
We thus obtain
\begin{align*}
&\Big( \frac{\partial}{\partial x^i} + \frac{\partial}{\partial p^i} \Big)
\big( \tilde{F}(p) - F_p(x) \big) \\
&= \Big( \frac{\partial}{\partial x^i} + \frac{\partial}{\partial p^i} \Big)
\big( F_p(p) - F_p(x) \big)
= \frac{\partial}{\partial p^i} \big( F_p(y) - F_p(x) \big) \Big|_{y=p}
\end{align*}
Now we can apply~\eqref{oscFFp},
\begin{align*}
&\Big( \frac{\partial}{\partial x^i} + \frac{\partial}{\partial p^i} \Big)
\big( \tilde{F}(p) - F_p(x) \big) \\
&= \frac{\partial}{\partial p^i} \Big( \tilde{F}(y) + \O\big((y-p)^2 \big) - \tilde{F}(x)
+ \O\big((x-p)^2 \big) \Big) \Big|_{y=p} = \O(x-p)\:,
\end{align*}
concluding the proof.
\QED
\subsection{The Curvature of~$\nabla^\L$} \label{seccurvL}
We now compute the corresponding curvature tensor.
\begin{Lemma} The curvature tensor takes the form
\beq \label{RL}
R^k_{ijl}(q) = \Big( \frac{\partial}{ \partial p^i} \frac{\nabla^2}{\partial x^j \partial x^l}
\phi^k_p(x) - \frac{\partial}{ \partial p^j} \frac{\nabla^2}{\partial x^i \partial x^l} \phi^k_p(x) \Big)
\Big|_{x=p=q} \:,
\eeq
where~$\nabla^2$ refers to the connection~$\nabla^\L$.
\end{Lemma}
\Proof Following the standard notation and conventions,
\begin{align*}
R^k_{ijl} u^l &= \nabla_i \nabla_j u^k - \nabla_j \nabla_i u^k \\
&= \Big( \partial_i \partial_j u^k + (\partial_i \Gamma^k_{jl}) \,u^l 
+ \Gamma^k_{jl} \,(\partial_i u^l) + \Gamma^k_{il} \partial_j u^l
+ \Gamma^k_{ia} \Gamma^a_{jl} \,u^l \Big)  - (i \leftrightarrow j) \\
&= \Big( (\partial_i \Gamma^k_{jl}) \,u^l 
+ \Gamma^k_{ia} \Gamma^a_{jl} \,u^l \Big)  - (i \leftrightarrow j)
\end{align*}
and thus
\beq \label{Riem}
R^k_{ijl} = \partial_i \Gamma^k_{jl} - \partial_j \Gamma^k_{il}
+ \Gamma^k_{ia} \Gamma^a_{jl} - \Gamma^k_{ja} \Gamma^a_{il}  \:.
\eeq

For convenience, we compute the curvature tensor in the $\L$-induced
chart centered at~$q=p$. Then the undifferentiated Christoffel symbols vanish
according to Lemma~\ref{lemmagauss}. The derivatives of the Christoffel
symbols, on the other hand, can be computed by differentiating~\eqref{Lconn}
with respect to~$x=p$; i.e., again with the product and chain rules,
\beq \label{pchain}
\partial_i \Gamma^k_{jl} = \frac{\partial}{\partial p^i} \bigg(
\frac{\partial^2}{\partial x^j \partial x^l} \phi^k_p(x) \Big|_{x=p} \bigg)
= \Big( \frac{\partial}{\partial x^i} + 
\frac{\partial}{\partial p^i} \Big)
\frac{\partial^2}{\partial x^j \partial x^l} \phi^k_p(x) \bigg|_{x=p} \:.
\eeq
A straightforward computation gives the result.
\QED

Since the $\L$-induced connection is torsion free (see~\eqref{Lconn}), its curvature tensor satisfies the usual first and second Bianchi identities (see for example~\cite[Chapter~7]{nakahara}). It is instructive to verify this fact by direct computation using our formulas for the
Christoffel symbol:
\begin{Lemma} The curvature tensor satisfies the Bianchi identities
\begin{align}
R^k_{[ijl]} &= 0 \label{bianchi1} \\
\nabla^\L_{[m} R^k_{ij]\:l} &= 0 \:. \label{bianchi2}
\end{align}
\end{Lemma}
\Proof The first Bianchi identities~\eqref{bianchi1} are verified immediately
by anti-symmetrizing the formula~\eqref{RL} and using that the connection~$\nabla^\L$
is torsion-free. In order to prove~\eqref{bianchi2},
we differentiate~\eqref{RL},
\begin{align*}
\nabla^\L_m R^k_{ijl}(q) &= \Big( \frac{\nabla}{\partial x^m} + \frac{\nabla}{\partial p^m} \Big)
\Big( \frac{\partial}{ \partial p^i} \frac{\nabla^2}{\partial x^j \partial x^l}
\phi^k_p(x) - \frac{\partial}{ \partial p^j} \frac{\nabla^2}{\partial x^i \partial x^l} \phi^k_p(x) \Big)
\Big|_{x=p=q} \:.
\end{align*}
Multiplying out and totally anti-symmetrizing in the indices~$m$, $i$ and~$j$ gives
\begin{align*}
\nabla^\L_{[m} R^k_{ij]\:l}(q) &= 2\:\Big( \frac{\nabla}{\partial x^{[m}} \:\frac{\partial}{ \partial p^i} \frac{\nabla^2}{\partial x^{j]} \partial x^l}
\phi^k_p(x)  \Big) \Big|_{x=p=q} \\
&= R^k_{mja}\: \Big( \frac{\partial}{ \partial p^i} \frac{\partial}{\partial x^l} \phi^a_p(x)\Big) \Big|_{x=p=q}
- R^k_{mia}\: \Big( \frac{\partial}{ \partial p^j} \frac{\partial}{\partial x^l} \phi^a_p(x)\Big) \Big|_{x=p=q} \\
&\quad\:-R^a_{mjl}\: \Big( \frac{\partial}{ \partial p^i} \frac{\partial}{\partial x^a} \phi^k_p(x)\Big) \Big|_{x=p=q}
+R^a_{mil}\: \Big( \frac{\partial}{ \partial p^j} \frac{\partial}{\partial x^a} \phi^k_p(x)\Big) \Big|_{x=p=q} \:.
\end{align*}
This expression is zero, as is verified most easily as follows.
Using~\eqref{delta}, we know that
\[ \Big( \frac{\partial}{\partial x^i} + \frac{\partial}{\partial p^i} \Big)\: \frac{\partial}{\partial x^l}
\phi^k_p(x) \Big|_{x=p} = 0 \:. \]
Hence, using Lemma~\ref{lemmagauss},
\[ \frac{\partial}{\partial p^i}\frac{\partial}{\partial x^l}
\phi^k_p(x) \Big|_{x=p} = - 
\frac{\partial^2}{\partial x^i \partial x^l}
\phi^k_p(x) \Big|_{x=p} = 0 \:. \]
This concludes the proof.
\QED

Finally, we form the {\em{Ricci}} tensor by contracting the first and third index,
\beq \label{Ricdef}
R_{il}(q) := R^k_{ikl}(q) =
\Big( \frac{\partial}{\partial p^i} \frac{\nabla^2}{\partial x^l \partial x^k}  \phi^k_p(x) -
 \frac{\nabla^2}{ \partial x^i \partial x^l} \frac{\partial}{\partial p^k} \phi^k_p(x) \Big)
\Big|_{x=p=q} \:.
\eeq
Note that, in general this Ricci tensor is {\em{not}} symmetric
in its two tensor indices. Instead, in the $\L$-induced chart centered at~$p$
we have the relation
\[ 
R_{il}(q) - R_{li}(q) = \Big( \frac{\partial}{ \partial p^i} \frac{\partial^2}{\partial x^k \partial x^l} \phi^k_p(x) 
- \frac{\partial}{ \partial p^l} \frac{\partial^2}{\partial x^k \partial x^i} \phi^k_p(x) \Big)
\bigg|_{x=p=q}
\:. \]
Moreover, contracting the Bianchi identities~\eqref{bianchi2}, we obtain
\begin{align*}
0 &= \nabla_m R^k_{ikl} + \nabla_i R^k_{kml} + \nabla_k R^k_{mil} \\
&= \nabla_m R_{il} - \nabla_i R_{ml} + \nabla_k R^k_{mil} \:.
\end{align*}
contracting the indices~$l$ and~$m$ (with respect to the Riemannian
metric~$g$), we obtain
\beq \label{bianchiL}
0 = \nabla^l R_{il} - \nabla_i R + \nabla_k R^k_{mil}\: g^{ml} \:.
\eeq
This is quite different from the corresponding formulas in Riemannian and Lorentzian geometry. We will come back to this point at the beginning of Section~\ref{secriemann}.

\subsection{Is $\nabla^\L$ a Metric Connection?}
Differentiating~\eqref{gchart}, we obtain in the $\L$-induced chart centered at~$q$,
\beq \label{derg}
\partial_i g^{jk}(q) = \frac{1}{\delta^2}\:
\frac{\partial}{\partial x^i} \fint_M \L\big( \tilde{F}(x), F_x(y) \big)\: \xi^j \xi^k\: f_x\: d\rho(y)
\bigg|_{x=q} \:,
\eeq
where we introduced the abbreviation
\beq \label{xidef}
\xi := y - x
\eeq
There is no reason why this expression should vanish. But, at least,
it tends to zero in the limit~$\delta \searrow 0$ when the
range of the Lagrangian tends to zero, as is made precise in the next lemma.

In preparation, we need to make mathematically
precise what we mean by an {\em{expansion for small~$\delta$}}.
Recall that the parameter~$\delta$ was introduced in Section~\ref{seccvp} as the
range of the Lagrangian (see~\eqref{shortrange}). This means that, if we decrease~$\delta$,
we need to change the Lagrangian (for example by a rescaling). When doing so,
one also needs to modify the corresponding minimizing measures~$\tilde{\rho}$ and~$\rho$.
This also means that the functions~$\tilde{F}$, $f$ in~\eqref{rhotilderho}
as well as the osculation maps~$F_p$ and~$f_p$ will depend on~$\delta$.
In this general situation, where all quantities depend on~$\delta$, the limit~$\delta \searrow 0$
cannot be studied in a sensible way. Also, this is not the limiting case we have in mind.
Instead, we are interested in the situation where the deviation of~$\tilde{\rho}$ from~$\rho$
as expressed by~$\tilde{F}$ and the osculation maps~$F_p$ (see Figure~\ref{figFp})
takes place on a length scale~$\ell_\text{macro}$ which is kept fixed and is much larger than the range~$\delta$ of the Lagrangian. In physical terms, we can think~$\ell_\text{macro}$
as the length scale of macroscopic physics (like the Compton scale), whereas~$\delta$
is the Planck length. In view of the fact that the Planck length is by many orders of magnitude
smaller than the Compton scale, it seems a good idea to analyze the limiting
case~$\delta \ll \ell_{\text{macro}}$. This can be implemented in a mathematically precise
way by assuming that the functions~$\tilde{F}$, $\tilde{f}$, $F_p$ and~$f_p$
do not depend on~$\delta$. But~$\delta$ comes into play in integral expressions
(like~\eqref{phichart} and~\eqref{gchart}) because the range of the Lagrangian
determines the size of the factors~$(y-p) \sim \delta$.
In this way, we can also work with error terms like~$h(q) = \O(\delta^2)$, meaning
that for every~$q \in \tilde{M}$ and in a given chart,
the expression~$h(q)$ vanishes quadratically as~$\delta \searrow 0$.
Due to the assumed smoothness of the involved functions~$\tilde{F}$, $\tilde{f}$, $F_p$ and~$f_p$,
this statement automatically holds locally uniformly in~$q$, in the sense that for every chart and
every compact set~$K$ contained in the domain of this chart, there is a constant~$C(K)>0$
such that
\[ \limsup_{\delta \searrow 0} \sup_{q \in K} \frac{|h(q)|}{\delta^2} < C(K)  \:. \]
For simplicity, we shall not emphasize the local uniformity but prefer to work with pointwise 
statements.

Expanding in this sense in powers of~$\delta$, we have the following result.
\begin{Lemma} \label{lemmametric}
The connection~$\nabla^\L$ is almost metric in the sense that
\[ 
\nabla^\L g =  \O \big( \delta^2 \big) \:. \]
\end{Lemma}
\Proof Since the Christoffel symbols of the
connection~$\nabla^\L$ vanish at~$q$
in the $\L$-induced chart centered at~$q$, it suffices to consider the
partial derivatives~\eqref{derg}.
We interchange differentiation and integration and add a $y$-derivative
(which vanishes after integration-by-parts),
\begin{align}
\partial_i g^{jk}(q) &= \frac{1}{\delta^2} 
\fint_M \Big( \frac{\partial}{\partial x^i} + \frac{\partial}{\partial y^i} \Big)
\Big(\L\big( \tilde{F}(x), F_x(y) \big)\: \xi^j \xi^k \Big)\: f_x\:d\rho(y) \Big|_{x=q} \notag \\
&= \frac{1}{\delta^2} 
\fint_M \Big( \frac{\partial}{\partial x^i} + \frac{\partial}{\partial y^i} \Big)
\Big(\L\big( \tilde{F}(x), F_x(y) \big)\, f_x \Big)\: \xi^j \xi^k\: d\rho(y) \Big|_{x=q} \:. \label{xydiff}
\end{align}
We now compute the derivatives with the chain rule.
The function~$f_x$ and its derivative can be left out in view of~$f_x=\tilde{f}(x)$ and
the osculation equation~\eqref{ex1}. We thus obtain
\begin{align*}
\Big( \frac{\partial}{\partial x^i} &+ \frac{\partial}{\partial y^i} \Big)
\Big(\L\big( \tilde{F}(x), F_x(y)\big)\, f_x \Big) \Big|_{x=q}\\
&= D_1 \L|_{(q, y)} \,\partial_i \tilde{F}(q) + D_2 \L|_{(q, y)} 
\Big( \frac{\partial}{\partial x^i} + \frac{\partial}{\partial y^i} \Big)
F_x(y) \big|_{x=q} \:.
\end{align*}
Using that the unperturbed Lagrangian is translation invariant, we obtain
\[ \Big( \frac{\partial}{\partial x^i} + \frac{\partial}{\partial y^i} \Big)
\Big(\L\big( \tilde{F}(x), F_x(y) \big) \Big)\bigg|_{x=q}
= D_1 \L|_{(q, y)} 
\bigg( \Big( \frac{\partial}{\partial x^i} + \frac{\partial}{\partial y^i} \Big)
\Big( \tilde{F}(x) - F_x(y) \bigg) \bigg|_{x=q}\:. \]
We now use the osculation equation~\eqref{ex4} to conclude that the last expression 
is of the order~$\O(y-q)$. Consequently, the integral in~\eqref{xydiff}
is of the order~$\O(\delta^3)$. Using that odd powers in~$y-q$ vanish in the integral
by symmetry, this integral is even of the order~$\O(\delta^4)$.
This concludes the proof.
\QED

\subsection{Relation Between the Curvatures of $\nabla^\L$ and $\nabla^g$}
We can interpret the result of Lemma~\ref{lemmametric} by saying that the
connection~$\nabla^\L$ is approximately metric, with an error of order~$\delta^2$.
Consequently, also the metric curvature of~$g$ agrees approximately
with the curvature of~$\nabla^\L$. We now work out more systematically
how these curvatures are related to each other.

Clearly, the Riemannian metric~$g$ gives us a Levi-Civita connection~$\nabla^g$.
Since the connection~$\nabla^\L$ is only approximately metric, we need to carefully distinguish
between the connections~$\nabla^g$ and~$\nabla^\L$. Their difference
defines a tensor field of order~$(1,2)$, which we refer as the
{\em{deviation tensor}} and denote by
\beq \label{devdef}
K^j_{ik} := (\Gamma^g)^j_{ik} - \Gamma^j_{ik} \:,
\eeq
where~$\Gamma$ denotes the Christoffel symbols of the
connection~$\nabla^\L$.
This tensor takes a particularly simple form in $\L$-induced charts:
\begin{Lemma} In an $\L$-induced chart centered at~$q$,
the deviation tensor takes the form
\begin{align}
K^j_{ik}(q) &= -\nabla^g_{ik} \phi^j_q(x) \Big|_{x=q} \label{K1} \\
&= -\fint_M \frac{\nabla^g}{\partial x^i} \frac{\nabla^g}{\partial x^k}
\L\big( \tilde{F}(x), F_q(y) \big)\: (y-q)^k\: d\rho(y) \Big|_{x=q} \:.
\end{align}
\end{Lemma}
\Proof We compute the deviation tensor in
a Gaussian coordinate system centered at~$q$ with respect to the metric~$g$.
Then the first summand in~\eqref{devdef} vanishes, so that
\[ K^j_{ik}(q) = - \Gamma^j_{ik}(q) \:. \]
Using~\eqref{Lconn} gives~\eqref{K1}. Since both sides of this equation
are tensorial, we conclude that~\eqref{K1} holds in any coordinate system.
Choosing the $\L$-induced chart centered at~$q$, we can use~\eqref{phichart}
to obtain the result.
\QED

As an immediate consequence of Lemma~\ref{lemmametric}, we
conclude that the deviation tensor is very small in the following sense.
\begin{Corollary} \label{corK}
The deviation tensor is almost zero in the sense that
\[ K^j_{ik}(q) = \O\big( \delta^2 \big) \:. \]
\end{Corollary}

We finally relate the curvatures.
\begin{Lemma} \label{lemmaRicrel1}
The Riemann and Ricci tensors are related to each other by
\begin{align}
(R^g)^k_{ijl} &= R^k_{ijl}
- \nabla^g_i K^k_{jl} + \nabla^g_j K^k_{il}
- K^k_{ia} K^a_{jl} + K^k_{ja} K^a_{il} \label{R1} \\
R^g_{il} &= R_{il}
- \nabla^g_i K^a_{al} + \nabla^g_a K^a_{il}
- K^b_{ia} K^a_{bl} + K^b_{ba} K^a_{il} \:. \label{R2}
\end{align}
\end{Lemma}
\Proof By definition of the deviation tensor,
\[ (\Gamma^g)^i_{jk} = \Gamma^i_{jk} + K^i_{jk} \:. \]
A straightforward computation using~\eqref{Riem} gives the result.
\QED

\section{Incorporating the Euler-Lagrange Equations} \label{secEL}
Our goal is to compute the Ricci tensor~$R_{jl}(q)$ in
the $\L$-induced chart centered at~$q$. For a compact notation, we use the abbreviations
\beq \label{parbar}
\partial_i = \frac{\partial}{\partial x^i} \qquad \text{and} \qquad
\overline{\partial}_i := \frac{\partial}{\partial x^i} + \frac{\partial}{\partial p^i} \:.
\eeq
Then we can write~\eqref{Ricdef} as
\beq \label{Ricshort}
R_{il}(q) = \big( \overline{\partial}_i \partial_{kl} \phi^k_p(x) - 
\overline{\partial}_k \partial_{il} \phi^k_p(x) \big) \big|_{x=p=q} \:.
\eeq
The key is to employ the EL equations of the
causal action principle~\eqref{EL1}, which hold for both~$\rho$
and~$\tilde{\rho}$.
In order to explain how this can be done, we proceed step by step.
We begin combining the osculation equations with the EL equations
(Section~\ref{secoscEL}). This makes it possible to rewrite the
curvature using so-called alignments (Section~\ref{secalign}).
The next crucial step is to make use of the fact that the index~$k$
in~\eqref{Ricshort} is contracted with a derivative to what we
call a divergence term. In order to exploit this structure,
we develop a method for expanding
divergences in powers of the microscopic length scale~$\delta$.
This method will be introduced for an expansion of the weight function~$\tilde{f}$
(Section~\ref{secf}), and it is then used for the desired expansion of the
divergences in the Ricci tensor (Section~\ref{secdiv}).

\subsection{Employing the Osculation Equations in the Euler-Lagrange Equations}
In Lemma~\ref{lemmalocexpand}, the osculation equations were
formulated in terms of a local expansion of the osculation maps.
The fact that~$M_p$ is tangential to~$\tilde{M}$ becomes apparent
in the fact that the expansion terms are quadratic.

In order to explain how to use this fact in the EL equations,
we evaluate the function~$\tilde{\ell}$ on~$M_p$. Thus, using the
osculation maps, we consider the function
\begin{align*}
\tilde{\ell}\big( F_p(x) \big) &= \int_{\tilde{M}} \L\big( F_p(x), y \big)\: d\tilde{\rho}(y) - \s \\
&= \int_M \L\big( F_p(x), \tilde{F}(y) \big)\: \tilde{f}(y)\: d\rho(y) - \s \:.
\end{align*}
We now expand~$x$ locally near~$p$. The first order vanishes
immediately in view of the EL equations~\eqref{EL1} for~$\tilde{\rho}$.
The point is that, due to the osculation equations, also the second derivatives vanish. 
In order to see how this comes about, we rewrite the Lagrangian as
\[ \L\big( F_p(x), \tilde{F}(y) \big) = \L\Big( \tilde{F}(x) + \big(F_p(x) - \tilde{F}(x) \big), \tilde{F}(y) \Big) \:. \]
Employing~\eqref{oscFFp}, we obtain
\begin{align*}
\L\big( F_p(x), \tilde{F}(y) \big) = \L\big( \tilde{F}(x), \tilde{F}(y) \big)
+ D_1 \L\big( \tilde{F}(x), \tilde{F}(y) \big)\: \O\big( \|x-p\|^2 \big) + \O\big( \|x-p\|^4 \big) \:.
\end{align*}
Hence
\begin{align*}
\tilde{\ell}\big( F_p(x) \big) = \tilde{\ell}\big( \tilde{F}(x) \big) 
+ D \tilde{\ell}|_{\tilde{F}(x)}\: \O\big( \|x-p\|^2 \big) + \O\big( \|x-p\|^4 \big)\:.
\end{align*}
The first summand vanishes for all~$x$ in view of the EL equations~\eqref{EL1}.
In the second summand, the term~$D \tilde{\ell}|_{\tilde{F}(x)}$ vanishes at~$x=p$.
We thus conclude that the second derivatives~$D^2 \tilde{\ell}\big( F_p(x) \big)$
vanish at~$x=p$.

This method works similarly if we evaluate~$\ell_p$ on~$\tilde{M}$.
Namely, considering the function
\begin{align*}
\ell_p\big(\tilde{F}(x) \big) &= \int_{M_p} \L\big( \tilde{F}(x),y \big)\: d\rho_p(y) \\
&= \int_{M} \L\big( \tilde{F}(x),F_p(y) \big)\: f_p\: d\rho(y)
\end{align*}
and rewriting the Lagrangian as
\[ \L\big( \tilde{F}(x),F_p(y) \big) = \L\Big( F_p(x) + \big(\tilde{F}(x) - F_p(x) \big) ,F_p(y) \Big) \:, \]
we can again use~\eqref{oscFFp} as well as the EL equations~\eqref{EL1}
for~$\rho_p$.

We summarize our findings as follows.
\begin{Lemma} \label{lemmaoscEL} Assume that both~$\rho$
and~$\tilde{\rho}$ are critical points of the causal action principle.
Moreover, assume that
the osculation equations of Lemma~\ref{lemmalocexpand} hold.
Then, in the $\L$-induced chart centered at~$q$, the following equations
hold for all~$p \in U$,
\[ D^2 \tilde{\ell}\big( F_p(x) \big)\big|_{x=p} = 0
= D^2 \ell_p\big( \tilde{F}(x) \big) \big|_{x=p} \:. \]
\end{Lemma}

\subsection{Formulation of Curvature with Alignments} \label{secalign}
Alignments were first introduced in~\cite[Section~5]{matter} in the
linearized setting as objects useful for describing the total mass and a quasi-local
mass for causal variational principles. We now introduce corresponding
objects in the fully non-linear setting (for a comparison of the notions
see Remark~\ref{remalignment}).
\begin{Def} \label{defalign}
Working in $\L$-induced chart centered at~$q$ in the formulation with the osculation map~$F_p$ (for details see Section~\ref{secchart}), the
{\bf{alignment vector field}}~$A_p$ is defined by
\beq \label{Apdef}
A^k_p(x) := \fint_M \L\big( \tilde{F}(x), F_p(y) \big)\: \xi^k\: f_p\: d\rho(y)
\eeq
(where~$\xi$ is again the difference vector defined by~\eqref{xidef}).
\end{Def} \noindent
Note that the structure of the alignment vector field~$A_p$ is very similar to that of the mapping~$\phi_p$
in~\eqref{phichart}. The only difference is that we replace the vector~$y-p$ in the integrand
by the difference vector~$\xi = y-x$. In particular, one sees immediately that
\[ A^k_p(p) = \phi^k_p(p) \:. \]
However, replacing~$y-p$ by~$y-x$ has major consequences.
Before elaborating on these consequences, we point out
that the alignment vector field can{\em{not}} be written in a coordinate independent
form similar to~\eqref{phiint}. The reason is that~$\tilde{x}$ in~\eqref{phiint} is a
point of~$\tilde{M}$, so that writing the difference~$y-\tilde{x}$ in~\eqref{phiint} is mathematically
not sensible. Only after choosing a chart and working with the osculation
map~\eqref{phichart}, the difference vector~$y-x \in M$
is well-defined. Keeping in mind that, having chosen the base point~$q$,
the corresponding $\L$-induced chart and the osculation maps are canonical (see Lemma~\ref{lemmaFp}).
Consequently, the base point~$q$ also uniquely determines the alignment vector field~$A_p(x)$.

We next point out that the factor~$\xi$ in the integrand in~\eqref{Apdef} changes the behavior of the integral considerably. Namely, noting that the Lagrangian has range~$\delta$,
the factor~$(y-x)$ gives a scaling factor~$\delta$.
Anticipating the results of the computations performed in detail in Section~\ref{secdiv},
one can say that alignments have a better scaling behavior
in~$\delta$ and can be used to show that certain expressions are very small
in the sense that they vanish in the limit~$\delta \searrow 0$. With this in mind,
it is a major step forward to replace the vector field~$\phi_p$ in the Riemann
tensor~\eqref{RL} by corresponding alignment vector field. This can indeed be done, making essential use of the EL equations of the osculation.
\begin{Lemma} \label{lemmaalign}
In the formula for the curvature tensor~\eqref{RL} we may replace
the factors~$\phi^k_p$ by~$A^k_p$. Thus,
in the $\L$-induced chart centered at~$q$,
\beq \label{RLtilde}
R^k_{ijl}(q) = 
\big( \overline{\partial}_i \partial_{jl} - \overline{\partial}_j \partial_{il} \big)
A^k_p(x) \Big|_{x=p=q} \:.
\eeq
\end{Lemma}
\Proof We first note that, in our chart,
\[ A^k_p(x) = \phi^k_p(x) - x^k\, \ell_p\big( \tilde{F}(x) \big) \:. \]
Therefore, we need to expand the last summand,
Applying the EL equations~\eqref{EL1} for~$\rho_p$ as well as
Lemma~\ref{lemmaoscEL}, it follows that
\[ \partial_i \ell_p\big(\tilde{F}(x) \big) \big|_{x=p} = 0 = \partial_{ij} \ell_p\big( \tilde{F}(x) \big)
\big|_{x=p}
\qquad \text{for all~$p \in \tilde{M}$}\:. \]
Since these equations hold at every spacetime point, we can also differentiate them
and apply the product and chain rules as in~\eqref{pchain}. Thus,
again using the short notations~\eqref{parbar},
\[ \overline{\partial}_l \partial_i \ell_p\big( \tilde{F}(x) \big) \big|_{x=p} = 0 = \overline{\partial}_l \partial_{ij} \ell_p \big( \tilde{F}(x) \big) \big|_{x=p} = 0 \:. \]
Using again the chain rule, it follows that
\[ 0 = \overline{\partial}_l \partial_{ij} \Big( x^k \ell_p\big( \tilde{F}(x) \big) \Big)
= \overline{\partial}_l \partial_{ij} \Big( \phi^k_p\big( \tilde{F}(x) \big) - A^k_p(x) \Big) \:. \]
Therefore, in~\eqref{RL} we may replace all factors~$\phi_p$ by~$A_p$,
concluding the proof.
\QED

\subsection{Expansion of the Weight Function~$\tilde{f}$ in Powers of~$\delta$} \label{secf}
We now derive a few properties of the weight function~$\tilde{f}$ in~\eqref{rhotilderho}.
To this end, we will expand~$\tilde{f}$ in powers of the regularization length~$\delta$.
The reason why we present this expansion here is that it is a preparation
for similar expansions of the Ricci tensor to be performed in Section~\ref{secdiv}.
Our general method works for expressions in divergence form.
In order to introduce and explain the method, it is easiest to begin by
taking the divergence of the relation~\eqref{xalign}
which characterizes the $\L$-induced chart centered at~$q$.
We thus obtain
\begin{align}
0 &= \frac{\partial}{\partial x^k} \bigg( \int_M \L\big( \tilde{F}(x), y \big)\:y^k\: d\rho(y) - \s x^k \bigg) \notag \\
&= \frac{\partial}{\partial x^k} \int_M \L\big( \tilde{F}(x), y \big)\:\xi^k\: d\rho(y)
+ \frac{\partial}{\partial x^k} \Big( x^k\: \ell\big( \tilde{F}(x) \big) \Big) \label{divterm}
\end{align}
(where we again set~$\xi=y-x$; see~\eqref{xidef}).
In order to relate this formula to the weight function~$\tilde{f}$, given a parameter~$s \in [0,1]$
we introduce the variables
\Christoph{Könnte man die Notation noch vereinfachen? Die Formeln sind manchmal grenzwertig 
überladen. F: OK, aber wie?}
\beq \label{xysdef}
x_s = \zeta - (1-s) \xi \:,\qquad y_s = \zeta + s \xi
\eeq
and introduce the function
\[ \varphi(s, \zeta) := \int_M \tilde{f}(x_s)\: \L \big( \tilde{F}(x_s), y_s \big)\: d\rho(\xi) \:. \]
The point is that, taking the $s$-derivative at~$s=1$ with the chain rule
gives precisely the first summand in~\eqref{divterm},
\begin{align}
\frac{d}{ds} \varphi(s,\zeta) \big|_{s=1}
&= \int_M \xi^k \frac{\partial}{\partial x^k} \Big( \tilde{f}(x)\: \L \big( \tilde{F}(x), x+\xi \big) \Big)\: d\rho(\xi) \label{fo0} \\
&= \int_M \xi^k \:\Big( \frac{\partial}{\partial x^k} + \frac{\partial}{\partial y^k} \Big) \Big( \tilde{f}(x)\: \L \big( \tilde{F}(x), y \big) \Big)\: d\rho(y) \notag \\
&= \int_M \Big( \frac{\partial}{\partial x^k} + \frac{\partial}{\partial y^k} \Big) \Big( \tilde{f}(x)\: \L \big( \tilde{F}(x), y \big)\:\xi^k  \Big)\: d\rho(y) \notag \\
&= \frac{\partial}{\partial x^k} \int_M \tilde{f}(x)\: \L \big( \tilde{F}(x), y \big)\:\xi^k\: d\rho(y) \label{firstorder}
\end{align}
(in the last line, the $y$-derivative was integrated by parts). On the other hand,
evaluating the function~$\phi$ at~$s=0$ and~$s=1$, we obtain expressions which
can be rewritten in terms of the function~$\ell$ and~$\tilde{\ell}$,
\begin{align}
\varphi(1, \zeta) &= \int_M \tilde{f}(\zeta)\: \L \big( \tilde{F}(\zeta), \zeta+\xi \big)\: d\rho(\xi) \notag \\
&= \int_M \tilde{f}(\zeta)\: \L \big( \tilde{F}(\zeta), y \big)\: d\rho(y) = \tilde{f}(\zeta)\: \Big( \ell\big(\tilde{F}(\zeta) \big) + \s \Big) \label{finfo} \\
\varphi(0, \zeta) &= \int_M \tilde{f}(\zeta-\xi)\: \L \big( \tilde{F}(\zeta-\xi), \zeta \big)\: d\rho(\xi) \notag \\
&= \int_M \tilde{f}(y)\: \L \big( \tilde{F}(y), \zeta \big)\: d\rho(y)
= \int_M\L \big( \zeta, \tilde{F}(y) \big)\: \tilde{f}(y)\:  d\rho(y) = \tilde{\ell}(\zeta) + \s \label{finfo2}
\end{align}
(in the last line we used the symmetry of the Lagrangian).
The fact that we get the functions~$\ell$ and~$\tilde{\ell}$ makes it possible to
employ the EL equations. In this way, we obtain information on the function~$\tilde{f}(\zeta)$
in~\eqref{finfo}.
In order to complete the argument, we expand in a Taylor polynomial about~$s=1$ with integral 
remainder. For any~$k \in \N$,
\beq \label{phiex}
\varphi(1,\zeta) - \varphi(0,\zeta) = -\sum_{r=1}^k \frac{(-1)^r}{r!}\: \frac{d^r}{ds^r} \varphi(s,
\zeta)\big|_{s=1} - R_k(x) \:,
\eeq
where
\begin{align}\label{phiremainer}
  R_k(\zeta) := \frac{(-1)^k}{k!} \int_0^1 s^k\: \frac{d^{k+1}}{ds^{k+1}} \varphi(s,\zeta)\: ds \:.
\end{align}
Note that convergence of the Taylor series (i.e.\ the real analyticity of~$\varphi$) is not needed, 
because we are only interested in an expansion in powers of~$\delta$ with polynomial
error term.

The left side was computed in~\eqref{finfo} and~\eqref{finfo2},
whereas the first Taylor summand is the divergence term~\eqref{firstorder}.
In the next lemma, we summarize these results and show that the
higher orders of the Taylor series give an expansion in powers of
the range~$\delta$ of the Lagrangian.
\begin{Lemma} \label{lemmaf}
The weight function~$\tilde{f}$ in~\eqref{rhotilderho} has for any $N \in \N$ the expansion
\begin{align}
\tilde{f}(x) -1 &= - \frac{1}{\s}\: \Big( \tilde{f}(x)\, \ell \big( \tilde{F}(x) \big) + \tilde{\ell}(x) \Big)
-\frac{1}{\s}\:\frac{\partial}{\partial x^k} \Big( x^k\: \ell\big( \tilde{F}(x) \big) \Big) \label{f1} \\
&\quad\: -\sum_{r=2}^N \frac{(-1)^r}{r!} \int_M \Big( \xi^k \frac{\partial}{\partial x^k} \Big)^r \tilde{f}(x)\: \L \big( \tilde{F}(x), x+\xi \big)\: d\rho(\xi) \:. \label{f2} \\
&\quad\: - R_N(\zeta) \: ,
\end{align}
where the remainder is given by
\begin{align}\label{f3}
  R_N(\zeta) = \dfrac{(-1)^N}{N!}\int_0^1 s^N \int_M \Big(\xi^k \dfrac{\partial}{\partial \zeta^k}
  \Big)^{N+1}\tilde{f}(x_s)\L(\tilde{F}(x_s), x_s + \xi)\,d\rho(\xi)\,ds \: .
\end{align}
Moreover, the function~$\tilde{f}$ has the local expansion
\beq \label{floc}
\tilde{f}(x)= 1 + \O \big( (x-q)^3 \big) + \O \big( \delta^2 \big) \:.
\eeq
\end{Lemma}
\Proof Computing the higher $s$-derivatives to the $N^\text{th}$ order in~\eqref{phiex} similar to~\eqref{fo0} and combining the remainders in $R_N$, we obtain
\begin{align}
&\varphi(1,\zeta) - \varphi(0,\zeta) = -\sum_{r=1}^N \frac{(-1)^r}{r!} \int_M \Big( \xi^k\: \frac
{\partial}{\partial x^k} \Big)^r \tilde{f}(x)\: \L \big( \tilde{F}(x), x+\xi \big)\: d\rho(\xi) 
+ R_N(\zeta) \notag \\
&= \frac{\partial}{\partial x^k} \int_M
\tilde{f}(x)\: \L \big( \tilde{F}(x), y \big)\: (y-x)^k\: d\rho(y) \notag \\
&\quad\: -\sum_{r=2}^N \frac{(-1)^r}{r!} \int_M \Big( \xi^k\: \frac{\partial}{\partial x^k} \Big)^r
\tilde{f}(x)\: \L \big( \tilde{F}(x), x+\xi \big)\: d\rho(\xi) + R_N(\zeta) \:. \notag
\end{align}
Now we can employ~\eqref{divterm}. Combining all the terms yields
\begin{align*}
&\s \big( \tilde{f}(x) - 1) + \tilde{f}(x)\, \ell \big( \tilde{F}(x) \big) - \tilde{\ell}(x) = -\frac{\partial}{\partial x^k} \Big( x^k\: \ell\big( \tilde{F}(x) \big) \Big) \\
&\quad\: -\sum_{r=2}^N \frac{(-1)^r}{r!} \int_M \Big( \xi^k\: \frac{\partial}{\partial x^k} \Big)^r
\tilde{f}(x)\: \L \big( \tilde{F}(x), x+\xi \big)\: d\rho(\xi) + R_N(\zeta) \:,
\end{align*}
proving~\eqref{f1} and~\eqref{f2}.

Next, we derive the explicit form of the remainder $R_N(\zeta)$. By looking at the definitions for $x_s$ and $y_s$ in~\eqref{xysdef} we notice
\[
  \dfrac{d}{ds}x_s = \dfrac{d}{ds}y_s = \xi \quad \text{and} \quad \dfrac{\partial}{\partial \zeta} x_s = \dfrac{\partial}{\partial \zeta} y_s = 1 \: .
\]
Therefore, we can write
\begin{align}
  \dfrac{d}{ds}\varphi(s, \zeta) &= \int_M \xi^k(\partial_{x^k} + \partial_{y^k}) \tilde{f}(x_s)\L(\tilde{F}(x_s), y_s)d\rho(\xi) \notag \\
  &= \int_M \Big(\xi^k\dfrac{\partial}{\partial \zeta^k}\Big) \tilde{f}(x_s)\L(\tilde{F}(x_s), x_s + \xi)d\rho(\xi)
\end{align}
Using~\eqref{phiremainer} gives~\eqref{f3}.

In order to derive~\eqref{floc}, we Taylor expand~\eqref{f1} and~\eqref{f2}
about~$q$. The second derivative of~$\ell$ and~$\tilde{\ell}$ vanish in view
of Lemma~\ref{lemmaoscEL} (specialized for~$p=q$, in which case~$\ell_p=\ell$
and~$F_p = \tilde{F}$).
As a consequence, the right side of~\eqref{f1} is of
the order~$\O((x-q)^3)$. The summands in~\eqref{f2} and~\eqref{f3} all contain at least two
factors of~$\xi$. These factors of~$\xi$ remain if we differentiate with respect
to~$x$ and integrate by parts in~$y$ as done before~\eqref{firstorder}.
Therefore, the summands in~\eqref{f2} and~\eqref{f3} are all of the order~$\O(\delta^2)$.
This completes the proof.
\QED

We conclude this section with a computation of the weight function~$f_p$
of the osculating vacuum introduced in~\eqref{rhoprep}.
It is related to~$\tilde{f}$ via the formula~\eqref{fpf}.
The next result shows that~$f_p$ agrees with~$\tilde{f}(p)$ up to an error
which is even of the order~$\delta^4$.
\begin{Lemma} \label{lemmafpf}
The weight~$f_p$ of the osculating measure is related to the weight function~$\tilde{f}$ by
\beq \label{fpfquant}
f_p = \tilde{f}(p) + \O \big( \delta^4 \big) \:.
\eeq
\end{Lemma}
\Proof Linearizing the interacting measure~$\tilde{\rho}$ about the vacuum
measure~$\rho_p$, the numerator in~\eqref{fpf} becomes
\begin{align}
&\fint_M \L\big( \tilde{F}(p), \tilde{F}(y) \big)\: \tilde{f}(y)\: d\rho(y) = \fint_M \L\big(F_p(p), F_p(y) \big)\: \tilde{f}(p)\: d\rho(y) \\
&\qquad +  \fint_M \nabla_{2,\v} \L\big(F_p(p), F_p(y) \big)\,d\rho(y) + \O \big(\v^2 \big)\:, \label{linearex}
\end{align}
where the jet~$\v(y) := \big(\tilde{f}(y) - \tilde{f}(p), \tilde{F}(y)-F_p(y)\big)$ is a solution of the linearized field equations (for basics on the linearized field equations and the
jet formalism see~\cite[Chapter~8]{intro}), i.e.\
\[ \fint_{M} \big( \nabla_{1,\v} + \nabla_{2,\v} \big) \L\big(F_p(p), y \big)\: d\rho(y) = 
\big( \nabla_\v \s \big)(p) = 0 \:. \]
Using also the EL equations in the vacuum, we conclude that the integral
in~\eqref{linearex} vanishes. Therefore,
\[ \fint_M \L\big( \tilde{F}(p), \tilde{F}(y) \big)\: \tilde{f}(y)\: d\rho(y) \\
= \tilde{f}(p) \fint_M \L\big(F_p(x), y \big)\: d\rho(y) + \O \big( \v^2 \big)\:. \]
According to~\eqref{oscFFp}, the jet~$\v$ vanishes quadratically at~$p$,
$\v(y) = \O ((y-p)^2)$. Therefore, the second order contribution is of the order~$\O(\delta^4)$.
Using these findings in~\eqref{fpf} gives the result.
\QED

\subsection{Expansion of the Divergences in the Ricci Tensor} \label{secdiv}
Contracting indices in~\eqref{RLtilde}, we obtain similar to~\eqref{Ricshort}
\beq \label{Ril}
R_{il}(q)
= \big( \overline{\partial}_i \partial_{kl} - \overline{\partial}_k \partial_{il} \big) A^k_p(x)
\big|_{x=p=q} \:.
\eeq
We now make essential use of the divergence structure of this equation.
By ``divergence'' we mean that the index~$k$ is contracted with
a derivative. Note that each summand in~\eqref{Ril} is of divergence form.
However, the corresponding derivatives are different: In the first summand, it is
an $x$-derivative, whereas in the second summand also $p$-derivatives
occur (note that, according to~\eqref{parbar}, the derivative~$\overline{\partial}_k$
is a sum of an $x$- and a $p$-derivative). We need to treat these two types of
divergences separately. We begin with the $x$-divergence.

\begin{Lemma} (Expansion of the $x$-divergence) \label{lemmaxdiv}
The $x$-divergence of the alignment vector field (introduced in Definition~\ref{defalign})
has the expansion for any $N \in \N$
\begin{align}
\partial_k A^k_p(x) &= 
\tilde{f}(x)\: \ell_p \big( \tilde{F}(x) \big) -  f_p\:\tilde{\ell}\big(F_p(x) \big) 
+ \big( \tilde{f}(x) - f_p \big)\: \s
- (\partial_k \tilde{f})(x)\: A^k_p(x) \label{xdiv} \\
&\quad \: - \sum_{r=2}^N \frac{(-1)^r}{r!} \fint_M \Big( \xi^m \frac{\partial}{\partial x^m} \Big)^r\:\tilde{f}(x)\: \L \big( \tilde{F}(x), F_p(y) \big)\: f_p\: d\rho(y) + R_N(x) \label{errs}
\end{align}
with the remainder term given by
\[ R_N(x) := - \dfrac{(-1)^N}{N!} \int_0^1 s^N \fint_M \Big(\xi^m \dfrac{\partial}{\partial \zeta^m}\Big)^{N+1}\:\tilde{f}(x_s)\: \L \big( \tilde{F}(x_s), F_p(x_s + \xi) \big)\:f_p\: d\rho(\xi)\:ds \:. \]
\end{Lemma} \noindent
The main point of this lemma is that on the right side in~\eqref{xdiv} the
functions~$\ell_p$ and~$\tilde{\ell}$ appear; this is where we can use the
EL equations. Note that the summands in~\eqref{errs} contain more and more
factors~$\xi$, showing that we have an expansion in powers of the
range~$\delta$ of the Lagrangian.
\Proof[Proof of Lemma~\ref{lemmaxdiv}.]
Using the symmetry of the Lagrangian, it is obvious that
\begin{align}
\tilde{f}(x)\: \Big( \ell_p \big( \tilde{F}(x) \big) + \s \Big)
&= \fint_M \tilde{f}(x)\: \L \big( \tilde{F}(x), F_p(y) \big)\:f_p\: d\rho(y) \label{varphi1} \\
f_p\:\Big( \tilde{\ell}\big(F_p(y) \big) + \s \Big)
&= \fint_M \tilde{f}(x)\: \L \big( \tilde{F}(x), F_p(y) \big)\: f_p\: d\rho(x) \:. \label{varphi0}
\end{align}
In order to interpolate between these two expressions, given~$\zeta \in M$ 
and a parameter~$s \in [0,1]$ we again consider the variables~$x_s$
and~$y_s$ in~\eqref{xysdef} and introduce the function
\[ 
\varphi(s, \zeta) := \fint_M \tilde{f}(x_s)\: \L \big( \tilde{F}(x_s), F_p(y_s) \big)\:f_p\:
d\rho(\xi) \:. \]
Evaluating this function at~$s=1$ and~$s=0$ gives~\eqref{varphi1} and~\eqref{varphi0},
respectively, both evaluated at~$\zeta$. Denoting~$\zeta$ by~$x$, we obtain
\[  \varphi(1, x) - \varphi(0, x) = \tilde{f}(x)\: \ell_p \big( \tilde{F}(x) \big) - 
f_p\:\tilde{\ell}\big(F_p(x) \big) + \big( \tilde{f}(x) - f_p \big)\: \s \:. \]
Moreover, differentiating the function~$\varphi(s, x)$ at~$s=1$ with the
chain rule yields
\begin{align}
\frac{d}{ds} \varphi(s,x)\big|_{s=1} &= \frac{\partial}{\partial \zeta^k}
\fint_M \tilde{f}(\zeta)\: \L \big( \tilde{F}(\zeta), F_p(\zeta + \xi) \big)\: \xi^k\:f_p\: d\rho(\xi) \Big|_{x=\zeta} \\
&= \fint_M \Big( \frac{\partial}{\partial x^k} + \frac{\partial}{\partial y^k} \Big) \:\Big( \tilde{f}(x)\: \L \big( \tilde{F}(x), F_p(y) \big)\: \xi^k\:f_p\Big) \: d\rho(y) \\
&= \frac{\partial}{\partial x^k}\fint_M
\tilde{f}(x)\: \L \big( \tilde{F}(x), F_p(y) \big)\: \xi^k\:f_p \: d\rho(y) \label{phidiff} \\
&= (\partial_k \tilde{f})(x)\: A^k_p(x) + \tilde{f}(x)\: \frac{\partial}{\partial x^k} A^k_p(x)
\end{align}
(where in~\eqref{phidiff} we integrated by parts to eliminate the $y$-derivative).
In this way, we get the desired divergence on the left side of~\eqref{xdiv}
(as well as the last summand on the right of~\eqref{xdiv}).

This is the motivation for taking the Taylor expansion of~$\varphi(s,x)$ about~$s=1$ (again, 
denoting $\zeta$ by $x$),
\[ \varphi(1,x) - \varphi(0,x) = -\sum_{r=N}^\infty \frac{(-1)^r}{r!} \frac{d^r}{ds^r}\: \varphi(1,
x) + R_N(x)\:. \]
The higher expansion terms are computed by
\[
  \frac{d^r}{ds^r} \varphi(s,x) = \fint_M \Big(\xi^k \dfrac{\partial}{\partial \zeta^k}\Big)
  ^r\:\tilde{f}(x_s)\: \L \big( \tilde{F}(x_s), F_p(y_s) \big)f_p\: d\rho(y) \: ,
\]
where in analogy to~\eqref{phidiff} at evaluated at $s=1$ gives
\[ \frac{d^r}{ds^r} \varphi(s,x)\big|_{s=1} = 
\fint_M \Big( \xi^k \frac{\partial}{\partial \zeta^k} \Big)^r\:\tilde{f}(x)\: \L \big( \tilde{F}(x), 
F_p(y) \big)\: \xi^k\:f_p\: d\rho(y) \:, \]
Collecting all the terms and using the definition of the remainder in~\eqref{phiremainer} gives the 
result.
\QED

\begin{Remark} \label{remalignment} {\bf{(Expansion of alignments)}} {\em{
We remark how the expansion in this lemma is related
to the earlier construction in~\cite[Theorem~5.1]{matter}.
Both expansions are almost the same. The only difference is that, here, we expand about~$s=1$, whereas in~\cite[Theorem~5.1]{matter} the expansion was performed about~$s=1/2$.
Expanding about~$s=1/2$ has the advantage that the resulting formulas are
anti-symmetric under the replacement~$s \leftrightarrow 1-s$,
implying that only the odd orders contribute. The present expansion about~$s=1$
harmonizes better with our definition of the alignment~\eqref{Apdef}, where~$x$
is fixed and it is integrated over~$y$ (the definition of the alignment
in~\cite{matter}, however, seems to work only in the linearized description). }}
\QEDrem
\end{Remark}

Taking the $x$-derivatives and using the EL equation gives the following result.
\begin{Lemma} \label{lemmaxdivfinal}
Derivatives of the $x$-divergence of the alignment vector field have the expansions for any $N \in
\N$
\begin{align}
\partial_l \partial_k A^k_p(x) \big|_{x=p} \quad&= \sum_{r=2}^N \frac{(-1)^r}{r!} \frac{\partial}{\partial x^l} \fint_M \Big( \xi^m \frac{\partial}{\partial x^m} \Big)^r\:\tilde{f}(x)\: \L \big( \tilde{F}(x), F_p(y) \big)\: f_p\: d\rho(y) \bigg|_{x=p} \notag \\
&\:\quad+  \frac{\partial}{\partial x^l} R_N(x)\Big|_{x=p} \label{ldiff} \\
\partial_{il} \partial_k A^k_p(x) \big|_{x=p=q} &= \sum_{r=2}^N \frac{(-1)^r}{r!} \frac{\partial^2}{\partial x^i \partial x^l} \fint_M \Big( \xi^m \frac{\partial}{\partial x^m} \Big)^r\:\tilde{f}(x)\: \L \big( \tilde{F}(x), F_q(y) \big)\: f_q\: d\rho(y) \bigg|_{x=q} \notag \\
&\:\quad + \frac{\partial^2}{\partial x^i \partial x^l} R_N(x)\Big|_{x=q} \:. \label{ildiff}
\end{align}
\end{Lemma}
\Proof In order to derive~\eqref{ldiff}, we differentiate the formula in Lemma~\ref{lemmaxdiv} with respect to~$x$ and set~$x=p$. The contributions from~\eqref{xdiv} all vanish
in view of the EL equations, Lemma~\ref{lemmaoscEL} and Lemma~\ref{lemmaf}.
This gives the result.

For the derivation of~\eqref{ildiff}, we differentiate the formula in Lemma~\ref{lemmaxdiv} 
twice with respect to~$x$ and set~$x=p=q$. The contributions from~\eqref{xdiv} again vanish
in view of the EL equations, Lemma~\ref{lemmaoscEL}
and Lemma~\ref{lemmaf}. \QED

Our remaining task is to treat the $p$-divergence
\[ \frac{\partial}{\partial p^k} A^k(p) \:. \]
We would like to use a similar strategy as in Lemma~\ref{lemmaxdiv}
and rewrite the divergence in terms of the functions~$\ell$ or~$\tilde{\ell}$.
However, it is not obvious how an interpolation between~$x$ and~$y$
as in~\eqref{xysdef} can be used. This becomes possible only after taking
the $x$-derivatives and evaluating~$x=p=q$. We first note that, differentiating
through with the product rule,
\begin{align}
&\partial_{il} \frac{\partial}{\partial p^k} A^k(p) \Big|_{x=p=q} \notag \\
&\overset{\eqref{Apdef}}{=}
\fint_M \frac{\partial}{\partial p^k} \Big( \partial_{1,il} \L\big( \tilde{F}(x), F_p(y) \big)\: \xi^k\: f_p\Big) \: d\rho(y)  \Big|_{x=p=q} \label{l1} \\
&\quad\:\; -\fint_M \frac{\partial}{\partial p^k} \Big( \big( \delta^k_i \partial_{1,l} + \delta^k_l \partial_{1,i} \big)
\L\big( \tilde{F}(x), F_p(y) \big)\: f_p \Big)\: d\rho(y) \Big|_{x=p=q}\:. \label{l2}
\end{align}
The last line involves no factor of~$\xi$ and can be rewritten in terms of~$\ell_p$,
\begin{align}
-\fint_M & \frac{\partial}{\partial p^k} \Big( \big( \delta^k_i \partial_{1,l}
+ \delta^k_l \partial_{1,i} \big) \L\big( \tilde{F}(x), F_p(y) \big)\: f_p \Big)\: d\rho(y) \Big|_{x=p=q} \notag \\
&= - \Big( \frac{\partial}{\partial p^i} \frac{\partial}{\partial x^l} \ell_p\big( \tilde{F}(x) \big)
+ \frac{\partial}{\partial p^l} \frac{\partial}{\partial x^i} \ell_p\big( \tilde{F}(x) \big)
\Big) \Big|_{x=p=q} \:. \label{r3}
\end{align}
This vanishes by the osculation equations.

The first line~\eqref{l1}, on the other hand, can be written as an $s$-derivative,
\[ \fint_M \frac{\partial}{\partial p^k} \Big(\partial_{1,il} \L\big( \tilde{F}(x), F_p(y) \big)f_p\Big)\: \xi^k\: d\rho(y) \Big|_{x=p=q} = \frac{d}{ds} \sigma(s,x)\big|_{s=0,\: x=q} \]
with
\beq \label{phipdiv}
\sigma(s,x) := 
\fint_M \partial_{1,il} \L\big( \tilde{F}(x), F_{x+s \xi}(y) \big)\: f_{x+s \xi}\: d\rho(y)\:.
\eeq
This function is well defined for all $x \in U$ from a chart $(\phi, U)$ as long as the base points 
on $(1-s)x + sy$ lie in $U$. This can always be assumed because $|y-x| \leq \delta$. Moreover,
\begin{align}
\sigma(0,x) &= \fint_M \partial_{1,il} \L\big( \tilde{F}(x), F_x(y) \big)\: f_x\: d\rho(y) = \dfrac{\partial^2}{\partial x'^i \partial x'^l} \ell_x\big( \tilde{F}(x') \big) \big|_{x'=x} \label{r1} \\
\sigma(1,x) &= \fint_M \partial_{1,il} \L\big( \tilde{F}(x), F_y(y) \big)\: f_y\: d\rho(y)
= \fint_M \partial_{1,il} \L\big( \tilde{F}(x), \tilde{F}(y) \big)\: f_y\: d\rho(y) \:, \notag
\end{align}
where  in the last step we used~\eqref{Fp2} in Lemma~\ref{lemmaFp}.
If~$f_y$ were equal to~$\tilde{f}(y)$, the last integral could be rewritten
as~$\partial_{il} \tilde{\ell}(\tilde{F}(x))$. However, according to Lemma~\ref{lemmafpf},
this equality holds only approximately. Therefore, we write~$\sigma(1,x)$ as
\beq
\sigma(1,x) = \partial_{il} \tilde{\ell}\big( \tilde{F}(x) \big) - W_{il}(x) \:, \label{r2}
\eeq
with the {\em{weight deviation tensor}}~$W_{il}$ defined by
\beq \label{Wdeviate}
W_{il}(x) := \fint_M \partial_{1,il} \L\big( \tilde{F}(x), \tilde{F}(y) \big)\: \big(
\tilde{f}(y) - f_y \big) \: d\rho(y) \:.
\eeq
According to Lemma~\ref{lemmafpf}, the weight deviation tensor is of
order~$\O(\delta^4)$ and will thus only
describe Planck-scale corrections to the Einstein equations.

The right side of~\eqref{r1} and~\eqref{r2} vanishes for all $p \in U$ in view of Lemma~\ref{lemmaoscEL}. 
We thus obtain the following result.
\begin{Lemma} (Expansion of the $p$-divergence) \label{lemmapdiv}
The $p$-divergence of the alignment vector field (introduced in Definition~\ref{defalign})
has  for any $N \in \N$ the expansion
\begin{align*}
&\partial_{il} \frac{\partial}{\partial p^k} A^k(p) \Big|_{x=p=q} = -W_{il}(q)  \\
&\quad\: -\sum_{r=2}^N \frac{1}{r!}\fint_M \Big( \xi^m \frac{\partial}{\partial p^m} \Big)^r
\Big(\partial_{1,il} \L\big( \tilde{F}(x), F_p(y) \:f_p\Big)\: d\rho(y) \Big|_{x=p=q} \\
&\quad\: -\frac{1}{N!} \int_0^1 (1-s)^N \fint_M \Big( \xi^m \frac{\partial}{\partial \hat{p}^m} 
\Big)^{N+1} \Big( \partial_{1,il} \L\big( \tilde{F}(x), F_{\hat{p}}(y) \big)\: f_{\hat{p}} \Big) \bigg|_{\hat{p} = x + s\xi,\: x=p=q} \!\!\!\!\!\!\!\!\!\!\!\!\!\!\!\!\!\!\!\! d\rho(y)\: ds \:.
\end{align*}
\end{Lemma}
\Proof Since~\eqref{r3}, \eqref{r1} and~\eqref{r2} vanish, it follows that
\begin{align*}
\partial_{il} \frac{\partial}{\partial p^k} A^k(p) \Big|_{x=p=q} &=
\frac{d}{ds} \sigma(s,x)\big|_{s=0} = \sigma(1,x) - \sigma(0,x)
- \sum_{r=2}^N \frac{1}{r!} \frac{d^r}{ds^r}\: \sigma(0,x) + R^0_N(x)\:.
\end{align*}
Note that here we expanded around~$s=0$. Consequently, the remainder term
is of the form
\[ R^0_N(x) = \frac{1}{N!} \int_0^1 (1-s)^N\: \frac{d^{N+1}}{ds^{N+1}} \sigma(s,x)\: ds \:, \]
Computing the $s$-derivatives of~\eqref{phipdiv} gives the result.
\QED

\section{The Einstein Equations} \label{seceinstein}

\subsection{The Riemannian Einstein Equations} \label{secriemann}
We saw in Section~\ref{seccurvL} that the Ricci curvature of~$\nabla^\L$
satisfies the contracted Bianchi identities~\eqref{bianchiL}, which are quite different
from the standard formulas known from Riemannian and Lorentzian geometry.
This has major disadvantage that, formulating the Einstein equations with this Ricci
tensor, the conservation laws of energy and momentum as expressed by the
fact that the energy-momentum tensor is divergence-free would no be apparent.
For this reason, it is preferable to work with the Ricci tensor of the metric~$g$.
This can be done because the difference of the connections as expressed by the
deviation tensor~\eqref{devdef} is quadratic in~$\delta$ (see Corollary~\ref{corK})
and can therefore be associated not to the geometry, but to the matter fields.
Combining its effects on curvature with the energy-momentum tensor
makes it possible to formulate the Einstein equation in the familiar setting
with a metric connection. Using that the Ricci tensor of~$\nabla^\L$ 
can be written with the help of the EL equations as expressions quadratic
in~$\delta$ (see Lemmas~\ref{lemmaxdivfinal} and~\ref{lemmapdiv}),
we obtain the following result. We note that, in the two-dimensional case, the curvature
tensor is determined by the Gau{\ss} curvature. As a consequence, the Einstein
tensor is zero, and the Einstein equations simply state that the energy-momentum tensor must vanish. In view of these results, we restrict attention to the case
of dimension greater than two.

\begin{Thm} {\bf{(The Riemannian Einstein Equations)}} \label{thmrein}
Assume that~$\tilde{M}$ is of dimension~$k > 2$. Then the Riemannian metric~$g$ introduced in Section~\ref{secg}
satisfies the Einstein equations
\beq \label{rein}
R^g_{il} - \frac{1}{2}\: R^g\, g_{il} = T_{il} \:,
\eeq
where the energy-momentum tensor has for any~$N>2$ the expansion
\begin{align}
&T_{il}(q) - \frac{1}{k-2}\: T(q)\, g_{il}(q) = 
\big( - \nabla^g_i K^a_{al} + \nabla^g_a K^a_{il}
- K^b_{ia} K^a_{bl} + K^b_{ba} K^a_{il} \big)(q) \label{Tb} \\
&+ \sum_{r=2}^N \frac{(-1)^r}{r!} \frac{\partial^2}{
\partial p^i \partial x^l} \fint_M \Big( \xi^m \frac{\partial}{\partial x^m} \Big)^r\:\tilde{f}(x)\: \L \big( \tilde{F}(x), F_p(y) \big)\: f_p\: d\rho(y) \bigg|_{x=p=q} \label{Tc}  \\
&+\sum_{r=2}^N \frac{1}{r!} \fint_M \Big( \xi^m \frac{\partial}{\partial p^m} \Big)^r\: \partial_{1,il} \L\big( \tilde{F}(x), F_p(y) \big)\: f_p\: d\rho(y) \Big|_{x=p=q} + W_{il}(q) \label{Te} \\
&+ \O \big(\delta^{N+1} \big)\:. 
\end{align}
The energy-momentum tensor is symmetric and divergence-free,
\beq \label{Tprop}
T_{il} = T_{li} \qquad \text{and} \qquad \nabla^g_i T^{il} = 0 \:.
\eeq
It is very small compared to the Ricci tensor in the sense that it has the scaling behavior
for small~$\delta$
\beq \label{Tdelta}
T_{jk} = \O\big(\delta^2 \big)\:.
\eeq
\end{Thm}
\Proof The relations~\eqref{Tprop} follow immediately
from corresponding properties of the Ricci tensor of the metric~$g$
(we remark that the symmetry in the indices~$i$ and~$l$
is not apparent in~\eqref{Tb}--\eqref{Te}; it comes about
due to relations between the deviation tensor and the 
expression in~\eqref{Tc}).
The scaling~\eqref{Tdelta} was already explained
right before the statement of the theorem.
For the derivation of the formula for the energy-momentum tensor,
we first note that, taking the trace of the Einstein equations~\eqref{rein},
\[ \Big( 1 - \frac{k}{2} \Big)\: R^g = T \:. \]
Therefore, the Einstein equations can be written equivalently as
\[ R^g_{il} = T_{il} - \frac{1}{k-2}\: T\: g_{il} \:. \]
With this in mind, on the right side in~\eqref{Tb}--\eqref{Te}
we collected all the contributions to the Ricci tensor as computed in Lemmas~\ref{lemmaRicrel1}, \ref{lemmaxdivfinal} and~\ref{lemmapdiv}.

More precisely, the contributions in~\eqref{R2} give~\eqref{Tb}.
Thus it remains to consider the Ricci-tensor of the connection~$\nabla^\L$,
which we write according to~\eqref{Ril}. Using~\eqref{parbar},
the derivatives~$\partial_{ik}$ drop out,
\beq \label{twosummands}
R_{il}(q)
= \frac{\partial}{\partial p^i} \partial_{kl} A^k_p(x) \Big|_{x=p=q}
- \frac{\partial}{\partial p^k} \partial_{il} A^k_p(x) \Big|_{x=p=q} \:.
\eeq
In order to compute the first summand, we act on~\eqref{ldiff}
with the differential operator~$\overline{\partial}_i$ and subtract~\eqref{ildiff}.
We thus obtain
\begin{align*}
&\frac{\partial}{\partial p^i} \partial_{kl} A^k_p(x) \Big|_{x=p=q} \\
&= \sum_{r=2}^N \frac{(-1)^r}{r!} \frac{\partial^2}{
\partial p^i \partial x^l} \fint_M \Big( \xi^m \frac{\partial}{\partial x^m} \Big)^r\:\tilde{f}(x)\: \L \big( \tilde{F}(x), F_p(y) \big)\: f_p\: d\rho(y) \bigg|_{x=p=q} + R_N(x) \:.
\end{align*}
The second summand in~\eqref{twosummands}, on the other hand,
was computed in Lemma~\ref{lemmapdiv}. 
Collecting all the terms gives the result.
\QED

\subsection{The Regularizing Vector Field and the Lorentzian Metric~$\eta$}\label{secregvec}
In the analysis so far, we worked with a Riemannian metric.
In order to move on to Lorentzian signature, we need to introduce a
canonical vector field which distinguishes timelike from spacelike directions.
The naive idea would be to take the alignment vector field~\eqref{Apdef}.
But this does not work, because the alignment vector field vanishes in the 
vacuum by symmetry. Therefore, we need an additional structure
which is available in the setting of causal fermion systems:
the time direction functional~${\mathscr{C}}$.
Therefore, from now on we consider the setting of causal fermion systems
in dimension~$k=4$. Then the {\em{time direction functional}}
takes the form (see~\cite[eq.~(1.1.11)]{cfs}) 
\[ 
{\mathscr{C}} \::\: M \times M \rightarrow \R\:,\qquad
{\mathscr{C}}(x, y) := i \tr \big( y\,x \,\pi_y\, \pi_x - x\,y\,\pi_x \,\pi_y \big) \:, \]
where~$\pi_x : \H \rightarrow x(\H) \subset \H$ denotes the orthogonal projection
to the spin space. It is obviously anti-symmetric if~$x$ and~$y$ are interchanged.
Now we can insert this functional into the integrand of the alignment vector field~\eqref{Apdef}.

\begin{Def} \label{defreg}
We define the {\bf{future-directed Lagrangian}}~$\L^\vee(x,y)$ by
\[ \L^\vee(x,y) := 2\: \Theta\big( {\mathscr{C}}(x,y) \big)\: \L(x,y) \:, \]
where for the Heaviside function~$\Theta$ we use the convention
\[ \Theta(z) = \left\{ \begin{array}{cl} 0 & \text{if~$z<0$} \\[0.2em]
\displaystyle \frac{1}{2} & \text{if~$z=0$} \\[0.7em]
1 & \text{if~$z>0\:.$} \end{array} \right. \]
The {\bf{regularizing vector field}}~$u$ is defined for any~$p \in \tilde{M}$ by
\beq \label{uint}
u(p) := \fint_{M_p} \L^\vee(p,y)\: y\: d\rho_p(y) \:.
\eeq
\end{Def} \noindent
The regularizing vector field was first introduced in globally hyperbolic
spacetimes in the context of a baryogenesis mechanism in~\cite{baryogenesis};
see also~\cite{baryomink, baryoconform}.
The more abstract definition given here indeed generalizes the previous notion,
as will be explained in Appendix~\ref{appregvac}.
Moreover, it is shown that the regularizing vector field is non-zero. One key is the observation 
that the regularizing vector field is parallel, up to a small error:
\begin{Lemma} \label{lemmau} The regularizing vector field is almost parallel
in the sense that
\[ \nabla^\L u = \O\big( \delta^2 \big) \:. \]
\end{Lemma}
\Proof We proceed similar as in the proof of Lemma~\ref{lemmametric}.
\Felix{Gibt es so Probleme mit Glattheit?}%
In $\L$-induced chart centered at~$q$, the regularizing vector field can be
written as
\[ u^k(p) = \fint_M \L^\vee\big(\tilde{F}(p), F_p(y) \big)\: (y-p)^k\: f_p\:d\rho(y) \:. \]
Taking the first derivatives, we obtain
\[ \nabla^\L_i u^k(q) = \frac{\partial}{\partial p^i} \fint_M \L^\vee\big(\tilde{F}(p), F_p(y) \big)\:  (y-p)^k\: f_p\:d\rho(y) \Big|_{p=q} \:. \]
Now we expand~$\tilde{F}$, $f_p$ and~$F_p$ about~$q$, exactly as explained
in Lemma~\ref{lemmametric} (the fact that the Heaviside function is not differentiable
causes no problems, because we integrate over~$y$ before differentiating).
This gives the result.
\QED

We now introduce a Lorentzian metric as the {\em{flip metric}} of~$g$.
We follow the standard setup and conventions; for details see for example~\cite{olea}. Is convenient to normalize the regularizing vector field
by setting
\beq \label{hatudef}
\hat{u}(p) := \frac{u(p)}{\sqrt{g_p(u, u)}}\:.
\eeq
We denote the corresponding one-form by~$\omega$, i.e.\
\[ \omega(v) = g_p(\hat{u}, v) \quad \text{for all~$v \in M_p$} \]
(thus, in components, $\omega_i(p) = g_{ij}(p)\: \hat{u}^j(p)$).
Given a parameter~$\tau>1$, the {\em{Lorentzian flip metric}}~$\eta$
of~$g$ along~$\hat{u}$ is defined by
\beq \label{flipgen}
\eta = \tau \, \omega \otimes \omega - g
\eeq
(thus, in components, $\eta_{jk} = \tau \:\omega_j \omega_k - g_{jk}$).
As is immediately verified, the flip metric is indeed Lorentzian
with signature convention~$(+,-,-,-)$.
In our context, the parameter~$\tau$ must be chosen in such a way
that the causal cones of the Lorentzian metric agree with the causal
structure of the causal fermion system. As is worked out in Appendix~\ref{appregvac}, the correct choice is~$\tau=4$.
 In order to avoid  a conflict of notation with the Riemannian metric, we denote this Lorentzian metric by~$\eta$.
\begin{Def}\label{defflip}
The {\bf{Lorentzian metric}}~$\eta$ of~$g$ is defined by
  \begin{align} \label{etadef}
    \eta = 4\, \omega \otimes \omega - g\: .
  \end{align}
\end{Def}

We conclude this section by relating the Ricci curvature
of~$g$ with that of~$\eta$
(general relations between the connections and curvature tensors
have been studied in~\cite{olea, reddy}).
In preparation, we note that the normalized
regularized vector field is again almost parallel:
\begin{Lemma}\label{propsu}
The vector field~$\hat{u}$ introduced in~\eqref{hatudef} is
almost parallel in the sense that
\[ \nabla^\L \hat{u} = \mathcal{O}(\delta^2) \:. \]
\end{Lemma}
\begin{proof}
The statement follows from Lemma~\ref{lemmau} by a direct computation, working in
the $\L$-induced chart centered at $q$. Let $v \in \Gamma(T\tilde{M})$ be an arbitrary
vector field. Writing $N = \sqrt{g_p(u,u)}>0$ and applying the product rule, we get
\[
  \nabla^\L_v \hat{u}(q)\big|_{q=p} = \Bigg(\dfrac{\nabla^\L_v u}{N} - \dfrac{v(N)}{N^2}\Bigg)
  \Bigg|_{q=p} \, .
\]
Using that $v(N^2) = 2 N v(N)$ , we obtain  
\[
  v(N) = \dfrac{g(\nabla^\L_v u, u)}{N} + \dfrac{\nabla^\L_v (g)(u,u)}{2N} \:.
\]
Applying Lemma~\ref{lemmametric}, it follows that
\[
  \nabla^\L_v \hat{u}\Big|_{q=p} = \underbrace{\dfrac{\nabla^\L_v u}{N}}_{= \mathcal{O}(\delta^2)} 
  - \underbrace{\Bigg(\dfrac{g(\nabla^\L_v u,u)}{N^3} + \dfrac{\nabla^\L_v (g)(u,u)}{2N^3}\Bigg)}_
  {= \mathcal{O}(\delta^2)}\Bigg|_{q=p} = \mathcal{O}(\delta^2) \:,
\]
concluding the proof.
\end{proof} 
\begin{Lemma}\label{diffchristoffel}
  Let $g$ be the Riemannian metric from equation~\eqref{gpstar} and $\eta$ the 
  Lorentzian metric from Definition~\ref{defflip}.
  The corresponding Levi-Civita connections are related byt
    \begin{align}
    \nabla_i^\eta v^k- \nabla_i^g v^k &= C^k_{ij} \, v^j \label{Cdef} \\
    &= \dfrac{4}{3} S_{ij} \hat{u}^k v^j - \dfrac{8}{3}(a_i \hat{u}_j + a_j \hat{u}_i) \hat{u}^k 
    v^j - 4(A^k_i \hat{u}_j + A^k_j \hat{u}_i) v^j\, .
  \end{align}
Here~$v$ is a vector field,
$S_{ij} := \nabla^g_{(i}\hat{u}_{j)}$ and~$A_{ij} := \nabla^g_{[i}\hat{u}_{j]}$ denotes the 
symmetric respectively anti-symmetric part of~$\nabla^g \hat{u}$,
and~$a_i := (\nabla^g_{\hat{u}} \hat{u})_i$ is the acceleration of the $\hat{u}$-congruence.
\end{Lemma}
\begin{proof}
  We start by noting that $C^k_{ij}$, defined as the difference of two Christoffel symbols, is a 
  tensor. Thus, the result holds in all coordinates. We choose Gaussian coordinates for $g$, with 
  $p \in \tilde{M}$ and $v \in T_p \tilde{M}$ arbitrary. Since $\nabla^g$ and $\nabla^\eta$ are the 
  associated Levi-Civita connections, we get
  \[ \big(\nabla_i^\eta v^k - \nabla_i^g v^k \big)\big|_p = \Gamma^{\eta, k}_{ij}\: v^j =
    \dfrac{1}{2} \:
    \eta^{kl} \: \big( \nabla^g_i \eta_{jl} + \nabla^g_j \eta_{il} - \nabla^g_l \eta_{ij}
    \big)\: v^j \, .
  \]
  Next, we use $\nabla^g g = 0$ and substitute~\eqref{etadef}.
  Decomposing~$\nabla^g_j 
  \hat{u}_i$ into its symmetric and antisymmetric parts, we get~\eqref{Cdef}
  with
  \[
    C^k_{ij} = 4\: \eta^{kl}\: (S_{ij} \hat{u}_l + A_{il} \hat{u}_j + A_{jl} \hat{u}_i) \big|_p\, .
  \]
  Next, we compute the inverse metric $\eta^{kl}$ and obtain after a short computation
  \[
    \eta^{kl} = \dfrac{4}{3} \hat{u}^k \hat{u}^l - g^{kl} \, .
  \]
  In a last step, we use that $\hat{u}^l\nabla^g_l \hat{u}_i = \frac{1}{2}\nabla^g_{\hat{u}} (\hat{u}^l\hat{u}_l) = 0$ because $g(\hat{u},
  \hat{u}) = 1$ by construction. Thus, we can write $\hat{u}^l A_{li} = - \frac{1}{2} \hat{u}^l\nabla^g_l\hat{u}_i$, which 
  gives the result.
\end{proof}
Next, knowing the difference tensor of the Christoffel symbols, we can directly compute the 
difference of the Riemann tensors. Let $\nabla^\eta = \nabla^g + C$ and using the standard 
convention for the Riemann tensor, we get
\beq \label{RRformula}
  R^{(\eta) k}_{ijl} - R^{(g) k}_{ijl} = \nabla^g_j C^k_{li} - \nabla^g_l C^k_{ji} + C^k_{jm} C^m_
  {li} - C^k_{lm} C^m_{ji} 
\eeq
The difference consists of terms proportional to $\nabla^g \hat{u}$. Therefore, if $\hat{u}$ is parallel, the 
corrections from changing to the Lorentzian case vanish. Additionally, these corrections, by 
Lemma~\ref{propsu}, are at least of order $\mathcal{O}(\delta^2)$,
which will make it possible to include them again in the energy-momentum tensor.

Contracting indices in~\eqref{RRformula}, we obtain the following
relation between the Ricci tensors.
\begin{Prp}\label{diffRic} In the setting of
Lemma~\ref{diffchristoffel}, the difference of the Ricci 
  tensors has the form
\[ {\rm{Ric}}^\eta_{il} - {\rm{Ric}}^g_{il} = \nabla^g_k C^k_{li} - C^k_{lj}C^j_{ki} \:. \]
\end{Prp}
\begin{proof}
  We start by taking the contraction, which yields
  \[
    {\rm{Ric}}^\eta_{il} - {\rm{Ric}}^g_{il} = \nabla^g_k C^k_{li} - \nabla^g_l C^k_{ki} + C^k_{kj}
    C^j_{li} - C^k_{lj}C^j_{ki} \, .
  \]
  It remains to show that $C$ is trace-free, which gives the result. To see this, we first 
  note that for a general metric $g$,
  \[
    \Gamma^k_{kj} = \dfrac{1}{2} g_{kl} \partial_j g_{kl} = \tr(g^{-1}\partial_j g) = \tr
    (\partial_j\ln g) = \partial_j \ln \sqrt{|\det g|} \, .
  \]
  Thus, we get for the difference tensor $C$
\beq \label{Ckk}
C^k_{kj} = \partial_j \ln \sqrt{|\det \eta|} - \partial_j \ln \sqrt{|\det g|} = \dfrac{1}{2} \:
    \partial_j \ln | \det(g^{-1} \eta) | \:.
\eeq
In order to compute the determinant of $g^{-1}\eta$, we determine its eigenvalues as an endomorphism of 
  the tangent space,
  \[
    g^{kl}\eta_{li} = 4\, \hat{u}^k \hat{u}_i - \delta^k_i \, .
  \]
  Since $\hat{u}^k \hat{u}_i$ is a rank-one operator on the subspace spanned by $\hat{u}$, we get the eigenvalues
  \begin{align}
    \text{For } \hat{u} :\, &4 \hat{u}^k (\hat{u}_i \hat{u}^i) - \hat{u}^k = 3 \hat{u}^k \, \Longrightarrow \lambda_1 = 3 \notag \\
    \text{For } v \in \hat{u}^\perp :\, &4 \hat{u}^k (\hat{u}_i v^i) - v^k = -v^k \Longrightarrow \lambda_2 = -1 
    \notag \, .
  \end{align}
  The multiplicity of $\lambda_1$ is one and that of $\lambda_2$ equals $n-1$, where $n$ is the 
  dimension of $\tilde{M}$. We thus conclude that
\[ \det \big( g^{-1}\eta \big) = \prod_i \lambda_i = 3\: (-1)^3 = -3 \:, \]
which is a constant. Therefore, its derivative in~\eqref{Ckk} vanishes, implying that~$C^k_{kj} = 0$.
\end{proof} \noindent
We remark that the above observation that~$C$ is trace-free can be understood
geometrically from the fact that volume forms of $\eta$ and $g$
differ only by an overall constant. 

\subsection{The Lorentzian Einstein Equations} \label{seclorentz}
We now come to our main result: the Einstein equations in the Lorentzian setting.
\begin{Thm} {\bf{(The Lorentzian Einstein Equations)}} \label{thmlein}
Assume~$\tilde{M}$ is of dimension four. 
Then the Lorentzian metric~$\eta$ introduced in Definition~\ref{defflip} satisfies the Einstein equations
\[ R^\eta_{il} -\frac{1}{2}\: R^\eta\, \eta_{il} = T_{il} \:, \]
where the energy-momentum tensor has for any~$N>2$ the expansion
\begin{align} \label{emom} 
T^\eta_{il}(q) - &\frac{1}{2}\: T^\eta(q)\, \eta_{il}(q) = 
\big( - \nabla^g_i K^a_{al} + \nabla^g_a K^a_{il}
- K^b_{ia} K^a_{bl} + K^b_{ba} K^a_{il} \big)(q) \\
&+ \sum_{r=2}^N \frac{(-1)^r}{r!} \frac{\partial^2}{
\partial p^i \partial x^l} \fint_M \Big( \xi^m \frac{\partial}{\partial x^m} \Big)^r\:\tilde{f}(x)\: \L \big( \tilde{F}(x), F_p(y) \big)\: f_p\: d\rho(y) \bigg|_{x=p=q}  \\
&+\sum_{r=2}^N \frac{1}{r!} \fint_M \Big( \xi^m \frac{\partial}{\partial p^m} \Big)^r\: \partial_{1,il} \L\big( \tilde{F}(x), F_p(y) \big)\: f_p\: d\rho(y) \Big|_{x=p=q} + W_{il}(q) \\
&+ \big( \nabla^g_a C^a_{li} - C^a_{lj}C^j_{ai} \big)(q) + \O \big(\delta^{N+1} \big) \:, \label{emomfin}
\end{align}
where $C$ is the difference tensor from Lemma~\ref{diffchristoffel}.
The energy-momentum tensor is symmetric and divergence-free,
\[ T_{il} = T_{li} \qquad \text{and} \qquad \nabla^\eta_i T_\eta^{il} = 0 \:. \]
It is very small compared to the Ricci tensor in the sense that it has the scaling behavior
for small~$\delta$
\[ T_{jk} = \O\big(\delta^2 \big)\:. \]
\end{Thm}
\Proof
  From the result of Proposition~\ref{diffRic}, we can express the Lorentzian Ricci tensor as the corresponding Riemannian tensor 
  plus extra terms containing the unit vector field $\hat{u}$ and its derivatives.
  \[
    {\rm{Ric}}^\eta_{il} = {\rm{Ric}}^g_{il} + D_{il} =: T^\eta_{il} - \dfrac{1}{2}
    T^\eta \eta_{il} \, .
  \]
  From Theorem~\ref{thmrein}, we get the result for ${\rm{Ric}}^g$, and from Proposition~\ref{diffRic} and Lemma~\ref{diffchristoffel} the contributions for $D_{il}$, proving the form of 
  equation~\eqref{emom}.

  Since $T_{jl}$ is symmetric and one can verify directly that $D_{il} = D_{li}$, the Lorentzian 
  energy-momentum tensor $T^\eta_{jl}$ is symmetric as well. Furthermore, since $T_{jl}$ and $D_{il}
  $ are of order $\mathcal{O}(\delta^2)$, so is $T^\eta_{jl}$. Finally, because $\nabla^\eta$ is 
  the Levi-Civita connection constructed from $\eta$, 
  it satisfies the usual contracted Bianchi identities, 
  and hence $T^\eta_{jl}$ is divergence-free. 
\QED
We note that, using the constructions of the recent paper~\cite{cfs-curved},
the results of this theorem could be extended to any spacetime dimension
greater than two. We do not do this here, because some of the technical results
in~\cite{oppio, cfs} so far have been worked out only in the four-dimensional case.

We close with a few remarks on the {\em{structure and physical interpretation of the energy-momentum tensor}} in Theorem~\ref{thmlein}. The above theorem states that the Lo\-rentzian Ricci tensor is zero up to contributions of the order~$\O(\delta^2)$, and that
the Ricci tensor can be expressed
in terms of a power expansion in~$\delta$. It is a-priori not clear if the
contributions~\eqref{emom}--\eqref{emomfin} can be associated with matter
fields or whether they should be regarded as corrections to the geometric evolution equations.
This shortcoming is closely related to the fact that, in the causal fermion system approach,
the geometry and matter fields are desribed in a unified way by the measure~$\tilde{\rho}$.
This means that, on the fundamental level, there is no clear distinction between
the spacetime geometry and the matter fields. The only way to clarify the situation
is to compute the energy-momentum in detail in the presence of different matter fields
(like a Dirac particle, a Maxwell field, etc.). If this is done, the leading contribution to~$T_{il}$
proportional to~$\delta^2$ has the interpretation as the usual energy-momentum tensor,
with the prefactor~$\delta^2$ taking the role of the {\em{gravitational coupling constant}}.

We again point out that the distinction between geometry
and matter fields is not clear-cut and depends on the physical interpretation.
For example, a correction due to torsion could be associated with a modification of the
geometry, but alternatively one could also interpret the torsion tensor as a new
type of matter field. For conceptual clarity, we here adopt the convention that all corrections
to the standard Einstein equations are to be included in the energy-momentum tensor.
A particular example of this convention is our treatment of the
{\em{cosmological constant}}.
In the formulation of Theorem~\ref{thmlein}, the cosmological
term~$\Lambda \eta_{il}$ is included in the energy-momentum tensor~$T_{il}$.
Thus, in order to study the cosmological constant, one would have to
analyze the contributions in~\eqref{emom}--\eqref{emomfin} in detail.
We note that this concept of associating the cosmological term to the
energy-momentum tensor fits together with the recent proposal in~\cite{cosmo}
where the cosmological term is obtained as a contribution to the energy-momentum
tensor resulting from the collective behavior of all the wave functions which form the
Dirac sea, needed in order to arrange correlated initial and end quantum states of the universe.

\section{Discussion of Corrections to the Einstein Equations} \label{seccorrect}
Our methods provide a systematic procedure for deriving correction terms
to the Einstein equations. We now compile and briefly discuss
different corrections.
\bitem
\item[{\bf{(a)}}] {\bf{Planck scale corrections:}}
Here we consider the higher order terms in~$\delta$.
They are interesting, but probably too small for being detected
directly in experiments.
\item[{\bf{(b)}}] {\bf{Corrections from the non-optimality of the osculation:}}
Additional corrections come about if one takes into account that
the tangent space~$T_p \tilde{M}$ may deviate from the osculating
vacuum~$M_p$ (see Figure~\ref{figosculate}). In order to take the
resulting corrections into account, one needs to study corrections coming from
the operator~$K$ in~\eqref{frakSdef}.
Moreover, we expect corrections to our findings that the metric and the
regularization vector field are almost parallel (see Lemmas~\ref{lemmametric}
and~\ref{lemmau}).
Finally, there are also corrections arising from the fact that the base points
of the osculation do in general not fit (i.e.\ that~$F_p(p) \neq \tilde{F}(p)$).
An interesting feature is that, in this case, the connection~$\nabla^\L$ has
{\em{torsion}}. We expect the resulting corrections to the
Einstein equations to be again Planck scale corrections.
\item[{\bf{(c)}}] {\bf{Corrections from the regularizing vector field:}}
It is a specific feature of the causal fermion system approach that there
is a distinguished {\em{regularization vector field}}
which is timelike and almost parallel (see  Definition~\ref{defreg} and Lemma~\ref{lemmau}). 
Already the existence of such a vector field poses
constraints on the geometry of spacetime.
The modifications to the Dirac dynamics coming from the regularizing
vector field have already been studied in the context of a baryogenesis mechanism
in~\cite{baryogenesis, baryomink, baryoconform}.
The effects on the Einstein equations remain to be analyzed.
Moreover, we expect that, in order to quantify the constraints coming from the existence of
an almost-parallel regularization vector field, one should take
into account the corrections coming from the non-optimality of the osculation
mentioned in {\bf{(b)}}.
\item[{\bf{(d)}}] {\bf{Corrections related to modified measures:}}
It is a specific feature of causal fermion system approach that the
volume measure~$d\rho$ in spacetime does not necessarily need to coincide
with the usual Lorentzian volume measure~$\sqrt{|\det \eta|}\: d^4x$.
This makes a connection to modified measure theories, as was analyzed and discussed in~\cite{mmt-cfs}.
In the present paper, the appearance of modified measures
becomes apparent in the fact that the weight function~$\tilde{f}$ 
in Lemma~\ref{lemmaf} seems in general different from the weight
function~$\sqrt{|\det \eta|}$ coming from the Lorentzian metric.
A related effect is the appearance of the weight deviation tensor defined in~\eqref{Wdeviate},
which should also give rise to Planck scale corrections of the Einstein equations.
\eitem
All these corrections still need to be worked out in detail.
To this end, we plan to build the bridge between the abstract
framework presented here and the
analytical methods developed in~\cite{cfs, cfs-curved}.

\section{Outlook: The Einstein Equations in Non-Smooth and Quantum Spacetimes}
\label{secoutlook}
With the above constructions we showed that, for smooth spacetimes~$\tilde{M}$,
the Einstein equations follow from the EL equations of the
causal action principle.
If~$\tilde{M}$ is not assumed to be a smooth manifold, the causal action principle
is still well-defined. Therefore, the corresponding EL equations
are mathematically well-defined equations which include the gravitational
interaction. The only point which is not quite satisfying is that, in the
non-smooth setting, it is no longer obvious how to interpret the
EL equations geometrically. Therefore, it is an interesting problem to explore
how and to which extent geometric constructions can be extended
or generalized to non-smooth situations.

A first step in this direction is made in the recent paper~\cite{lcalc},
where a differential calculus is developed in the non-smooth setting,
again working with osculating vacua.
These constructions apply in particular to discrete spacetimes, as
is illustrated in Figure~\ref{figosculate-discrete}.
\begin{figure}[tb]
\psset{xunit=.4pt,yunit=.4pt,runit=.4pt}
\begin{pspicture}(598.5941537,127.59481115)
{
\newrgbcolor{curcolor}{0 0 0}
\pscustom[linestyle=none,fillstyle=solid,fillcolor=curcolor]
{
\newpath
\moveto(50.59844906,59.89968657)
\curveto(51.32864497,57.94419428)(50.31631516,55.82072147)(48.33734716,55.15677732)
\curveto(46.35837915,54.49283317)(44.16216787,55.53983981)(43.43197196,57.4953321)
\curveto(42.70177606,59.4508244)(43.71410587,61.5742972)(45.69307387,62.23824136)
\curveto(47.67204188,62.90218551)(49.86825316,61.85517887)(50.59844906,59.89968657)
\closepath
}
}
{
\newrgbcolor{curcolor}{0 0 0}
\pscustom[linewidth=1.72947494,linecolor=curcolor]
{
\newpath
\moveto(50.59844906,59.89968657)
\curveto(51.32864497,57.94419428)(50.31631516,55.82072147)(48.33734716,55.15677732)
\curveto(46.35837915,54.49283317)(44.16216787,55.53983981)(43.43197196,57.4953321)
\curveto(42.70177606,59.4508244)(43.71410587,61.5742972)(45.69307387,62.23824136)
\curveto(47.67204188,62.90218551)(49.86825316,61.85517887)(50.59844906,59.89968657)
\closepath
}
}
{
\newrgbcolor{curcolor}{0 0 0}
\pscustom[linestyle=none,fillstyle=solid,fillcolor=curcolor]
{
\newpath
\moveto(77.04301771,80.81195297)
\curveto(77.77321362,78.85646067)(76.76088381,76.73298787)(74.7819158,76.06904371)
\curveto(72.8029478,75.40509956)(70.60673651,76.4521062)(69.87654061,78.4075985)
\curveto(69.1463447,80.36309079)(70.15867451,82.4865636)(72.13764252,83.15050775)
\curveto(74.11661052,83.8144519)(76.31282181,82.76744526)(77.04301771,80.81195297)
\closepath
}
}
{
\newrgbcolor{curcolor}{0 0 0}
\pscustom[linewidth=1.72947494,linecolor=curcolor]
{
\newpath
\moveto(77.04301771,80.81195297)
\curveto(77.77321362,78.85646067)(76.76088381,76.73298787)(74.7819158,76.06904371)
\curveto(72.8029478,75.40509956)(70.60673651,76.4521062)(69.87654061,78.4075985)
\curveto(69.1463447,80.36309079)(70.15867451,82.4865636)(72.13764252,83.15050775)
\curveto(74.11661052,83.8144519)(76.31282181,82.76744526)(77.04301771,80.81195297)
\closepath
}
}
{
\newrgbcolor{curcolor}{0 0 0}
\pscustom[linestyle=none,fillstyle=solid,fillcolor=curcolor]
{
\newpath
\moveto(118.2726809,86.70813173)
\curveto(119.00287681,84.75263943)(117.990547,82.62916663)(116.011579,81.96522247)
\curveto(114.03261099,81.30127832)(111.83639971,82.34828496)(111.1062038,84.30377726)
\curveto(110.3760079,86.25926955)(111.3883377,88.38274236)(113.36730571,89.04668651)
\curveto(115.34627372,89.71063066)(117.542485,88.66362402)(118.2726809,86.70813173)
\closepath
}
}
{
\newrgbcolor{curcolor}{0 0 0}
\pscustom[linewidth=1.72947494,linecolor=curcolor]
{
\newpath
\moveto(118.2726809,86.70813173)
\curveto(119.00287681,84.75263943)(117.990547,82.62916663)(116.011579,81.96522247)
\curveto(114.03261099,81.30127832)(111.83639971,82.34828496)(111.1062038,84.30377726)
\curveto(110.3760079,86.25926955)(111.3883377,88.38274236)(113.36730571,89.04668651)
\curveto(115.34627372,89.71063066)(117.542485,88.66362402)(118.2726809,86.70813173)
\closepath
}
}
{
\newrgbcolor{curcolor}{0 0 0}
\pscustom[linestyle=none,fillstyle=solid,fillcolor=curcolor]
{
\newpath
\moveto(153.43325555,103.49030862)
\curveto(154.16345145,101.53481632)(153.15112165,99.41134352)(151.17215364,98.74739936)
\curveto(149.19318563,98.08345521)(146.99697435,99.13046185)(146.26677845,101.08595415)
\curveto(145.53658254,103.04144644)(146.54891235,105.16491925)(148.52788036,105.8288634)
\curveto(150.50684836,106.49280755)(152.70305964,105.44580091)(153.43325555,103.49030862)
\closepath
}
}
{
\newrgbcolor{curcolor}{0 0 0}
\pscustom[linewidth=1.72947494,linecolor=curcolor]
{
\newpath
\moveto(153.43325555,103.49030862)
\curveto(154.16345145,101.53481632)(153.15112165,99.41134352)(151.17215364,98.74739936)
\curveto(149.19318563,98.08345521)(146.99697435,99.13046185)(146.26677845,101.08595415)
\curveto(145.53658254,103.04144644)(146.54891235,105.16491925)(148.52788036,105.8288634)
\curveto(150.50684836,106.49280755)(152.70305964,105.44580091)(153.43325555,103.49030862)
\closepath
}
}
{
\newrgbcolor{curcolor}{0 0 0}
\pscustom[linestyle=none,fillstyle=solid,fillcolor=curcolor]
{
\newpath
\moveto(190.64193345,98.00163609)
\curveto(191.37212936,96.04614379)(190.35979955,93.92267099)(188.38083155,93.25872684)
\curveto(186.40186354,92.59478268)(184.20565226,93.64178932)(183.47545635,95.59728162)
\curveto(182.74526045,97.55277392)(183.75759025,99.67624672)(185.73655826,100.34019088)
\curveto(187.71552627,101.00413503)(189.91173755,99.95712839)(190.64193345,98.00163609)
\closepath
}
}
{
\newrgbcolor{curcolor}{0 0 0}
\pscustom[linewidth=1.72947494,linecolor=curcolor]
{
\newpath
\moveto(190.64193345,98.00163609)
\curveto(191.37212936,96.04614379)(190.35979955,93.92267099)(188.38083155,93.25872684)
\curveto(186.40186354,92.59478268)(184.20565226,93.64178932)(183.47545635,95.59728162)
\curveto(182.74526045,97.55277392)(183.75759025,99.67624672)(185.73655826,100.34019088)
\curveto(187.71552627,101.00413503)(189.91173755,99.95712839)(190.64193345,98.00163609)
\closepath
}
}
{
\newrgbcolor{curcolor}{0 0 0}
\pscustom[linestyle=none,fillstyle=solid,fillcolor=curcolor]
{
\newpath
\moveto(257.41376743,99.14321022)
\curveto(258.14396334,97.18771793)(257.13163353,95.06424512)(255.15266552,94.40030097)
\curveto(253.17369752,93.73635682)(250.97748623,94.78336346)(250.24729033,96.73885575)
\curveto(249.51709442,98.69434805)(250.52942423,100.81782085)(252.50839224,101.48176501)
\curveto(254.48736024,102.14570916)(256.68357153,101.09870252)(257.41376743,99.14321022)
\closepath
}
}
{
\newrgbcolor{curcolor}{0 0 0}
\pscustom[linewidth=1.72947494,linecolor=curcolor]
{
\newpath
\moveto(257.41376743,99.14321022)
\curveto(258.14396334,97.18771793)(257.13163353,95.06424512)(255.15266552,94.40030097)
\curveto(253.17369752,93.73635682)(250.97748623,94.78336346)(250.24729033,96.73885575)
\curveto(249.51709442,98.69434805)(250.52942423,100.81782085)(252.50839224,101.48176501)
\curveto(254.48736024,102.14570916)(256.68357153,101.09870252)(257.41376743,99.14321022)
\closepath
}
}
{
\newrgbcolor{curcolor}{0 0 0}
\pscustom[linestyle=none,fillstyle=solid,fillcolor=curcolor]
{
\newpath
\moveto(226.6732125,98.95567244)
\curveto(227.40340841,97.00018015)(226.3910786,94.87670734)(224.41211059,94.21276319)
\curveto(222.43314259,93.54881904)(220.2369313,94.59582568)(219.5067354,96.55131797)
\curveto(218.77653949,98.50681027)(219.7888693,100.63028307)(221.76783731,101.29422723)
\curveto(223.74680531,101.95817138)(225.9430166,100.91116474)(226.6732125,98.95567244)
\closepath
}
}
{
\newrgbcolor{curcolor}{0 0 0}
\pscustom[linewidth=1.72947494,linecolor=curcolor]
{
\newpath
\moveto(226.6732125,98.95567244)
\curveto(227.40340841,97.00018015)(226.3910786,94.87670734)(224.41211059,94.21276319)
\curveto(222.43314259,93.54881904)(220.2369313,94.59582568)(219.5067354,96.55131797)
\curveto(218.77653949,98.50681027)(219.7888693,100.63028307)(221.76783731,101.29422723)
\curveto(223.74680531,101.95817138)(225.9430166,100.91116474)(226.6732125,98.95567244)
\closepath
}
}
{
\newrgbcolor{curcolor}{0 0 0}
\pscustom[linestyle=none,fillstyle=solid,fillcolor=curcolor]
{
\newpath
\moveto(290.99856429,85.26668274)
\curveto(291.72876019,83.31119044)(290.71643038,81.18771764)(288.73746238,80.52377348)
\curveto(286.75849437,79.85982933)(284.56228309,80.90683597)(283.83208718,82.86232827)
\curveto(283.10189128,84.81782056)(284.11422109,86.94129337)(286.09318909,87.60523752)
\curveto(288.0721571,88.26918167)(290.26836838,87.22217503)(290.99856429,85.26668274)
\closepath
}
}
{
\newrgbcolor{curcolor}{0 0 0}
\pscustom[linewidth=1.72947494,linecolor=curcolor]
{
\newpath
\moveto(290.99856429,85.26668274)
\curveto(291.72876019,83.31119044)(290.71643038,81.18771764)(288.73746238,80.52377348)
\curveto(286.75849437,79.85982933)(284.56228309,80.90683597)(283.83208718,82.86232827)
\curveto(283.10189128,84.81782056)(284.11422109,86.94129337)(286.09318909,87.60523752)
\curveto(288.0721571,88.26918167)(290.26836838,87.22217503)(290.99856429,85.26668274)
\closepath
}
}
{
\newrgbcolor{curcolor}{0 0 0}
\pscustom[linestyle=none,fillstyle=solid,fillcolor=curcolor]
{
\newpath
\moveto(326.10410526,64.11890662)
\curveto(326.83430116,62.16341432)(325.82197135,60.03994152)(323.84300335,59.37599737)
\curveto(321.86403534,58.71205321)(319.66782406,59.75905985)(318.93762815,61.71455215)
\curveto(318.20743225,63.67004445)(319.21976206,65.79351725)(321.19873006,66.4574614)
\curveto(323.17769807,67.12140556)(325.37390935,66.07439892)(326.10410526,64.11890662)
\closepath
}
}
{
\newrgbcolor{curcolor}{0 0 0}
\pscustom[linewidth=1.72947494,linecolor=curcolor]
{
\newpath
\moveto(326.10410526,64.11890662)
\curveto(326.83430116,62.16341432)(325.82197135,60.03994152)(323.84300335,59.37599737)
\curveto(321.86403534,58.71205321)(319.66782406,59.75905985)(318.93762815,61.71455215)
\curveto(318.20743225,63.67004445)(319.21976206,65.79351725)(321.19873006,66.4574614)
\curveto(323.17769807,67.12140556)(325.37390935,66.07439892)(326.10410526,64.11890662)
\closepath
}
}
{
\newrgbcolor{curcolor}{0 0 0}
\pscustom[linestyle=none,fillstyle=solid,fillcolor=curcolor]
{
\newpath
\moveto(367.85028065,33.79506196)
\curveto(368.58047655,31.83956966)(367.56814675,29.71609686)(365.58917874,29.0521527)
\curveto(363.61021074,28.38820855)(361.41399945,29.43521519)(360.68380355,31.39070749)
\curveto(359.95360764,33.34619978)(360.96593745,35.46967259)(362.94490546,36.13361674)
\curveto(364.92387346,36.79756089)(367.12008475,35.75055425)(367.85028065,33.79506196)
\closepath
}
}
{
\newrgbcolor{curcolor}{0 0 0}
\pscustom[linewidth=1.72947494,linecolor=curcolor]
{
\newpath
\moveto(367.85028065,33.79506196)
\curveto(368.58047655,31.83956966)(367.56814675,29.71609686)(365.58917874,29.0521527)
\curveto(363.61021074,28.38820855)(361.41399945,29.43521519)(360.68380355,31.39070749)
\curveto(359.95360764,33.34619978)(360.96593745,35.46967259)(362.94490546,36.13361674)
\curveto(364.92387346,36.79756089)(367.12008475,35.75055425)(367.85028065,33.79506196)
\closepath
}
}
{
\newrgbcolor{curcolor}{0 0 0}
\pscustom[linestyle=none,fillstyle=solid,fillcolor=curcolor]
{
\newpath
\moveto(420.55815087,9.19215182)
\curveto(421.28834677,7.23665952)(420.27601696,5.11318672)(418.29704896,4.44924257)
\curveto(416.31808095,3.78529841)(414.12186967,4.83230505)(413.39167376,6.78779735)
\curveto(412.66147786,8.74328965)(413.67380767,10.86676245)(415.65277567,11.5307066)
\curveto(417.63174368,12.19465076)(419.82795496,11.14764412)(420.55815087,9.19215182)
\closepath
}
}
{
\newrgbcolor{curcolor}{0 0 0}
\pscustom[linewidth=1.72947494,linecolor=curcolor]
{
\newpath
\moveto(420.55815087,9.19215182)
\curveto(421.28834677,7.23665952)(420.27601696,5.11318672)(418.29704896,4.44924257)
\curveto(416.31808095,3.78529841)(414.12186967,4.83230505)(413.39167376,6.78779735)
\curveto(412.66147786,8.74328965)(413.67380767,10.86676245)(415.65277567,11.5307066)
\curveto(417.63174368,12.19465076)(419.82795496,11.14764412)(420.55815087,9.19215182)
\closepath
}
}
{
\newrgbcolor{curcolor}{0 0 0}
\pscustom[linestyle=none,fillstyle=solid,fillcolor=curcolor]
{
\newpath
\moveto(526.03158984,30.21421232)
\curveto(526.76178574,28.25872003)(525.74945593,26.13524722)(523.77048793,25.47130307)
\curveto(521.79151992,24.80735892)(519.59530864,25.85436556)(518.86511273,27.80985785)
\curveto(518.13491683,29.76535015)(519.14724664,31.88882295)(521.12621464,32.55276711)
\curveto(523.10518265,33.21671126)(525.30139393,32.16970462)(526.03158984,30.21421232)
\closepath
}
}
{
\newrgbcolor{curcolor}{0 0 0}
\pscustom[linewidth=1.72947494,linecolor=curcolor]
{
\newpath
\moveto(526.03158984,30.21421232)
\curveto(526.76178574,28.25872003)(525.74945593,26.13524722)(523.77048793,25.47130307)
\curveto(521.79151992,24.80735892)(519.59530864,25.85436556)(518.86511273,27.80985785)
\curveto(518.13491683,29.76535015)(519.14724664,31.88882295)(521.12621464,32.55276711)
\curveto(523.10518265,33.21671126)(525.30139393,32.16970462)(526.03158984,30.21421232)
\closepath
}
}
{
\newrgbcolor{curcolor}{0 0 0}
\pscustom[linestyle=none,fillstyle=solid,fillcolor=curcolor]
{
\newpath
\moveto(478.44235143,21.42853669)
\curveto(479.17254733,19.4730444)(478.16021752,17.34957159)(476.18124952,16.68562744)
\curveto(474.20228151,16.02168329)(472.00607023,17.06868993)(471.27587432,19.02418222)
\curveto(470.54567842,20.97967452)(471.55800823,23.10314732)(473.53697623,23.76709148)
\curveto(475.51594424,24.43103563)(477.71215552,23.38402899)(478.44235143,21.42853669)
\closepath
}
}
{
\newrgbcolor{curcolor}{0 0 0}
\pscustom[linewidth=1.72947494,linecolor=curcolor]
{
\newpath
\moveto(478.44235143,21.42853669)
\curveto(479.17254733,19.4730444)(478.16021752,17.34957159)(476.18124952,16.68562744)
\curveto(474.20228151,16.02168329)(472.00607023,17.06868993)(471.27587432,19.02418222)
\curveto(470.54567842,20.97967452)(471.55800823,23.10314732)(473.53697623,23.76709148)
\curveto(475.51594424,24.43103563)(477.71215552,23.38402899)(478.44235143,21.42853669)
\closepath
}
}
{
\newrgbcolor{curcolor}{0 0 0}
\pscustom[linestyle=none,fillstyle=solid,fillcolor=curcolor]
{
\newpath
\moveto(556.34953854,62.52885091)
\curveto(557.07973444,60.57335862)(556.06740464,58.44988581)(554.08843663,57.78594166)
\curveto(552.10946862,57.12199751)(549.91325734,58.16900415)(549.18306144,60.12449645)
\curveto(548.45286553,62.07998874)(549.46519534,64.20346155)(551.44416334,64.8674057)
\curveto(553.42313135,65.53134985)(555.61934263,64.48434321)(556.34953854,62.52885091)
\closepath
}
}
{
\newrgbcolor{curcolor}{0 0 0}
\pscustom[linewidth=1.72947494,linecolor=curcolor]
{
\newpath
\moveto(556.34953854,62.52885091)
\curveto(557.07973444,60.57335862)(556.06740464,58.44988581)(554.08843663,57.78594166)
\curveto(552.10946862,57.12199751)(549.91325734,58.16900415)(549.18306144,60.12449645)
\curveto(548.45286553,62.07998874)(549.46519534,64.20346155)(551.44416334,64.8674057)
\curveto(553.42313135,65.53134985)(555.61934263,64.48434321)(556.34953854,62.52885091)
\closepath
}
}
{
\newrgbcolor{curcolor}{0 0 0}
\pscustom[linestyle=none,fillstyle=solid,fillcolor=curcolor]
{
\newpath
\moveto(563.82957234,102.58803672)
\curveto(564.55976825,100.63254442)(563.54743844,98.50907162)(561.56847044,97.84512747)
\curveto(559.58950243,97.18118331)(557.39329115,98.22818995)(556.66309524,100.18368225)
\curveto(555.93289934,102.13917455)(556.94522914,104.26264735)(558.92419715,104.9265915)
\curveto(560.90316516,105.59053566)(563.09937644,104.54352902)(563.82957234,102.58803672)
\closepath
}
}
{
\newrgbcolor{curcolor}{0 0 0}
\pscustom[linewidth=1.72947494,linecolor=curcolor]
{
\newpath
\moveto(563.82957234,102.58803672)
\curveto(564.55976825,100.63254442)(563.54743844,98.50907162)(561.56847044,97.84512747)
\curveto(559.58950243,97.18118331)(557.39329115,98.22818995)(556.66309524,100.18368225)
\curveto(555.93289934,102.13917455)(556.94522914,104.26264735)(558.92419715,104.9265915)
\curveto(560.90316516,105.59053566)(563.09937644,104.54352902)(563.82957234,102.58803672)
\closepath
}
}
{
\newrgbcolor{curcolor}{0 0 0}
\pscustom[linewidth=3.02362209,linecolor=curcolor]
{
\newpath
\moveto(0.17175307,85.3538776)
\lineto(356.44963654,126.09279288)
}
}
{
\newrgbcolor{curcolor}{0 0 0}
\pscustom[linewidth=3.02362209,linecolor=curcolor]
{
\newpath
\moveto(250.12209638,13.82094784)
\lineto(598.54081134,1.5108754)
}
\rput[bl](600,130){$\F$}
\rput[bl](5,25){$\tilde{M}$}
\rput[bl](143,75){$p$}
\rput[bl](412,19){$q$}
\rput[bl](365,112){$M_p$}
\rput[bl](210,5){$M_q$}
}
\end{pspicture}
\caption{Osculating vacua of a discrete spacetime.}
\label{figosculate-discrete}
\end{figure}%
Apart from providing a differential calculus, the constructions
in~\cite{lcalc} also explained in the examples of the Gau{\ss} divergence
theorem and various versions of Stokes' theorem
how corresponding conservation laws
can be formulated in terms of surface layer integrals.
This gives the hope that the Einstein equations can be extended
to the non-smooth setting in such a way that the contracted second
Bianchi identities hold in an integral sense. But the details are not
straightforward and still need to be worked out carefully.

In the discussion so far, the spacetime~$\tilde{M}$ was either smooth
or a discrete approximation of a smooth spacetime.
However, it is expected that minimizers of the causal action principle
in general have a more complicated structure.
Intuitively speaking, the support of the measure should be ``thickened''
in order to account for additional ``internal degrees of freedom'' or
``microscopic spacetime fluctuations''. Moreover, the measure could have
discrete or non-smooth components. Figure~\ref{figquant} gives
an impression of how such a ``quantum spacetime'' could look like.
\begin{figure}[tb]
\psset{xunit=.9pt,yunit=.9pt,runit=.9pt}
\begin{pspicture}(502.83300781,127.05599976)
{
\newrgbcolor{curcolor}{0.80000001 0.80000001 0.80000001}
\pscustom[linestyle=none,fillstyle=solid,fillcolor=curcolor]
{
\newpath
\moveto(297.49966667,45.7224442)
\lineto(301.94411111,43.94466642)
\lineto(310.16633333,41.05577753)
\lineto(319.94411111,39.05577753)
\lineto(330.61077778,37.7224442)
\lineto(344.16633333,37.94466642)
\lineto(357.72188889,38.38911087)
\lineto(368.833,39.94466642)
\lineto(385.94411111,42.83355531)
\lineto(397.49966667,45.50022198)
\lineto(407.94411111,48.38911087)
\lineto(424.61077778,52.61133309)
\lineto(435.27744444,55.05577753)
\lineto(448.16633333,57.7224442)
\lineto(463.49966667,60.61133309)
\lineto(474.38855556,62.16688864)
\lineto(482.61077778,62.83355531)
\lineto(482.833,91.94466642)
\lineto(466.61077778,95.27799976)
\lineto(449.27744444,99.05577753)
\lineto(431.27744444,102.16688864)
\lineto(417.94411111,103.7224442)
\lineto(397.72188889,103.7224442)
\lineto(383.05522222,102.16688864)
\lineto(366.38855556,97.7224442)
\lineto(349.94411111,91.94466642)
\lineto(335.72188889,88.16688864)
\lineto(316.833,85.05577753)
\lineto(304.16633333,87.05577753)
\lineto(297.72188889,89.50022198)
\closepath
}
}
{
\newrgbcolor{curcolor}{0.80000001 0.80000001 0.80000001}
\pscustom[linewidth=1,linecolor=curcolor]
{
\newpath
\moveto(297.49966667,45.7224442)
\lineto(301.94411111,43.94466642)
\lineto(310.16633333,41.05577753)
\lineto(319.94411111,39.05577753)
\lineto(330.61077778,37.7224442)
\lineto(344.16633333,37.94466642)
\lineto(357.72188889,38.38911087)
\lineto(368.833,39.94466642)
\lineto(385.94411111,42.83355531)
\lineto(397.49966667,45.50022198)
\lineto(407.94411111,48.38911087)
\lineto(424.61077778,52.61133309)
\lineto(435.27744444,55.05577753)
\lineto(448.16633333,57.7224442)
\lineto(463.49966667,60.61133309)
\lineto(474.38855556,62.16688864)
\lineto(482.61077778,62.83355531)
\lineto(482.833,91.94466642)
\lineto(466.61077778,95.27799976)
\lineto(449.27744444,99.05577753)
\lineto(431.27744444,102.16688864)
\lineto(417.94411111,103.7224442)
\lineto(397.72188889,103.7224442)
\lineto(383.05522222,102.16688864)
\lineto(366.38855556,97.7224442)
\lineto(349.94411111,91.94466642)
\lineto(335.72188889,88.16688864)
\lineto(316.833,85.05577753)
\lineto(304.16633333,87.05577753)
\lineto(297.72188889,89.50022198)
\closepath
}
}
{
\newrgbcolor{curcolor}{0 0 0}
\pscustom[linewidth=1,linecolor=curcolor]
{
\newpath
\moveto(372.81145805,24.83355531)
\curveto(363.24396272,25.87950029)(358.2370151,27.13811288)(353.99962848,29.55875869)
\curveto(349.76224186,31.9794045)(347.10803983,34.31197785)(344.16220465,37.05577753)
}
}
{
\newrgbcolor{curcolor}{0 0 0}
\pscustom[linestyle=none,fillstyle=solid,fillcolor=curcolor]
{
\newpath
\moveto(344.16220465,37.05577753)
\lineto(347.3384655,30.92295658)
\lineto(348.92168245,32.62277008)
\lineto(350.5048994,34.32258359)
\closepath
}
}
{
\newrgbcolor{curcolor}{0 0 0}
\pscustom[linewidth=1,linecolor=curcolor]
{
\newpath
\moveto(344.16220465,37.05577753)
\lineto(347.3384655,30.92295658)
\lineto(348.92168245,32.62277008)
\lineto(350.5048994,34.32258359)
\closepath
}
}
{
\newrgbcolor{curcolor}{0 0 0}
\pscustom[linewidth=1,linecolor=curcolor]
{
\newpath
\moveto(297.2099,81.89609976)
\curveto(297.8703,82.95429976)(298.8906,83.78339976)(300.0615,84.21339976)
\curveto(301.2324,84.64339976)(302.5467,84.67169976)(303.735,84.29239976)
\curveto(304.9233,83.91309976)(305.9783,83.12869976)(306.6836,82.09989976)
\curveto(307.3889,81.07109976)(307.7403,79.80429976)(307.6656,78.55909976)
\curveto(307.5676,76.93029976)(306.7739,75.42289976)(305.8403,74.08459976)
\curveto(304.9068,72.74639976)(303.8155,71.51529976)(302.9939,70.10549976)
\curveto(301.4026,67.37489976)(300.9184,64.14149976)(300.7693,60.98449976)
\curveto(300.6397,58.24059976)(300.7393,55.48869976)(300.9918,52.75339976)
\curveto(301.1008,51.56999976)(301.2393,50.38329976)(301.5584,49.23849976)
\curveto(301.8775,48.09369976)(302.3853,46.98549976)(303.1652,46.08879976)
\curveto(303.9451,45.19199976)(305.0138,44.51859976)(306.1895,44.34469976)
\curveto(306.7773,44.25769976)(307.3854,44.29669976)(307.952,44.47559976)
\curveto(308.5186,44.65469976)(309.0422,44.97559976)(309.4453,45.41219976)
\curveto(310.1485,46.17379976)(310.4485,47.25239976)(310.3935,48.28759976)
\curveto(310.3385,49.32279976)(309.9552,50.31849976)(309.4377,51.21669976)
\curveto(308.4026,53.01309976)(306.8346,54.46929976)(305.8859,56.31279976)
\curveto(305.062,57.91369976)(304.7455,59.74959976)(304.835,61.54779976)
\curveto(304.925,63.34609976)(305.4081,65.10989976)(306.1084,66.76859976)
\curveto(307.1198,69.16419976)(308.5747,71.34559976)(310.1127,73.44239976)
\curveto(310.8623,74.46439976)(311.6376,75.47459976)(312.5469,76.35759976)
\curveto(313.4562,77.24059976)(314.5085,77.99729976)(315.697,78.43779976)
\curveto(316.8855,78.87839976)(318.219,78.98669976)(319.4263,78.60079976)
\curveto(320.0299,78.40779976)(320.5965,78.09319976)(321.0667,77.66829976)
\curveto(321.5369,77.24339976)(321.9093,76.70759976)(322.1257,76.11199976)
\curveto(322.3485,75.49879976)(322.4038,74.83019976)(322.3211,74.18309976)
\curveto(322.2381,73.53599976)(322.0202,72.90989976)(321.7153,72.33319976)
\curveto(321.1054,71.17969976)(320.1641,70.23729976)(319.179,69.38169976)
\curveto(318.1939,68.52609976)(317.1478,67.73509976)(316.2496,66.78869976)
\curveto(315.3514,65.84229976)(314.5952,64.71109976)(314.3395,63.43159976)
\curveto(314.1307,62.38689976)(314.2695,61.28749976)(314.6534,60.29389976)
\curveto(315.0376,59.30029976)(315.6602,58.40909976)(316.4017,57.64419976)
\curveto(317.8847,56.11439976)(319.806,55.10049976)(321.6808,54.08819976)
\curveto(323.3914,53.16459976)(325.121,52.21029976)(327.0193,51.79119976)
\curveto(327.9684,51.58169976)(328.9543,51.51059976)(329.9158,51.65289976)
\curveto(330.8773,51.79519976)(331.8143,52.15649976)(332.5814,52.75339976)
\curveto(333.6038,53.54899976)(334.2847,54.73819976)(334.5789,55.99979976)
\curveto(334.8731,57.26139976)(334.7975,58.59119976)(334.4889,59.84929976)
\curveto(333.8717,62.36559976)(332.3722,64.55619976)(331.0239,66.76859976)
\curveto(329.6568,69.01179976)(328.411,71.35709976)(327.6945,73.88439976)
\curveto(326.978,76.41179976)(326.8152,79.14789976)(327.6035,81.65379976)
\curveto(327.9976,82.90669976)(328.6277,84.09009976)(329.4853,85.08489976)
\curveto(330.3429,86.07969976)(331.4305,86.88269976)(332.6526,87.36399976)
\curveto(333.8747,87.84529976)(335.2309,87.99899976)(336.5216,87.75559976)
\curveto(337.8123,87.51219976)(339.0313,86.86549976)(339.9221,85.90029976)
\curveto(340.9551,84.78099976)(341.5136,83.28319976)(341.7023,81.77179976)
\curveto(341.891,80.26039976)(341.7323,78.72619976)(341.4793,77.22429976)
\curveto(340.5467,71.69709976)(338.3431,66.45939976)(337.197,60.97249976)
\curveto(336.6239,58.22899976)(336.3176,55.40899976)(336.5824,52.61889976)
\curveto(336.8472,49.82879976)(337.704,47.06169976)(339.3061,44.76209976)
\curveto(340.1519,43.54799976)(341.2056,42.46799976)(342.4528,41.67179976)
\curveto(343.7,40.87559976)(345.1453,40.37009976)(346.6238,40.30959976)
\curveto(348.1022,40.24859976)(349.6103,40.64709976)(350.8151,41.50609976)
\curveto(352.0199,42.36519976)(352.9005,43.69439976)(353.1332,45.15569976)
\curveto(353.325,46.35969976)(353.0812,47.60669976)(352.5867,48.72119976)
\curveto(352.0924,49.83569976)(351.357,50.82819976)(350.5471,51.73949976)
\curveto(348.9272,53.56209976)(346.9722,55.11879976)(345.7061,57.20269976)
\curveto(344.4435,59.28099976)(343.9555,61.77519976)(344.0827,64.20359976)
\curveto(344.2099,66.63199976)(344.9269,69.00189976)(345.9286,71.21779976)
\curveto(346.6908,72.90379976)(347.6782,74.57769976)(349.2334,75.57999976)
\curveto(350.011,76.08109976)(350.9211,76.39559976)(351.8462,76.39849976)
\curveto(352.7713,76.39849976)(353.708,76.07819976)(354.3821,75.44459976)
\curveto(355.1468,74.72589976)(355.5249,73.66079976)(355.563,72.61209976)
\curveto(355.601,71.56339976)(355.3268,70.52519976)(354.9458,69.54739976)
\curveto(354.5648,68.56959976)(354.0766,67.63729976)(353.6609,66.67369976)
\curveto(353.2452,65.71009976)(352.8991,64.70089976)(352.8249,63.65409976)
\curveto(352.7109,62.04629976)(353.269,60.41169976)(354.2977,59.17089976)
\curveto(355.3264,57.93009976)(356.8075,57.08569976)(358.3888,56.77369976)
\curveto(359.9701,56.46169976)(361.6415,56.67369976)(363.1199,57.31569976)
\curveto(364.5983,57.95769976)(365.8827,59.02049976)(366.84,60.31719976)
\curveto(367.9724,61.85119976)(368.6462,63.69539976)(368.9325,65.58049976)
\curveto(369.2188,67.46559976)(369.1288,69.39359976)(368.8323,71.27709976)
\curveto(368.2394,75.04409976)(366.8367,78.63199976)(365.9502,82.34089976)
\curveto(365.4352,84.49559976)(365.0957,86.75289976)(365.5596,88.91919976)
\curveto(365.7915,90.00229976)(366.2273,91.04959976)(366.8949,91.93349976)
\curveto(367.5625,92.81739976)(368.4678,93.53279976)(369.5096,93.90899976)
\curveto(370.6976,94.33799976)(372.0307,94.31059976)(373.2283,93.90899976)
\curveto(374.4259,93.50739976)(375.488,92.74449976)(376.3171,91.79159976)
\curveto(377.9752,89.88579976)(378.6703,87.30939976)(378.853,84.78989976)
\curveto(379.0972,81.42359976)(378.5188,78.05979976)(378.2544,74.69499976)
\curveto(377.99,71.33019976)(378.0675,67.81469976)(379.5204,64.76829976)
\curveto(380.3391,63.05159976)(381.587,61.53369976)(383.1461,60.44439976)
\curveto(384.7052,59.35509976)(386.5754,58.70269976)(388.4764,58.64219976)
\curveto(390.3774,58.58119976)(392.3014,59.12169976)(393.8623,60.20849976)
\curveto(395.4232,61.29529976)(396.6052,62.93039976)(397.0949,64.76829976)
\curveto(397.4749,66.19459976)(397.44,67.71729976)(397.0839,69.14979976)
\curveto(396.7278,70.58229976)(396.0576,71.92679976)(395.2026,73.13009976)
\curveto(393.4927,75.53669976)(391.0871,77.35249976)(388.6413,79.00589976)
\curveto(387.4848,79.78769976)(386.2994,80.55129976)(385.3279,81.55379976)
\curveto(384.3564,82.55629976)(383.603,83.84109976)(383.5246,85.23489976)
\curveto(383.4466,86.62289976)(384.0663,88.00939976)(385.0738,88.96739976)
\curveto(386.0812,89.92539976)(387.4466,90.46749976)(388.8317,90.58649976)
\curveto(390.2168,90.70549976)(391.619,90.41849976)(392.9013,89.88139976)
\curveto(394.1836,89.34429976)(395.3524,88.56329976)(396.4275,87.68189976)
\curveto(400.8712,84.03889976)(403.7372,78.60159976)(404.4952,72.90569976)
\curveto(405.2532,67.20979976)(403.9509,61.30329976)(401.0992,56.31469976)
\curveto(399.4577,53.44299976)(397.3057,50.85079976)(394.7061,48.80549976)
\curveto(392.1065,46.76019976)(389.0552,45.26999976)(385.8165,44.59799976)
\curveto(382.5778,43.92599976)(379.1542,44.08529976)(376.0299,45.17159976)
\curveto(372.9056,46.25789976)(370.096,48.28309976)(368.1747,50.97559976)
\curveto(365.93,54.12129976)(364.9572,58.19319976)(365.7393,61.97769976)
\curveto(366.5214,65.76219976)(369.1061,69.17529976)(372.624,70.77479976)
\curveto(374.562,71.65599976)(376.7222,71.98929976)(378.85,71.92059976)
\curveto(380.9778,71.85159976)(383.0785,71.39159976)(385.1065,70.74399976)
\curveto(389.1625,69.44879976)(392.9486,67.41049976)(396.9954,66.08679976)
\curveto(399.0188,65.42499976)(401.1135,64.94279976)(403.2401,64.84419976)
\curveto(405.3667,64.74519976)(407.5308,65.04069976)(409.4869,65.88079976)
\curveto(411.443,66.72089976)(413.1818,68.12649976)(414.2597,69.96239976)
\curveto(415.3376,71.79819976)(415.7139,74.07119976)(415.1143,76.11389976)
\curveto(414.7259,77.43699976)(413.9541,78.62209976)(413.0381,79.65279976)
\curveto(412.1221,80.68349976)(411.0606,81.57279976)(410.003,82.45759976)
\curveto(408.9454,83.34239976)(407.8831,84.23079976)(406.9656,85.26009976)
\curveto(406.0481,86.28939976)(405.2738,87.47259976)(404.881,88.79429976)
\curveto(404.5032,90.06549976)(404.4936,91.44789976)(404.881,92.71619976)
\curveto(405.2684,93.98449976)(406.0549,95.13179976)(407.1142,95.92969976)
\curveto(408.1735,96.72759976)(409.5024,97.16789976)(410.8281,97.13359976)
\curveto(412.1538,97.09959976)(413.4663,96.58529976)(414.4452,95.69059976)
\curveto(415.5383,94.69139976)(416.1858,93.25239976)(416.3397,91.77939976)
\curveto(416.4936,90.30639976)(416.1752,88.80549976)(415.5575,87.45949976)
\curveto(414.7314,85.65939976)(413.4,84.14379976)(412.3915,82.43919976)
\curveto(411.8873,81.58689976)(411.4611,80.67869976)(411.2406,79.71329976)
\curveto(411.0201,78.74789976)(411.0131,77.71929976)(411.3306,76.78129976)
\curveto(411.7636,75.50189976)(412.7724,74.47579976)(413.9371,73.79179976)
\curveto(415.1018,73.10779976)(416.4202,72.73279976)(417.7376,72.43479976)
\curveto(419.055,72.13679976)(420.3919,71.90779976)(421.6746,71.48469976)
\curveto(422.9573,71.06159976)(424.2014,70.42709976)(425.1233,69.43999976)
\curveto(426.0993,68.39499976)(426.6624,66.98729976)(426.737,65.55929976)
\curveto(426.812,64.13129976)(426.4077,62.69139976)(425.653,61.47689976)
\curveto(424.8983,60.26239976)(423.8014,59.27239976)(422.5448,58.58999976)
\curveto(421.2882,57.90759976)(419.8757,57.52859976)(418.4494,57.42699976)
\curveto(416.6363,57.29779976)(414.7649,57.62669976)(413.1946,58.54219976)
\curveto(411.6243,59.45769976)(410.3831,60.99009976)(409.9959,62.76609976)
\curveto(409.7854,63.73169976)(409.827,64.74659976)(410.0779,65.70249976)
\curveto(410.3288,66.65839976)(410.786,67.55559976)(411.3804,68.34509976)
\curveto(412.5693,69.92419976)(414.2814,71.04979976)(416.1048,71.81279976)
\curveto(419.7516,73.33879976)(423.8313,73.49569976)(427.5703,74.77909976)
\curveto(429.0677,75.29309976)(430.5102,75.99079976)(431.7811,76.93479976)
\curveto(433.052,77.87879976)(434.1488,79.07509976)(434.8864,80.47589976)
\curveto(435.624,81.87669976)(435.9916,83.48619976)(435.831,85.06109976)
\curveto(435.6704,86.63609976)(434.9643,88.16729976)(433.7993,89.23919976)
\curveto(432.6797,90.26939976)(431.1686,90.84449976)(429.6479,90.89159976)
\curveto(428.1272,90.93859976)(426.6058,90.47039976)(425.3331,89.63679976)
\curveto(424.0604,88.80319976)(423.0343,87.61419976)(422.3289,86.26619976)
\curveto(421.6235,84.91819976)(421.2335,83.41499976)(421.1189,81.89789976)
\curveto(420.8477,78.30619976)(422.1682,74.63209976)(424.6289,72.00169976)
\curveto(427.0896,69.37129976)(430.6462,67.81239976)(434.2474,67.73699976)
\curveto(437.8486,67.66199976)(441.4503,69.05829976)(444.0927,71.50609976)
\curveto(446.7351,73.95389976)(448.3982,77.41919976)(448.7042,81.00809976)
\curveto(448.9491,83.88029976)(448.3423,86.81239976)(447.0265,89.37719976)
\curveto(445.7107,91.94199976)(443.6952,94.13559976)(441.2806,95.71019976)
\curveto(436.4515,98.85939976)(430.1146,99.38989976)(424.6783,97.47029976)
\curveto(423.4743,97.04519976)(422.3009,96.50279976)(421.2778,95.73879976)
\curveto(420.2547,94.97479976)(419.3845,93.97759976)(418.8943,92.79859976)
\curveto(418.3971,91.60279976)(418.3095,90.24989976)(418.5934,88.98639976)
\curveto(418.8773,87.72289976)(419.5246,86.55069976)(420.4014,85.59769976)
\curveto(422.155,83.69169976)(424.7622,82.71629976)(427.3478,82.56529976)
\curveto(430.0243,82.40899976)(432.6891,83.06259976)(435.37,83.02879976)
\curveto(436.7104,83.01179976)(438.0652,82.81749976)(439.2952,82.28449976)
\curveto(440.5252,81.75149976)(441.6264,80.85849976)(442.2528,79.67329976)
\curveto(442.7926,78.65189976)(442.9578,77.46539976)(442.9069,76.31129976)
\curveto(442.8559,75.15719976)(442.5997,74.02309976)(442.3373,72.89799976)
\curveto(442.0749,71.77299976)(441.8045,70.64339976)(441.7248,69.49089976)
\curveto(441.6448,68.33839976)(441.7648,67.15019976)(442.2528,66.10309976)
\curveto(442.749,65.03829976)(443.6194,64.15799976)(444.6646,63.62179976)
\curveto(445.7098,63.08559976)(446.9221,62.89179976)(448.0869,63.04389976)
\curveto(449.2517,63.19599976)(450.3656,63.68949976)(451.2813,64.42529976)
\curveto(452.197,65.16109976)(452.9146,66.13509976)(453.376,67.21539976)
\curveto(453.9086,68.46239976)(454.1003,69.83789976)(454.058,71.19319976)
\curveto(454.016,72.54849976)(453.7455,73.88809976)(453.3846,75.19509976)
\curveto(452.6627,77.80919976)(451.5732,80.33129976)(451.1513,83.01019976)
\curveto(450.8391,84.99249976)(450.9113,87.07469976)(451.665,88.93459976)
\curveto(452.0419,89.86449976)(452.5874,90.72999976)(453.2935,91.44289976)
\curveto(453.9996,92.15579976)(454.8679,92.71359976)(455.823,93.02099976)
\curveto(456.7458,93.31799976)(457.7438,93.37799976)(458.6957,93.19449976)
\curveto(459.6476,93.01099976)(460.5516,92.58469976)(461.2992,91.96749976)
\curveto(462.7944,90.73319976)(463.6043,88.71879976)(463.3868,86.79209976)
\curveto(463.1685,84.85879976)(462.0082,83.15449976)(460.667,81.74499976)
\curveto(459.3258,80.33549976)(457.7738,79.13329976)(456.4904,77.67109976)
\curveto(455.3698,76.39439976)(454.4577,74.90989976)(453.998,73.27459976)
\curveto(453.5383,71.63919976)(453.5526,69.84569976)(454.2043,68.27689976)
\curveto(454.5302,67.49249976)(455.012,66.77079976)(455.6276,66.18559976)
\curveto(456.2432,65.60029976)(456.9931,65.15349976)(457.8076,64.91239976)
\curveto(458.622,64.67129976)(459.5002,64.63859976)(460.3252,64.84039976)
\curveto(461.1503,65.04219976)(461.9192,65.48059976)(462.4969,66.10319976)
\curveto(463.1737,66.83269976)(463.5723,67.79329976)(463.7215,68.77719976)
\curveto(463.8707,69.76109976)(463.7815,70.76939976)(463.5736,71.74249976)
\curveto(463.1572,73.68869976)(462.2736,75.50879976)(461.8295,77.44879976)
\curveto(461.4519,79.09839976)(461.4094,80.87649976)(462.0603,82.43859976)
\curveto(462.3858,83.21959976)(462.8835,83.93409976)(463.5343,84.47479976)
\curveto(464.1852,85.01549976)(464.9916,85.37699976)(465.8339,85.45739976)
\curveto(466.571,85.52739976)(467.3245,85.38239976)(467.9923,85.06199976)
\curveto(468.6601,84.74209976)(469.242,84.25079976)(469.6878,83.65949976)
\curveto(470.5794,82.47699976)(470.903,80.92149976)(470.728,79.45089976)
\curveto(470.5062,77.58669976)(469.5479,75.89349976)(469.1034,74.06949976)
\curveto(468.8812,73.15749976)(468.7898,72.19709976)(468.9905,71.28019976)
\curveto(469.1912,70.36329976)(469.7091,69.49179976)(470.5056,68.99519976)
\curveto(471.145,68.59649976)(471.9347,68.46049976)(472.6788,68.57919976)
\curveto(473.4229,68.69789976)(474.1192,69.06329976)(474.6711,69.57629976)
\curveto(475.2231,70.08929976)(475.6335,70.74559976)(475.8964,71.45179976)
\curveto(476.1593,72.15799976)(476.2776,72.91339976)(476.2896,73.66689976)
\curveto(476.3326,76.36739976)(475.0476,78.89049976)(474.4415,81.52249976)
\curveto(474.1384,82.83849976)(474.0061,84.21909976)(474.2983,85.53759976)
\curveto(474.4444,86.19679976)(474.6969,86.83509976)(475.0679,87.39929976)
\curveto(475.4389,87.96349976)(475.93,88.45229976)(476.5121,88.79439976)
\curveto(477.1573,89.17349976)(477.9102,89.36689976)(478.6582,89.34569976)
\curveto(479.4062,89.32469976)(480.147,89.08869976)(480.7696,88.67359976)
\curveto(481.3923,88.25849976)(481.8948,87.66539976)(482.2021,86.98309976)
\curveto(482.5094,86.30079976)(482.6204,85.53139976)(482.5185,84.78999976)
}
}
{
\newrgbcolor{curcolor}{0 0 0}
\pscustom[linestyle=none,fillstyle=solid,fillcolor=curcolor]
{
\newpath
\moveto(400.53783398,99.80778491)
\curveto(400.53783398,99.25550016)(400.09011873,98.80778491)(399.53783398,98.80778491)
\curveto(398.98554923,98.80778491)(398.53783398,99.25550016)(398.53783398,99.80778491)
\curveto(398.53783398,100.36006966)(398.98554923,100.80778491)(399.53783398,100.80778491)
\curveto(400.09011873,100.80778491)(400.53783398,100.36006966)(400.53783398,99.80778491)
\closepath
}
}
{
\newrgbcolor{curcolor}{0 0 0}
\pscustom[linewidth=0,linecolor=curcolor]
{
\newpath
\moveto(400.53783398,99.80778491)
\curveto(400.53783398,99.25550016)(400.09011873,98.80778491)(399.53783398,98.80778491)
\curveto(398.98554923,98.80778491)(398.53783398,99.25550016)(398.53783398,99.80778491)
\curveto(398.53783398,100.36006966)(398.98554923,100.80778491)(399.53783398,100.80778491)
\curveto(400.09011873,100.80778491)(400.53783398,100.36006966)(400.53783398,99.80778491)
\closepath
}
}
{
\newrgbcolor{curcolor}{0 0 0}
\pscustom[linestyle=none,fillstyle=solid,fillcolor=curcolor]
{
\newpath
\moveto(394.20946484,98.22563159)
\curveto(394.20946484,97.67334684)(393.76174959,97.22563159)(393.20946484,97.22563159)
\curveto(392.65718009,97.22563159)(392.20946484,97.67334684)(392.20946484,98.22563159)
\curveto(392.20946484,98.77791634)(392.65718009,99.22563159)(393.20946484,99.22563159)
\curveto(393.76174959,99.22563159)(394.20946484,98.77791634)(394.20946484,98.22563159)
\closepath
}
}
{
\newrgbcolor{curcolor}{0 0 0}
\pscustom[linewidth=0,linecolor=curcolor]
{
\newpath
\moveto(394.20946484,98.22563159)
\curveto(394.20946484,97.67334684)(393.76174959,97.22563159)(393.20946484,97.22563159)
\curveto(392.65718009,97.22563159)(392.20946484,97.67334684)(392.20946484,98.22563159)
\curveto(392.20946484,98.77791634)(392.65718009,99.22563159)(393.20946484,99.22563159)
\curveto(393.76174959,99.22563159)(394.20946484,98.77791634)(394.20946484,98.22563159)
\closepath
}
}
{
\newrgbcolor{curcolor}{0 0 0}
\pscustom[linestyle=none,fillstyle=solid,fillcolor=curcolor]
{
\newpath
\moveto(398.32299023,94.74504077)
\curveto(398.32299023,94.19275602)(397.87527498,93.74504077)(397.32299023,93.74504077)
\curveto(396.77070548,93.74504077)(396.32299023,94.19275602)(396.32299023,94.74504077)
\curveto(396.32299023,95.29732552)(396.77070548,95.74504077)(397.32299023,95.74504077)
\curveto(397.87527498,95.74504077)(398.32299023,95.29732552)(398.32299023,94.74504077)
\closepath
}
}
{
\newrgbcolor{curcolor}{0 0 0}
\pscustom[linewidth=0,linecolor=curcolor]
{
\newpath
\moveto(398.32299023,94.74504077)
\curveto(398.32299023,94.19275602)(397.87527498,93.74504077)(397.32299023,93.74504077)
\curveto(396.77070548,93.74504077)(396.32299023,94.19275602)(396.32299023,94.74504077)
\curveto(396.32299023,95.29732552)(396.77070548,95.74504077)(397.32299023,95.74504077)
\curveto(397.87527498,95.74504077)(398.32299023,95.29732552)(398.32299023,94.74504077)
\closepath
}
}
{
\newrgbcolor{curcolor}{0 0 0}
\pscustom[linestyle=none,fillstyle=solid,fillcolor=curcolor]
{
\newpath
\moveto(389.1467207,96.01066577)
\curveto(389.1467207,95.45838102)(388.69900545,95.01066577)(388.1467207,95.01066577)
\curveto(387.59443595,95.01066577)(387.1467207,95.45838102)(387.1467207,96.01066577)
\curveto(387.1467207,96.56295052)(387.59443595,97.01066577)(388.1467207,97.01066577)
\curveto(388.69900545,97.01066577)(389.1467207,96.56295052)(389.1467207,96.01066577)
\closepath
}
}
{
\newrgbcolor{curcolor}{0 0 0}
\pscustom[linewidth=0,linecolor=curcolor]
{
\newpath
\moveto(389.1467207,96.01066577)
\curveto(389.1467207,95.45838102)(388.69900545,95.01066577)(388.1467207,95.01066577)
\curveto(387.59443595,95.01066577)(387.1467207,95.45838102)(387.1467207,96.01066577)
\curveto(387.1467207,96.56295052)(387.59443595,97.01066577)(388.1467207,97.01066577)
\curveto(388.69900545,97.01066577)(389.1467207,96.56295052)(389.1467207,96.01066577)
\closepath
}
}
{
\newrgbcolor{curcolor}{0 0 0}
\pscustom[linestyle=none,fillstyle=solid,fillcolor=curcolor]
{
\newpath
\moveto(393.26024609,94.1121062)
\curveto(393.26024609,93.55982145)(392.81253084,93.1121062)(392.26024609,93.1121062)
\curveto(391.70796134,93.1121062)(391.26024609,93.55982145)(391.26024609,94.1121062)
\curveto(391.26024609,94.66439095)(391.70796134,95.1121062)(392.26024609,95.1121062)
\curveto(392.81253084,95.1121062)(393.26024609,94.66439095)(393.26024609,94.1121062)
\closepath
}
}
{
\newrgbcolor{curcolor}{0 0 0}
\pscustom[linewidth=0,linecolor=curcolor]
{
\newpath
\moveto(393.26024609,94.1121062)
\curveto(393.26024609,93.55982145)(392.81253084,93.1121062)(392.26024609,93.1121062)
\curveto(391.70796134,93.1121062)(391.26024609,93.55982145)(391.26024609,94.1121062)
\curveto(391.26024609,94.66439095)(391.70796134,95.1121062)(392.26024609,95.1121062)
\curveto(392.81253084,95.1121062)(393.26024609,94.66439095)(393.26024609,94.1121062)
\closepath
}
}
{
\newrgbcolor{curcolor}{0 0 0}
\pscustom[linestyle=none,fillstyle=solid,fillcolor=curcolor]
{
\newpath
\moveto(397.69017773,97.59269702)
\curveto(397.69017773,97.04041227)(397.24246248,96.59269702)(396.69017773,96.59269702)
\curveto(396.13789298,96.59269702)(395.69017773,97.04041227)(395.69017773,97.59269702)
\curveto(395.69017773,98.14498177)(396.13789298,98.59269702)(396.69017773,98.59269702)
\curveto(397.24246248,98.59269702)(397.69017773,98.14498177)(397.69017773,97.59269702)
\closepath
}
}
{
\newrgbcolor{curcolor}{0 0 0}
\pscustom[linewidth=0,linecolor=curcolor]
{
\newpath
\moveto(397.69017773,97.59269702)
\curveto(397.69017773,97.04041227)(397.24246248,96.59269702)(396.69017773,96.59269702)
\curveto(396.13789298,96.59269702)(395.69017773,97.04041227)(395.69017773,97.59269702)
\curveto(395.69017773,98.14498177)(396.13789298,98.59269702)(396.69017773,98.59269702)
\curveto(397.24246248,98.59269702)(397.69017773,98.14498177)(397.69017773,97.59269702)
\closepath
}
}
{
\newrgbcolor{curcolor}{0 0 0}
\pscustom[linewidth=1,linecolor=curcolor]
{
\newpath
\moveto(335.2711,69.47869976)
\lineto(344.7964,86.24319976)
\lineto(348.9875,65.28749976)
\lineto(354.7027,87.00519976)
\lineto(366.1331,70.62169976)
\lineto(371.4672,87.76719976)
}
\rput[bl](470,100){$\F$}
\rput[bl](375,17){$M:= \supp \tilde{\rho}$}
}
\end{pspicture}
\caption{A quantum spacetime.}
\label{figquant}
\end{figure}%
Spacetime including small-scale fluctuations have
 been studied in~\cite{fockfermionic,
fockentangle, fockdynamics} for the description of bosonic quantum fields
in Minkowski space (see also the recent survey~\cite{qftlimit}).
Moreover, in~\cite{lqg} the general geometric framework 
for quantum spacetimes was developed.
Whether and how the generalized
Einstein equations for such quantum spacetimes
will be related to common approaches to quantum gravity 
(see for example~\cite{rovelli-qg, thiemann,loll, surya})
is a challenging open problem.

\appendix
\section{The $\L$-Induced Weingarten Map} \label{apph}
in this section, we briefly explain how the Lagrangian induces a Weingarten map, and how it is related to the Riemannian curvature.
\begin{Def} \label{defortho} A tangent vector~$\nu \in T_p\F$ at a spacetime point~$p \in \tilde{M}$
is said to be {\bf{$\L$-normal}} if
\beq \label{ortho}
\fint_{M_p} D_{1,\nu} \L(p, y)\: y\: d\rho_p(y) = 0 \qquad \text{for all~$p \in \tilde{M}$}\:.
\eeq
\end{Def}
If~$\nu$ is tangential to~$M_p$, the integral in~\eqref{ortho}
can be written as~$D_\nu \phi_p(\tilde{x})|_{\tilde{x}=p}$.
Using~\eqref{delta} in Lemma~\ref{lemmadphi}, one sees that the only
$\L$-normal vector of~$M_p$ is zero. We thus obtain the direct sum decomposition
\beq \label{Tortho}
T_p \F = M_p \oplus N_p \:,
\eeq
where~$N_p$ denotes all normal vectors.

In the $\L$-induced chart centered at~$q$ and using the osculation maps
(as introduced in Section~\ref{secosc}), the normality condition~\eqref{ortho}
can be written as
\beq \label{lorthochart}
\fint_M D_{1,\nu} \L\big( \tilde{F}(x), F_p(y) \big)\: y^k\: f_p\: d\rho(y) \Big|_{x=p} = 0 \:.
\eeq
We now introduce the Weingarten map by differentiating with respect
to the base point of the osculation. In order to get a well-defined operation, we need
to compose with the parallel transport~$\nabla^\L_{q,p} : M_p \rightarrow
M_q$ (for example along a minimal geodesic).
\begin{Def} Let~$\nu \in T_p \F$ be $\L$-normal. We then define the {\bf{Weingarten map}}
\begin{align*}
&W^{(\nu)}_p : M_p \rightarrow M_p \qquad \text{by} \\
&W^{(\nu)}_p(u) := \Big( u\: \frac{\partial}{\partial z} \Big)
\fint_{M_z} D_{1,\nu} \L(p, y)\: \big( \nabla^\L_{p,z} y \big)\: d\rho_z(y) \Big|_{z=p} \:.
\end{align*}
\end{Def} \noindent
This definition simplifies considerably in the $\L$-induced chart centered at~$q$
to
\[ (W^{(\nu)}_p)^k_j =
\fint_M \frac{\partial}{\partial p^j} \Big( D_{1,\nu} \L\big( \tilde{F}(x), F_p(y) \big)\: f_p \Big)
\: y^k\: d\rho(y) \Big|_{x=p} \]
(note that the parallel transport~$\nabla^\L_{p,z}$ can be left out because
its derivative term vanishes in view of~\eqref{lorthochart}).
The Weingarten map describes the extrinsic curvature of~$\tilde{M}$
in~$\F$. It gives a tensorial description of the curvature as
described above by local expansions
of the osculation maps in Lemma~\ref{lemmalocexpand}.

Now let~$\nu \in \Gamma(\tilde{M}, T\F)$ be an $\L$-normal vector field
on~$\tilde{M}$ (in the sense that~\eqref{lorthochart} holds for all~$p \in \tilde{M}$).
Then, differentiating the relation~\eqref{lorthochart} with respect to both~$x$ and~$p$, we obtain
\[ (W^{(\nu)}_p)^k_j = -\frac{\partial}{\partial x^j}
\fint_M D_{1,\nu} \L\big( \tilde{F}(x), F_p(y) \big)\: y^k\: f_p\: d\rho(y) \Big|_{x=p} 
= -\partial_j \big(\nu \phi^k_p(x) \big)\big|_{x=p} \:. \]
Comparing with~\eqref{nablaLdef}, we can write the last expression
as an $\L$-covariant derivative,
\[ (W^{(\nu)}_p)^k_j = -(\nabla^\L_j \nu)^k \:. \]
This formula resembles the Weingarten map for
surfaces in~$\R^n$ (see for example~\cite[Definition~4.17]{lee-manifold}).
However, there are also major differences.
We first point out that the derivative~$\nabla^\L \nu$ is not the derivative in
Euclidean space (nor the covariant derivative in an ambient Riemannian manifold),
but it is merely an extension of the covariant derivative~$\nabla^\L$ to normal vectors. This covariant derivative is a tangent vector, which means that the usual ``projection to the tangent space'' is already included in our formula for~$\nabla^\L \nu$.
Next, one should keep in mind that direct sum decomposition~\ref{Tortho}
does not come from a scalar product on~$T_p\F$.
Indeed, we have a scalar product only on~$M_p$, but not on~$T_p\F$.
Consequently, the mapping
\[ T_p \F \rightarrow M_p\:, \qquad v \mapsto
\fint_{M_p} D_{1,v} \L(p, y)\: y\: d\rho_p(y) \]
is idempotent, but it is not an orthogonal projection operator.

In view of these differences, it seems impossible to relate the Riemannian curvature
to the Weingarten map. In particular, it does not seem possible to formulate analogs of the Gau{\ss} or Codazzi-Mainardi equations.
We note that these extensions become possible if we specialize
to the setting of causal fermion systems and endow~$\F$ with the
Riemannian metric induced by the
Hilbert-Schmidt scalar product (for details see~\cite[Section~4]{gaugefix}
and~\cite[Section~3.4]{banach}). We shall not enter these constructions here, 
also because they do not seem to be helpful for the formulation of the Einstein equations.

\section{The Riemannian and Lorentzian Metrics of the Regularized Dirac Sea Vacuum} \label{appregvac}
In Section~\ref{secregvec} we derived a Lorentzian metric~$\eta$ from the regularizing vector field 
$u$ in~\eqref{uint} and the Riemannian metric $g$ defined in Section~\ref{secg}
(see Definition~\ref{defflip}).
A-priori, the definition of the flip metric involves a free parameter~$\tau>1$
(see~\eqref{flipgen}). This free parameter can be fixed by the requirement
that the causal structure of the Lorentzian metric coincides with that of the
causal fermion system~$(\H, \F, \tilde{\rho})$. Since our argument is local,
it suffices to do the computations for the causal fermion system~$(\H, \F, \rho)$
describing the Minkowski vacuum. More specifically, we consider the
regularized Dirac sea vacuum in four-di\-men\-sio\-nal Minkowski space~$M=\R^{1,3}$.
Here we do not need to enter the detailed construction
(as given in~\cite{oppio} and the textbooks~\cite[Section~5.5]{intro} or~\cite[Section~1.2]{cfs}). Instead, it suffices to state a few properties of the resulting causal Lagrangian.
As already stated in Section~\ref{secosccfs}, the Lagrangian is
translation invariant~\eqref{translation} and reflection symmetric~\eqref{reflection}.
Moreover, choosing the reference frame where~$e_0$ points into the direction of the
regularization, the Lagrangian is spherically symmetric, i.e.\
\beq \label{Lspherical}
\L(x,y) = \L\big[\xi^0, |\vec{\xi}| \big]
\eeq
(where again~$\xi := y-x$). Next, the Lagrangian has its main contribution
on the light cone, meaning that (up to errors which we disregard; for details
see~\cite{reg} and~\cite[Appendix~A]{jacobson}),
\beq \label{Lcausal}
\L\big[\xi^0, |\vec{\xi}| \big] = 0 \qquad \text{unless~$|\xi^0|=|\vec{\xi}|$} \:.
\eeq

\begin{Lemma} Under the assumptions~\eqref{Lspherical} and~\eqref{Lcausal},
the causal structures of the Lorentzian flip metric~$\eta$ in~\eqref{flipgen}
coincides with that of Minkowski space if and only if~$\tau=4$.
  \end{Lemma}
  \begin{proof}
    We begin by computing $g_p$ and use the chart expression from equation~\eqref{gchart}. Since 
    $M = \R^{1,3}$, our parametrization $\tilde{F}$ and osculation maps $F_p$ are 
    trivial everywhere. Thus, we end up with
  \[
      g^{jk}_p = \dfrac{1}{\delta^2} \fint_M \L(p,y) (y-p)^j (y-p)^k \dif{\rho}(y) \, ,
  \]
  where $\L$ is the causal Lagrangian of the regularized Minkowski space. 
  Due to spherical symmetry~\eqref{Lspherical}, we can use Schur's lemma for the 
  spatial components to obtain
  \[
   g^*_p = \begin{bmatrix}
      \alpha & \0_{1 \times 3} \\
      \0_{3 \times 1} & \beta \,\1_{3 \times 3}
    \end{bmatrix} \, ,
  \]
  with $\alpha, \beta>0$. More precisely, these numbers are defined by the 
  integrals 
  \[
    \alpha = \dfrac{1}{\delta^2}\fint_M \L(x,y)\:|y^0|^2 \:\dif{\rho}(y) \quad \text{and} 
    \quad \beta = \dfrac{1}{3\delta^2}\fint_M \L(x,y)\: |\vec{y}|^2 \:\dif{\rho}(y) \, ,
  \]
  where we can drop $p$ because of translation invariance. 
  Since~$\L$ is supported on the light cone~\eqref{Lcausal},
  it follows that~$|y^0|^2 = |\vec{y}|^2$, and thus~$\alpha = 3 \beta$.
 
 Computing the regularization vector field as introduced in Definition~\ref{defreg},
 we obtain, again using spherical symmetry and Schur's lemma that the spatial
 part vanishes. Normalizing with respect to the above Riemannian metric,
 we obtain
 \[ \omega = \frac{e_0}{\sqrt{\alpha}} \qquad \text{and} \qquad
\hat{u} = \sqrt{\alpha}\: e^0 \:. \]
 
 A direct computation gives
 \begin{align*}
g_p &= \begin{bmatrix}
      \alpha^{-1} & \0_{1 \times 3} \\
      \0_{3 \times 1} & 3 \alpha^{-1} \,\1_{3 \times 3} \end{bmatrix} \:,\qquad
\omega \otimes \omega = 
\begin{bmatrix}
      \alpha^{-1} & \0_{1 \times 3} \\
      \0_{3 \times 1} & \0_{3 \times 3} \:.
\end{bmatrix}
\end{align*}
Substituting into the formula for the general flip metric~\eqref{flipgen},
we see that~$\eta$ is a multiple of the Minkowski metric if and only if~$\tau=4$.
\end{proof}

\section{Construction of Almost-Optimal Osculations} \label{appoptimal}
In this appendix, it is shown how one can satisfy the equations for an optimal osculation~\eqref{osceq} approximately, up to error terms which are ``small'' in a sense
to be quantified below. We work in the setting of causal fermion systems.
The {\em{wave evaluation operator}}~$\Psi(x)$ of the vacuum spacetime~$(\H, \F, \rho)$
is defined by
\[ \Psi(x) = \pi_x \::\: \H \rightarrow S_xM \:, \]
where~$\pi_x$ is the orthogonal projection to the {\em{spin space}}~$S_xM := x(\H) \subset \H$
(for more details on the basic definitions see~\cite[Chapter~1]{cfs} or~\cite[Section~5.7]{intro}).
We consider the mappings
\beq \label{IJrel}
\Psi(\0)^*, \partial_j \Psi(\0)^* \::\: S_\0M \rightarrow \H \:.
\eeq
and denote their images by
\[ I := \Psi(\0)^*(S_\0M)\:,\qquad J := \text{span} \{ \partial_j \Psi(\0)^* \:|\: j=0,\ldots, 3 \} \:. \]
Moreover, we set~$K:= \text{span}(I \cup J)$. The space~$I$ is four-dimensional.
Keeping in mind that the Dirac equation holds, which we can write as
\beq \label{DirPsi}
i \partial_j \Psi(\0)^* \gamma^j + m \Psi(\0)^* = 0 \:,
\eeq
the space~$K$ is $16$-dimensional.
In the interacting spacetime~$(\tilde{\H}, \tilde{\F}, \tilde{\rho})$ we consider similarly the
mappings
\[ \tilde{\Psi}(p)^*, \partial_j \tilde{\Psi}(p)^* \::\: S_p\tilde{M} \rightarrow \tilde{\H} \]
and introduce the subspace~$\tilde{I}, \tilde{K} \subset \tilde{\H}$.
Choosing Gaussian normal coordinates and a normal spinor frame, the Dirac equation
at~$p$ takes the same form as in Minkowski space, so that~\eqref{IJrel} holds similarly
at~$p$, giving the same algebraic relations on~$K$ and~$\tilde{K}$.
We identify the spin spaces~$S_\0M$ and~$S_p\tilde{M}$ in these spinor frames and denote them simply by~$S$.

The naive idea for getting an optimal osculation is to choose a linear mapping~$V : K \rightarrow
\tilde{K}$ such that
\[ V \Psi(\0)^* = \tilde{\Psi}(p)^* \qquad \text{and} \qquad
V \,\partial_j \Psi(\0)^* = \partial_j \tilde{\Psi}(p)^* \quad \text{for~$i=0,\ldots,3$} \:. \]
Here we can leave out the index~$i=0$ because of the linear dependence~\eqref{DirPsi}.
Then the existence of~$V$ is obvious, because simply map the
corresponding column vectors to each other.
Suppose for the that~$V$ can be extended
unitarily to a mapping~$\scrU : \H \rightarrow \tilde{\H}$. Then using this mapping
to define~$\rho_p$ via~\eqref{rhop}, we would get an optimal osculation,
meaning that~\eqref{transprop} and~\eqref{osceq} hold.

The basic difficulty is that the linear mapping~$V : K \rightarrow \tilde{K}$ does in general {\em{not}} admit a unitary extension~$\scrU : \H
\rightarrow \tilde{\H}$. In order to see the obstructions, we consider the Gram matrices of the vectors in~$K$ and~$\tilde{K}$,
\beq \label{gram}
\partial_\kappa \Psi(0)\: \partial_{\kappa'} \Psi(0)^* : S \rightarrow S
\qquad \text{and} \qquad
\partial_\kappa \tilde{\Psi}(p) \:\partial_{\kappa'} \tilde{\Psi}(p)^* : S \rightarrow S \:.
\eeq
Here, for a compact notation, we introduced the indices~$\kappa$
and~$\kappa'$ which can take the values~$\nind$ (corresponding to no derivative) or~$1,2,3$ (corresponding to the three spatial derivatives).
In this formulation, we can say that a unitary extension~$\scrU$ exists if
and only if the Gram matrices coincide.
This will of course in general not be the case.
For example, perturbing by a single wave function~$\psi$
localized in a spatial region of volume~$\ell^3$, the Gram matrices are
perturbed by a relative error~$E$ with the scalings
\beq \label{scalings}
E= \frac{\varepsilon^{3}}{\ell^3} \qquad \text{or even} \qquad
E= \frac{m \varepsilon^4}{\ell^3}
\eeq
(where~$\varepsilon$ denotes the regularization length;
the second term is relevant if only the spatial
derivatives of~$\Psi$ are perturbed, as is the case in the Dirac
energy-momentum tensor, where first derivatives of the wave functions
come into play).

The question is how to account for the error term in Gram matrices
of the form~\eqref{scalings}. Here we can makes use of the
concept of {\em{wave functions of separated supports}} as introduced in~\cite[Section~3]{current}. The idea is to choose a subspace~$\tilde{L} \subset \tilde{\H}$ of Dirac wave functions which are located at a very large distance of the spacetime point~$p$ (for details on scalings and error terms
we refer to~\cite[Section~3]{current}; here for simplicity we leave
out these error terms). For the operator~$\scrU$
on~$K$ we make the ansatz
\beq \label{lambdaV}
\scrU|_K = \lambda V + \scrU^\perp : K \rightarrow \tilde{\H}
\qquad \text{with} \qquad
\lambda \in \R \text{ and } \scrU^\perp : K \rightarrow \tilde{L} \:.
\eeq
The choice of~$\scrU^\perp$ has no influence on~$\tilde{\ell}$
and its derivatives~\eqref{osceqgen}, simply because the
wave functions in~$\tilde{L}$ are supported far away from~$p$.
The remaining question is whether~$\lambda$ and~$\scrU^\perp$ can be
chosen in such a way that the mapping~\eqref{lambdaV} is
an isometric embedding. Considering the Gram matrix of~$\scrU|_{K}$,
it is the Gram matrix of~$\lambda V$ plus the Gram matrix of~$\scrU^\perp$.
Since the latter Gram matrix can be chosen to be an arbitrary
non-negative matrix, we conclude that~$\scrU|_K$ can be arranged
to be isometric if and only if~$\lambda^2$ times the Gram matrix
on the right side of~\eqref{gram} is smaller or equal than the
Gram matrix on the left side of~\eqref{gram}.
This can be arranged by choosing~$\lambda$ slightly smaller than one,
with~$1-\lambda$ scaling as the error terms in~\eqref{scalings}.

Our findings can be summarized as follows.
\begin{Prp} \label{prpnonoptimal}
Perturbing the wave evaluation operator of
the Dirac sea vacuum in Minkowski space by an error term with relative scaling~$E$ (for example as in~\eqref{scalings}),
there is a unitary operator~$\scrU$
such that the transformed vacuum~$M_p$ has the following properties,
\begin{align}
\big\|p - \scrU \0 \scrU^{-1} \big\| &\lesssim E\: \|p\| \label{o1} \\
\tilde{\ell} \big( \scrU \0 \scrU^{-1} \big) &\lesssim E\:\s \label{o2} \\
D^2 \tilde{\ell}|_{M_p} \big( \scrU \0 \scrU^{-1} \big) &= 0 \label{o3}
\end{align}
(where~$\| \cdot \|$ denotes the sup-norm on~$\Lin(\tilde{\H})$).
\end{Prp}
\Proof The estimate~\eqref{o1}
follows immediately from the fact that~$\|\scrU|_K - V\| \lesssim E$.
For the analysis of~$\tilde{\ell}$, the contributions involving second derivatives of the
local correlation operators can be left out, because they can be
rewritten as a first order variation with~$\delta \Psi(x) = \partial_{jk} \Psi(x)$,
which vanishes due to the EL equations.
Moreover,, we may disregard the operator~$\scrU^\perp$.
Then the local correlation operators and their derivatives are mapped
to each other up to the factor~$\lambda$, i.e.\
\[ \scrU \Psi(\0)^* = \lambda \tilde{\Psi}(p)^* \qquad \text{and} \qquad
\scrU \,\partial_j \Psi(\0)^* = \lambda \,\partial_j \tilde{\Psi}(p)^* \quad \text{for~$i=0,\ldots,3$} \:. \]
Hence also~$\tilde{\ell}$ can be computed
simply by inserting a corresponding scaling factor,
\[ \tilde{\ell}\big( \scrU \0 \scrU^{-1} \big) = \lambda^2\, \tilde{\ell}(p) = \lambda^2 \s \:. \]
This proves~\eqref{o2}. Differentiating twice tangential to~$M_p$ gives~\eqref{o3}.
\QED
We finally explain and discuss this result. Clearly, with~\eqref{o3}
we have realized the condition for an optimal osculation in~\eqref{osceq}.
However, the constraints~\eqref{transprop} or~\eqref{notransprop}
are in general violated. Instead, with~\eqref{o1} we only arranged that~$p$
and~$\scrU \0 \scrU^{-1}$ are close together.
Moreover, with~\eqref{o2} we made sure that the function~$\tilde{\ell}$
is close to its minimal value at~$\scrU \0 \scrU^{-1}$.
This fact is very helpful because it gives control of the 
resulting errors in the Einstein equations.
More precisely, we expect that the approximate osculations as in
Proposition~\ref{prpnonoptimal} are suitable for describing the geometry
of~$\tilde{M}$, including the contributions to the Einstein equations
of order~$\delta^2$ (i.e., contributions which scale like
the usual energy-momentum tensor).
But, in order to derive the correction terms of
order~$\delta^3$ and higher systematically, one needs to study
the variational principle for the osculation and the resulting
non-optimal osculations, as will be outlined in the next appendix.

\section{Non-Optimal Osculations and Torsion} \label{appnonoptimal}
In this appendix, we discuss the effects of a non-optimal osculation and show how
this introduces torsion to the $\L$-induced connection.
As discussed in Section~\ref{seccorrect} and outlined in Appendix~\ref{appoptimal},
we expect that resulting corrections to the Einstein equations
are of higher order in the Planck length;  this is why we do not consider them
in the main part of this paper.

Let $u, v \in \Gamma(\tilde{M}, T\tilde{M})$ be two vector fields. In general, the torsion is 
defined as
\begin{align}\label{deftorsion}
  T^\L(u,v) := \nabla^\L_u v - \nabla^\L_v u - [u,v] \, ,
\end{align}
Recall that, for an optimal osculation, we could identify~$T_p\tilde{M}$ with $M_p$ for every $p \in \tilde{M}$. In this case, torsion vanishes (see~\eqref{Lconn}).
For non-optimal osculations, however, the tangent space~$T_p\tilde{M}$
and the osculating vacuum~$M_p$ are different objects (see Figure~\ref{figosculate}).
But we can identify them using linearization of the mapping~$\phi_p$ defined in~\eqref{phiint},
\[ \Lambda_p := D\phi_p|_p : T_p \tilde{M} \rightarrow M_p \:. \]
It is sensible to assume that this mapping is invertible.
\begin{Def}\label{almostosc}
The {\bf{osculation is regular}} at~$p \in \tilde{M}$ if the mapping~$\Lambda_p$
is invertible.
\end{Def} \noindent
If the osculation is regular at~$p$, the mapping~$\phi_p$ is a local diffeomorphism,
defining again an $\L$-induced chart of~$\tilde{M}$.

We want to extend the definition of the covariant derivative~\eqref{nablaLdef}
to the setting of a regular osculation. To this end, we again consider vector field~$u \in
\Gamma(\tilde{M}, T\tilde{M})$ and a tangent vector~$v \in T_p \tilde{M}$.
For any~$\tilde{x} \in \tilde{M}$ we introduce the vector
\[ (\Lambda u)(\tilde{x}) := \Lambda_{\tilde{x}} u(\tilde{x}) \in M_{\tilde{x}} \:. \]
Again identifying~$M_{\tilde{x}}$ with its tangent space, we can regard
the vector~$\Lambda u$ as a tangent vector in~$T_{\tilde{x}} \tilde{M}$.
But one must keep in mind that this vector need not be tangential to~$\tilde{M}$,
i.e.\ 
\[ (\Lambda u)(\tilde{x}) \in T_{\tilde{x}} \F \qquad \text{but in
general} \quad (\Lambda u)(\tilde{x}) \not \in T_{\tilde{x}} \tilde{M}\:. \]
As a consequence, it is impossible to proceed in analogy to~\eqref{nablaLdef} and
differentiate~$\phi_p$ in the direction of~$\Lambda u$.
Instead, we must first extend~$\phi_p$ to all of~$\F$. Denoting the extension by
\[ \Phi_p : \F \rightarrow M_p \:,\qquad \Phi_p(\tilde{x}) := 
\fint_{M_p} \L(\tilde{x},y)\: y\: d\rho_p(y) \:, \]
we can take the corresponding directional derivative of and then differentiate in the direction~$v$,
\[ D_v \big( D_{\Lambda u} \Phi_p(\tilde{x}) \big) \big|_{\tilde{x}=p} \in M_p \]
(we note for clarity that all derivatives act on the variable~$\tilde{x}$,
whereas~$p$ is fixed).
In order to get back to~$T_p \tilde{M}$, we apply the mapping~$\Lambda_p^{-1}$.
We thus define the {\em{$\L$-induced connection}} by
\[ \nabla^\L_v u \big|_p := 
\Lambda_p^{-1} \Big(  D_v \big( D_{\Lambda u} \Phi_p(\tilde{x}) \big) 
\Big) \Big|_{\tilde{x}=p} \in
T_p \tilde{M} \:. \]
This connection has torsion, as one can understand directly from the fact that
the $v$-derivative also acts on~$\Lambda(\tilde{x})$.

In order to see the role of the mapping~$\Lambda(\tilde{x})$ in more detail,
we now work out the connection in the formalism with the osculation maps
as introduced in Section~\ref{secosc}.
We begin with an $\L$-induced chart $(\phi, V)$ centered at $q$ and a local
parametrization~$\tilde{F} = \phi^{-1} : U := \phi(V) \rightarrow V$.
At each $p \in U$, we denote the osculating vacuum at $\tilde{F}(p) = F_p(p)$ by $M_p$, where $F_p$ is the affine linear osculation map from Lemma~\ref{lemmaFp}.
The coordinates give rise to distinguished bases of the tangent spaces, denoted
as usual by
\[ \frac{\partial}{\partial x^j} \Big|_{\tilde{x}} \in T_{\tilde{x}} \tilde{M} \:. \]
Moreover, the osculation maps gives rise to bases of the osculating vacua, denoted by
\[ e_j(\tilde{x}) := F_{\tilde{x}}(e_j)  \in M_{\tilde{x}} \]
(where~$e_j$ denotes the standard basis of~$M=M_q$ and~$F_{\tilde{x}}$ is the
affine linear osculation map introduced in Lemma~\ref{lemmaFp}).
Then the mapping~$\Lambda$ can be written in components as
\[ \Lambda_{\tilde{x}} \:\frac{\partial}{\partial x^j} \Big|_{\tilde{x}} 
= \Lambda_j^k(\tilde{x})\:  e_k(\tilde{x}) \:. \]
We thus obtain
\begin{align*}
D_{\Lambda u} \Phi_p(\tilde{x}) &= \Lambda_j^k(\tilde{x})\: u^j(\tilde{x})\: D_{e_k} \Phi_p(\tilde{x}) \in M_p \\
D_v \big( D_{\Lambda u} \Phi_p(\tilde{x}) \big) &=
v^l \frac{\partial}{\partial \tilde{x}^l} \Big(  \Lambda_j^k(\tilde{x})\: u^j(\tilde{x})\: D_{e_k} \Phi_p(\tilde{x})\Big) \in M_p \\
&= v^l \:\big( \partial_l u^j)\:\Big(  \Lambda_j^k(\tilde{x})\: D_{e_k} \Phi_p(\tilde{x})\Big) \\
&\quad\: + v^l u^j\: \frac{\partial}{\partial \tilde{x}^l} \Big(  \Lambda_j^k(\tilde{x})\: D_{e_k} \Phi_p(\tilde{x})\Big) \:.
\end{align*}
Evaluating at~$\tilde{x}=p$, we can use that
\[ D_{e_k} \Phi^m_p(\tilde{x}) \big|_{\tilde{x}=p} = \delta^m_k \]
(this follows exactly as in the proof of Lemma~\ref{lemmadphi}, because
the derivative is tangential to~$M_p$). We conclude that
\[ \nabla^\L_v u^i \big|_p = v^l \:(\partial_l u^i)
+ v^l u^j\: \underbrace{\frac{\partial}{\partial \tilde{x}^l} \Big(  \Lambda_j^k(\tilde{x})\: D_{e_k} \Phi^a_p(\tilde{x})\Big)\: \big(\Lambda_p^{-1} \big)_a^i}_{\displaystyle = \Gamma^i_{lj}}
\Big|_{\tilde{x}=p}
\]
Anti-symmetrizing the Christoffel symbols gives the torsion tensor~\eqref{deftorsion},
\begin{align*}
(T^\L_p)^i_{lj} &= \big(\Lambda_p^{-1} \big)_a^i \bigg(
\frac{\partial}{\partial \tilde{x}^j} \Big(  (\Lambda_p)_l^k\: D_{e_k} \Phi^a_p(\tilde{x})\Big)
-\frac{\partial}{\partial \tilde{x}^l} \Big(  (\Lambda_p)_j^k\: D_{e_k} \Phi^a_p(\tilde{x})\Big)\bigg) \bigg|_{\tilde{x}=p}\:.
\end{align*}

We finally remark that one can construct a torsion-free connection from $\nabla^\L$
using the
{\em{contorsion tensor}}~$K^\L$ defined by (see for example~\cite[eq.~(7.35)]{nakahara})
\[ K^\L(u,v)\,w = \dfrac{1}{2}\:\Big(T^\L(u,v)\, w + T^\L(v,w)\, u - T^\L(w,u) \,v\Big) \:. \]
The contorsion tensor is 
anti-symmetric in its arguments~$u$ and~$w$, and anti-symme\-tri\-zing in
the arguments~$u$ and~$v$ gives back torsion,
  \[
    K^\L(u,v)\,w = - K^\L(w,v)\,u \quad \text{and} \quad K^\L(u,v)\,w - K^\L(v,u)\,w = T^\L(u,v)\,w \, .
  \]
Now we can define the connection with $v \in T_p\tilde{M}$ and $u \in \Gamma(\tilde{M}, T\tilde{M})$ by
\[ \nabla_v u(p) := \nabla^\L_v u(p) + K^\L(v,u) \:. \]
This connection is indeed torsion-free, because
\[
  T(v,u) := \nabla_v u - \nabla_u v - [v,u] = T^\L(v,u) + K^\L(v,u) - K^\L(u,v) = 0 \, .
\]
Therefore, even for general regular osculations, one can construct a torsion-free 
connection that can be used to define the Einstein equations. 
All the corrections from torsion can then again
be included in the energy-momentum tensor, 
similar as done in Theorems~\ref{thmrein} and~\ref{thmlein}
for the deviation tensor and the regularization vector field.

\Thanks{{{\em{Acknowledgments:}}
We would like to thank Marco van den Beld Serrano,
Patrick Fi\-scher, Niky Kamran and Claudio Paganini
for helpful discussions. We are grateful to the referee for valuable feedback.
C.K.\ gratefully acknowledges support by the Heinrich-B\"oll-Stiftung.

\bibliographystyle{amsplain}

\begin{thebibliography}{10}

\bibitem{cfsweblink}
\emph{Link to web platform on causal fermion systems:
  \href{https://www.causal-fermion-system.com}{\textrm{www.causal-fermion-system.com}}}.

\bibitem{loll}
J.~Ambjorn, A.~G\"{o}rlich, J.~Jurkiewicz, and R.~Loll, \emph{Causal dynamical
  triangulations and the search for a theory of quantum gravity},
  \href{https://arxiv.org/abs/1305.6680}{arXiv:1305.6680 [gr-qc]}, Internat. J.
  Modern Phys. D \textbf{22} (2013), no.~9, 1330019, 18.

\bibitem{beem}
J.K. Beem, P.E. Ehrlich, and K.L. Easley, \emph{Global {L}orentzian
  {G}eometry}, second ed., Monographs and Textbooks in Pure and Applied
  Mathematics, vol. 202, Marcel Dekker, Inc., New York, 1996.

\bibitem{jacobson}
E.~Curiel, F.~Finster, and J.M. Isidro, \emph{Two-dimensional area and matter
  flux in the theory of causal fermion systems},
  \href{https://arxiv.org/abs/1910.06161}{arXiv:1910.06161} [math-ph],
  Internat. J. Modern Phys. D \textbf{29} (2020), 2050098.

\bibitem{qftlimit}
C.~Dappiaggi, F.~Finster, N.~Kamran, and M.~Reintjes, \emph{The quantum field
  theory limit of causal fermion systems}, in preparation.

\bibitem{fockdynamics}
\bysame, \emph{Holographic mixing and {F}ock space dynamics of causal fermion
  systems}, \href{https://arxiv.org/abs/2410.18045}{arXiv:2410.18045
  [math-ph]}, Ann. Henri Poincar{\'e} \textbf{27} (2026), no.~5, 1885--1969.

\bibitem{reg}
F.~Finster, \emph{On the regularized fermionic projector of the vacuum},
  \href{https://arxiv.org/abs/math-ph/0612003}{arXiv:math-ph/0612003}, J. Math.
  Phys. \textbf{49} (2008), no.~3, 032304, 60.

\bibitem{continuum}
\bysame, \emph{Causal variational principles on measure spaces},
  \href{https://arxiv.org/abs/0811.2666}{arXiv:0811.2666 [math-ph]}, J. Reine
  Angew. Math. \textbf{646} (2010), 141--194.

\bibitem{cfs}
\bysame, \emph{The {C}ontinuum {L}imit of {C}ausal {F}ermion {S}ystems},
  \href{https://arxiv.org/abs/1605.04742}{arXiv:1605.04742 [math-ph]},
  Fundamental Theories of Physics, vol. 186, Springer, Cham, 2016.

\bibitem{positive}
\bysame, \emph{Positive functionals induced by minimizers of causal variational
  principles}, \href{https://arxiv.org/abs/1708.07817}{arXiv:1708.07817
  [math-ph]}, Vietnam J. Math. \textbf{47} (2019), 23--37.

\bibitem{current}
F.~Finster and P.~Fischer, \emph{A canonical construction of the extended
  {H}ilbert space for causal fermion systems},
  \href{https://arxiv.org/abs/2504.18276}{arXiv:2504.18276 [math-ph]}, J. Math.
  Phys. \textbf{67} (2026), 032302.

\bibitem{cfs-curved}
\bysame, \emph{The continuum limit analysis of causal fermion systems for
  curved spacetimes}, \href{https://arxiv.org/abs/2605.30199}{arXiv:2605.30199
  [math-ph]} (2026).

\bibitem{lqg}
F.~Finster and A.~Grotz, \emph{A {L}orentzian quantum geometry},
  \href{https://arxiv.org/abs/1107.2026}{arXiv:1107.2026 [math-ph]}, Adv.
  Theor. Math. Phys. \textbf{16} (2012), no.~4, 1197--1290.

\bibitem{mmt-cfs}
F.~Finster, E.~Guendelman, and C.F. Paganini, \emph{Modified measures as an
  effective theory for causal fermion systems},
  \href{https://arxiv.org/abs/2303.16566}{arXiv:2303.16566 [gr-qc]}, Class.
  Quant. Gravity \textbf{41} (2024), no.~3, 035007, 25.

\bibitem{cosmo}
F.~Finster and J.M. Isidro, \emph{A mechanism for dark matter and dark energy
  in the theory of causal fermion systems},
  \href{https://arxiv.org/abs/2209.02234}{arXiv:2209.02234}, Class. Quant.
  Gravity \textbf{40} (2023), no.~1, 075017, 24.

\bibitem{review}
F.~Finster and M.~Jokel, \emph{Causal fermion systems: {A}n elementary
  introduction to physical ideas and mathematical concepts},
  \href{https://arxiv.org/abs/1908.08451}{arXiv:1908.08451 [math-ph]},
  {P}rogress and {V}isions in {Q}uantum {T}heory in {V}iew of {G}ravity
  (F.~Finster, D.~Giulini, J.~Kleiner, and J.~Tolksdorf, eds.), Birkh\"auser
  Verlag, Basel, 2020, pp.~63--92.

\bibitem{baryogenesis}
F.~Finster, M.~Jokel, and C.F. Paganini, \emph{A mechanism of baryogenesis for
  causal fermion systems},
  \href{https://arxiv.org/abs/2111.05556}{arXiv:2111.05556 [gr-qc]}, Class.
  Quant. Gravity \textbf{39} (2022), no.~16, 165005, 50.

\bibitem{topology}
F.~Finster and N.~Kamran, \emph{Spinors on singular spaces and the topology of
  causal fermion systems},
  \href{https://arxiv.org/abs/1403.7885}{arXiv:1403.7885 [math-ph]}, Mem. Amer.
  Math. Soc. \textbf{259} (2019), no.~1251, v+83 pp.

\bibitem{fockfermionic}
\bysame, \emph{Fermionic {F}ock spaces and quantum states for causal fermion
  systems}, \href{https://arxiv.org/abs/2101.10793}{arXiv:2101.10793
  [math-ph]}, Ann. Henri Poincar\'{e} \textbf{23} (2022), no.~4, 1359--1398.

\bibitem{matter}
\bysame, \emph{A positive quasilocal mass for causal variational principles},
  \href{https://arxiv.org/abs/2310.07544}{arXiv:2310.07544 [math-ph]}, Calc.
  Var. \textbf{64} (2025), no.~3, 91.

\bibitem{dirac}
F.~Finster, N.~Kamran, and M.~Oppio, \emph{The linear dynamics of wave
  functions in causal fermion systems},
  \href{https://arxiv.org/abs/2101.08673}{arXiv:2101.08673 [math-ph]}, J.
  Differential Equations \textbf{293} (2021), 115--187.

\bibitem{fockentangle}
F.~Finster, N.~Kamran, and M.~Reintjes, \emph{Entangled quantum states of
  causal fermion systems and unitary group integrals},
  \href{https://arxiv.org/abs/2207.13157}{arXiv:2207.13157 [math-ph]}, Adv.
  Theor. Math. Phys. \textbf{27} (2023), no.~5, 1463--1589.

\bibitem{lcalc}
F.~Finster, N.~Kamran, and F.~van~der Top, \emph{The {$\L$}-calculus for causal
  variational principles: {A}n exterior differential calculus on non-smooth
  spaces}, in preparation.

\bibitem{gaugefix}
F.~Finster and S.~Kindermann, \emph{A gauge fixing procedure for causal fermion
  systems}, \href{https://arxiv.org/abs/1908.08445}{arXiv:1908.08445
  [math-ph]}, J. Math. Phys. \textbf{61} (2020), no.~8, 082301.

\bibitem{intro}
F.~Finster, S.~Kindermann, and J.-H. Treude, \emph{{C}ausal {F}ermion
  {S}ystems: {A}n {I}ntroduction to {F}undamental {S}tructures, {M}ethods and
  {A}pplications}, \href{https://arxiv.org/abs/2411.06450}{arXiv:2411.06450
  [math-ph]}, Cambridge Monographs on Mathematical Physics, Cambridge
  University Press, 2025.

\bibitem{jet}
F.~Finster and J.~Kleiner, \emph{A {H}amiltonian formulation of causal
  variational principles},
  \href{https://arxiv.org/abs/1612.07192}{arXiv:1612.07192 [math-ph]}, Calc.
  Var. Partial Differential Equations \textbf{56:73} (2017), no.~3, 33.

\bibitem{noncompact}
F.~Finster and C.~Langer, \emph{Causal variational principles in the
  $\sigma$-locally compact setting: {E}xistence of minimizers},
  \href{https://arxiv.org/abs/2002.04412}{arXiv:2002.04412 [math-ph]}, Adv.
  Calc. Var. \textbf{15} (2022), no.~3, 551--575.

\bibitem{banach}
F.~Finster and M.~Lottner, \emph{Banach manifold structure and
  infinite-dimensional analysis for causal fermion systems},
  \href{https://arxiv.org/abs/2101.11908}{arXiv:2101.11908 [math-ph]}, Ann.
  Global Anal. Geom. \textbf{60} (2021), no.~2, 313--354.

\bibitem{pmt}
F.~Finster and A.~Platzer, \emph{A positive mass theorem for static causal
  fermion systems}, \href{https://arxiv.org/abs/1912.12995}{arXiv:1912.12995
  [math-ph]}, Adv. Theor. Math. Phys. \textbf{25} (2021), no.~7, 1735--1818.

\bibitem{baryoconform}
F.~Finster and M.~van~den Beld-Serrano, \emph{Baryogenesis in conformally flat
  spacetimes}, \href{https://arxiv.org/abs/2504.17434}{arXiv:2504.17434
  [math-ph]} (2025).

\bibitem{baryomink}
\bysame, \emph{Baryogenesis in {M}inkowski spacetime},
  \href{https://arxiv.org/abs/2408.01189}{arXiv:2408.01189 [math-ph]}, J. Geom.
  Phys. \textbf{207} (2025), no.~16, 105346, 29.

\bibitem{hawking+ellis}
S.W. Hawking and G.F.R. Ellis, \emph{The {L}arge {S}cale {S}tructure of
  {S}pace-{T}ime}, Cambridge University Press, London, 1973.

\bibitem{laugwitz}
D.~Laugwitz, \emph{Differential and {R}iemannian {G}eometry}, Academic Press,
  New York-London, 1965, Translated by F. Steinhardt.

\bibitem{lee-manifold}
J.M. Lee, \emph{Manifolds and {D}ifferential {G}eometry}, Graduate Studies in
  Mathematics, vol. 107, American Mathematical Society, Providence, RI, 2009.

\bibitem{lichnerowiczintro}
A.~Lichnerowicz, \emph{Elements of {T}ensor {C}alculus}, Methuen \& Co Ltd,
  John Wiley \& Sons, London-New York, 1962, Translated by J.W. Leech.

\bibitem{misner}
C.W. Misner, K.S. Thorne, and J.A. Wheeler, \emph{Gravitation}, W.H. Freeman
  and Co., San Francisco, Calif., 1973.

\bibitem{nakahara}
M.~Nakahara, \emph{Geometry, {T}opology and {P}hysics}, second ed., Graduate
  Student Series in Physics, Institute of Physics, Bristol, 2003.

\bibitem{olea}
B.~Olea, \emph{Canonical variation of a {L}orentzian metric},
  \href{https://arxiv.org/abs/1509.00793}{arXiv:1509.00793 [math.DG]}, J. Math.
  Anal. Appl. \textbf{419} (2014), no.~1, 156--171.

\bibitem{oneillsemi}
B.~O'Neill, \emph{Semi-{R}iemannian {G}eometry}, Pure and Applied Mathematics,
  vol. 103, Academic Press, Inc. [Harcourt Brace Jovanovich, Publishers], New
  York, 1983.

\bibitem{oppio}
M.~Oppio, \emph{On the mathematical foundations of causal fermion systems in
  {M}inkowski space}, \href{https://arxiv.org/abs/1909.09229}{arXiv:1909.09229
  [math-ph]}, Ann. Henri Poincar\'e \textbf{22} (2021), no.~3, 873--949.

\bibitem{lagrange-hoelder}
\bysame, \emph{H\"older continuity of the integrated causal {L}agrangian in
  {M}inkowski space}, \href{https://arxiv.org/abs/2109.04728}{arXiv:2109.04728
  [math-ph]}, Adv. Theor. Math. Phys. \textbf{26} (2022), no.~9, 3249--3318.

\bibitem{paganini+yadav}
C.F. Paganini and S.~Yadav, \emph{The preserver problem for causal fermion
  systems}, in preparation.

\bibitem{reddy}
V.V. Reddy, R.~Sharma, and S.~Sivaramakrishnan, \emph{Lorentzian metric induced
  from a background {R}iemannian metric}, Int. J. Pure Appl. Math. \textbf{47}
  (2008), no.~3, 343--351.

\bibitem{rovelli-qg}
C.~Rovelli, \emph{Quantum {G}ravity}, Cambridge Monographs on Mathematical
  Physics, Cambridge University Press, Cambridge, 2004.

\bibitem{straumann}
N.~Straumann, \emph{General {R}elativity}, Texts and Monographs in Physics,
  Springer-Verlag, Berlin, 2004.

\bibitem{surya}
S.~Surya, \emph{The causal set approach to quantum gravity},
  \href{https://arxiv.org/abs/1903.11544}{arXiv:1903.11544 [gr-qc]}, Living
  Rev. Relativ. \textbf{22} (2019), no.~5, 75pp.

\bibitem{thiemann}
T.~Thiemann, \emph{Modern {C}anonical {Q}uantum {G}eneral {R}elativity},
  Cambridge Monographs on Mathematical Physics, Cambridge University Press,
  Cambridge, 2007.

\bibitem{wald}
R.M. Wald, \emph{General {R}elativity}, University of Chicago Press, Chicago,
  IL, 1984.

\end{thebibliography}
\providecommand{\bysame}{\leavevmode\hbox to3em{\hrulefill}\thinspace}
\providecommand{\MR}{\relax\ifhmode\unskip\space\fi MR }
\providecommand{\MRhref}[2]{%
  \href{http://www.ams.org/mathscinet-getitem?mr=#1}{#2}
}
\providecommand{\href}[2]{#2}

\end{document}